%% file: sicomp-xia.tex
\documentclass{article}
\usepackage[paper=letterpaper,margin=1in]{geometry}

\usepackage{graphicx,yhmath}
\usepackage{pstricks,pstricks-add}
\usepackage{pst-node}
\usepackage{pst-coil}
\usepackage{amsthm}
\usepackage{amsmath}
\usepackage{amssymb}
\usepackage{amsfonts}
\usepackage{epsfig}
\usepackage{verbatim}
\usepackage{enumitem}
\usepackage{arcs}
\newrgbcolor{green}{0 .8 0}
\newrgbcolor{red}{.9 0 0}
\newtheorem{definition}{Definition}

\newtheorem{theorem}{Theorem}
\newtheorem{proposition}{Proposition}
\newtheorem{lemma}{Lemma}
\newtheorem{corollary}{Corollary}

\newlength{\alginputwidth}
\newlength{\algboxwidth}

\newcommand{\algtitle}[1]{\underline{Algorithm \ {\bf #1}} \vspace*{1mm}\\}

\newsavebox{\algbox}
\newsavebox{\captionbox}
\newenvironment{algorithm}[2]%
    {
        \setlength{\algboxwidth}{\columnwidth}
        \addtolength{\algboxwidth}{-\columnsep}
        \addtolength{\algboxwidth}{-1mm}
        \setlength{\alginputwidth}{\algboxwidth}
        \addtolength{\alginputwidth}{-1.7cm}
        \begin{figure}[htbp]
            \vspace*{-1mm}
            \centering
            \begin{lrbox}{\captionbox}
                \begin{minipage}[b]{\algboxwidth}
                    \centering
                    \label{#2}
                \end{minipage}
            \end{lrbox}
            \begin{lrbox}{\algbox}
                \begin{minipage}[b]{\algboxwidth}
                    \footnotesize
                    \vspace*{2mm}
    } 
    {
                    \vspace*{0.2mm}
               \end{minipage}
            \end{lrbox}
            \fbox{\usebox{\algbox}\hspace*{1mm}}
            \usebox{\captionbox}
            \vspace*{-1mm}
        \end{figure}
    }
\newsavebox{\algcodebox}
\newenvironment{codeblock}%
    {
        \begin{enumerate}
            \setlength{\itemsep}{2pt}
            \setlength{\parsep}{0pt}
            \setlength{\topsep}{0pt}
            \setlength{\parskip}{0pt}
            \setlength{\partopsep}{0pt}
    } 
    {\end{enumerate}}
\newcommand{\step}{\item}

\begin{document}

\title{The Stretch Factor of the Delaunay Triangulation Is Less Than 1.998}

\author{\sc{Ge Xia}\thanks{
Department of Computer Science,
Lafayette College,
Easton, PA 18042, USA.
xiag{\tt @}lafayette.edu. The preliminary version of the paper appears in the {\em Proceedings of the 27th Annual Symposium on Computational Geometry} (SoCG 2011).}\\
}
\date{}
\maketitle

\begin{abstract}
Let $S$ be a finite set of points in the Euclidean plane. Let $D$ be a Delaunay triangulation of $S$. The {\em stretch factor}  (also known as {\em dilation} or {\em spanning ratio}) of $D$ is the maximum ratio, among all points $p$ and $q$ in $S$, of the shortest path distance from $p$ to $q$ in $D$ over the Euclidean distance $||pq||$. Proving a tight bound on the stretch factor of the Delaunay triangulation has been a long standing open problem in computational geometry.

In this paper we prove that the stretch factor of the Delaunay triangulation is less than $\rho = 1.998$, significantly improving the current best upper bound of 2.42 by Keil and Gutwin (1989). Our bound of 1.998 also improves the upper bound of the best stretch factor that can be achieved by a plane spanner of a Euclidean graph (the current best upper bound is 2). Our result has a direct impact on the problem of constructing spanners of Euclidean graphs, which has applications in the area of wireless computing.
\vspace*{5mm}
\end{abstract}




\section{Introduction} \label{intro} \setcounter{page}{1} \pagestyle{plain}
Let $S$ be a finite set of points in the Euclidean plane. A {\em Delaunay triangulation} of $S$ is a triangulation in which the circumscribed circle of every triangle contains no point of $S$ in its interior. An alternative equivalent definition is: An edge $xy$ is in the Delaunay triangulation of $S$ if and only if there is a circle through points $x$ and $y$ whose interior is devoid of points of $S$. A Delaunay triangulation of $S$ is the dual graph of the Voronoi diagram of $S$.

Let $D$ be a Delaunay triangulation of $S$. The {\em stretch factor}  (also known as {\em dilation} or {\em spanning ratio}) of $D$ is the maximum ratio, among all points $p$ and $q$ in $S$, of the shortest path distance from $p$ to $q$ in $D$ over the Euclidean distance $||pq||$.

Proving a tight bound on the stretch factor of the Delaunay triangulation has been a long standing open problem in computational geometry. Chew~\cite{chew} showed a lower bound of $\pi/2$ on the stretch
factor of the Delaunay triangulation. This lower bound of $\pi/2$ was widely believed to be tight until recently (2009) when Bose et al.~\cite{BoseCGTA} proved that the lower bound is at least $1.5846 > \pi/2$, which was further improved to 1.5932 by Xia and Zhang~\cite{fwcg2010}. Dobkin, Friedman, and Supowit~\cite{dobkin0,dobkin} in 1987 showed that the Delaunay triangulation has stretch factor at most $(1+\sqrt{5})\pi/2 \approx 5.08$. This upper bound was improved by Keil and Gutwin~\cite{keil0,keil} in 1989 to $2\pi/(3\cos{(\pi/6)}) \approx 2.42$. Despite considerable efforts since then, 2.42 currently stands as the best upper bound on the stretch factor of the Delaunay triangulation. For the special case when the point set $S$ is in convex position, Cui, Kanj and Xia~\cite{CGTA} recently proved that the Delaunay triangulation of $S$ has stretch factor at most 2.33.

In this paper we prove that the stretch factor of the Delaunay triangulation of a point set in the plane is less than $\rho = 1.998$, significantly improving the current best upper bound of 2.42 from 1989~\cite{keil0,keil}. Our bound of 1.998 is also better than the upper bound of 2.33 for the special case when the point set is in convex position~\cite{CGTA}.

Our approach in proving the upper bound on the stretch factor of the Delaunay triangulation is different from the previous approaches~\cite{dobkin0,dobkin,keil0,keil,CGTA}. Our approach focuses on the geometry of a ``chain'' of disks in the plane (the formal definition of the chain is given in Section~\ref{section:prelim}). The ``stretch factor'' of a chain can be defined in analogy to that of the Delaunay triangulation. We prove that the upper bound on the stretch factor of a chain is less than 1.998.
\begin{theorem}\label{thm:main}
The stretch factor of a chain of disks $\mathcal{O}$ in the plane is less than $\rho=1.998$.
\end{theorem}
As a special case of Theorem~\ref{thm:main}, we derive the same upper bound on the stretch factor of the Delaunay triangulation.
\begin{corollary}\label{cor:main}
The stretch factor of a Delaunay triangulation $D$ of a set of points $S$ in the plane is less than $\rho=1.998$.
\end{corollary}

Our result has a direct impact on the problem of constructing {\em spanners} of Euclidean graphs~\cite{Bose_onplane}, which has applications in the area of wireless computing (for more details, see~\cite{spannerbook}).
Many spanner constructions in the literature rely on extracting subgraphs of the Delaunay triangulation (see for example~\cite{bosealgorithmica,levco1,siam10,iitunbounded}) and their spanning ratio is expressed as a function of the stretch factor of the Delaunay triangulation. Hence the new upper bound of 1.998 on the stretch factor of the Delaunay triangulation automatically improves the upper bounds on the spanning ratio of all such spanners.

Another important consequence of our result is that it improves the upper bound of the best stretch factor achieved by plane spanners of the complete 2-D Euclidean graph. Previously, the plane spanner with the best known upper bound on the stretch factor is the Triangular Distance Delaunay triangulation by Chew~\cite{chew}, whose stretch factor is 2. Our result shows that the Delaunay triangulation has a smaller upper bound of 1.998 on the stretch factor.

We believe that future research based on our approach will yield further improvements in both upper bound and lower bound, and may eventually lead to the tight bound of the stretch factor of the Delaunay triangulation. At the end of this paper, we will discuss possible ways to improve our approach for a better upper bound. Following our approach, Xia and Zhang~\cite{fwcg2010} showed an improved lower bound of 1.5932 for the stretch factor of the Delaunay triangulation by giving a sequence of chains with increasing stretch factors. It was conjectured in~\cite{fwcg2010} that the tight bound occurs at the limit of the sequence.

The paper is organized as follows. The necessary definitions are given in Section~\ref{section:prelim}. The main theorem of the paper is presented in Section~\ref{sec:outline}. There we also discuss the proof strategy and provide an outline of the proof for the main theorem. Most of the technical details are captured by two lemmas, whose proofs appear in Section~\ref{section:tech} and Section~\ref{section:main}. The paper ends with a discussion on the possible improvements of our approach in Section~\ref{sec:conclu}.

\section{Preliminaries} \label{section:prelim}

We label the points in the plane by lower case letters, such as $p,q,u,v,$ etc. For any two points $p,q$ in the plane, denote by $pq$ the line in the plane passing through $p$ and $q$, by $\overline{pq}$ the line segment connecting $p$ and $q$, and by $\overrightarrow{pq}$ the ray from $p$ to $q$. The Euclidean distance between $p$ and $q$ is denoted by $||pq||$. The length of a path $P$ in the plane is denoted by $|P|$. Any angle denoted by $\angle poq$ is measured from $\protect\overrightarrow{op}$ to $\protect\overrightarrow{oq}$ in the {\em counterclockwise} direction. Unless otherwise specified, all angles in this paper are defined in the range $(-\pi,\pi]$.

\begin{definition}\label{def:circles}\rm
We say that a finite sequence of distinct disks\footnote{In this paper, a disk is considered to be a closed subset of the plane} $\mathcal{O}=(O_1, O_2, \ldots, O_n)$ in the plane is a {\em chain} if it has the following two properties. {\bf Property (1)}: Every two consecutive disks $O_i, O_{i+1}$ intersect\footnote{This includes the case where $O_i, O_{i+1}$ are tangent.}, $1\leq i\leq n-1$, but neither disk contains the other. Denote by $C_i^{(i-1)}$ and $C_i^{(i+1)}$ the arcs on the boundary of $O_i$ that is in $O_{i-1}$ and $O_{i+1}$, respectively. We refer to $C_i^{(i-1)}$ and $C_i^{(i+1)}$ as the ``connecting arcs'' of $O_i$. {\bf Property (2)}: The connecting arcs $C_{i}^{(i-1)}$ and $C_{i}^{(i+1)}$ of $O_i$ do not overlap, for $2\leq i\leq n-1$; however $C_{i}^{(i-1)}$ and $C_{i}^{(i+1)}$ can share an endpoint.
See Figure~\ref{fig:G} for an illustration. For any $1\leq i \leq j \leq n$, denote by $\mathcal{O}_{i,j}$ a sub-chain of $\mathcal{O}$: $\mathcal{O}_{i,j}=(O_i, \ldots, O_j)$.
\end{definition}

   \input{fig-v2-chain-only}

\begin{definition}\label{def:terminals}\rm
Given a chain $\mathcal{O}=(O_1, O_2, \ldots, O_n)$. Let $u$ be a point on the boundary of $O_1$ that is not in the interior of $O_2$. Let $v$ be a point on the boundary of $O_n$ that is not in the interior of $O_{n-1}$. We call $u,v$ {\em a pair of terminal points} (or simply {\em terminals}) of the chain $\mathcal{O}$. Let $o_1, \ldots, o_n$ be the centers of $O_1, \ldots, O_n$, respectively. We call the polyline $uo_1 \ldots o_nv$ the {\em centered polyline} between $u$ and $v$. For $1\leq i\leq n-1$, let $a_i$ and $b_i$ be the intersections of the boundaries of $O_i$ and $O_{i+1}$ (in the special case where $O_i, O_{i+1}$ are tangent, $a_i=b_i$). Without loss of generality, assume all $a_i$'s are on one side of the centered polyline $uo_1 \ldots o_nv$ and all $b_i$'s are on the other side\footnote{The two sides of $uo_1 \ldots o_nv$ can be distinguished by the left- and right-hand side as we move along $uo_1 \ldots o_nv$.} (if $a_i=b_i$ then both of them are on the centered polyline). For notational convenience, define $a_0=b_0=u$ and $a_n=b_n=v$. Every disk $O_i$ has two arcs on its boundary between the line segments $\overline{a_{i-1}b_{i-1}}$ and $\overline{a_ib_i}$, denoted by $A_i$ and $B_i$. Without loss of generality, assume that $a_{i-1},a_i$ are the ends of $A_i$ and $b_{i-1},b_i$ are the ends of $B_i$, for $1\leq i \leq n$. This means that $A_1\ldots A_n$ is a path from $u$ to $v$ on one side of the chain and $B_1\ldots B_n$ is a  path from $u$ to $v$ on the other side of the chain. An arc $A_i$ or $B_i$ may degenerate to a point, in which case $a_{i-1}=a_i$ or $b_{i-1}=b_i$, respectively. Refer to Figure~\ref{fig:polyline} for an illustration.
\end{definition}

   \input{fig-v2-chain-terminal}

\begin{definition}\label{def:polyline}\rm
Let $D_{\mathcal{O}}(u,v)=up_1\ldots p_{n-1}v$ be the shortest polyline from $u$ to $v$ that consists of line segments $\overline{up_1}$, $\overline{p_1p_2}$, $\ldots$, $\overline{p_{n-1}v}$, where $p_i \in \overline{a_ib_i}$ for $1\leq i \leq n-1$. In other words, the polyline $D_{\mathcal{O}}(u,v)$ is the shortest polyline (which can be visualized as a ``rubber band'') from $u$ to $v$ that intersects line segments $\overline{a_1b_1},\ldots,\overline{a_{n-1}b_{n-1}}$ in that order. See Figure~\ref{fig:polyline} for an illustration. The length of $D_{\mathcal{O}}(u,v)$, denoted by $|D_{\mathcal{O}}(u,v)|$, is the sum of the lengths of the line segments in $D_{\mathcal{O}}(u,v)$. If the polyline $D_{\mathcal{O}}(u,v)$ contains a point $p_j$ which is $a_j$ or $b_j$, for some $1\leq j \leq n-1$, we say that $u,v$ are {\em obstructed}. If $u,v$ are unobstructed, then $D_{\mathcal{O}}(u,v)$ is the straight line segment $\overline{uv}$. Note that the converse of this statement is not true, because even when $D_{\mathcal{O}}(u,v)$ is the straight line segment $\overline{uv}$, $D_{\mathcal{O}}(u,v)$ may still contain a point $p_j \in \{a_j,b_j\}$. See Figure~\ref{fig:polyline} for an illustrations of obstructed and unobstructed cases\footnote{Note that for the purpose of bounding the stretch factor of the Delaunay triangulation, only the case where $u,v$ are unobstructed is relevant. However, for our proof to work, it is necessary to consider the case when $u,v$ are obstructed (see Section~\ref{section:tech}).}. \end{definition}

We have the following simple proposition.
\begin{proposition}\label{fact}
If $D_{\mathcal{O}}(u,v)$ is the straight line segment $\overline{uv}$, then $\overrightarrow{uv}$ stabs $O_1,\ldots,O_n$ in order. That is, for any $1\leq i<j \leq n$, $\overrightarrow{uv}$ enters $O_i$ no later than entering $O_j$ and exits $O_i$ no later than exiting $O_j$. \end{proposition}
\begin{proof}
For $2\leq i\leq n-1$, the two connecting arcs on $O_i$ do not overlap. Hence $\overline{a_{i-1}b_{i-1}}$ and $\overline{a_{i+1}b_{i+1}}$ must appear on the different sides of the line $a_ib_i$. Refer to Figure~\ref{fig:stab}. Without loss of generality, assume that $a_ib_i$ is a vertical line which divides the plane into two half-planes so that $\overline{a_{i-1}b_{i-1}}$ appears in the left half-plane and $\overline{a_{i+1}b_{i+1}}$ appears in the right half-plane. If $D_{\mathcal{O}}(u,v)$ is the straight line segment $\overline{uv}$, then $u,p_1,p_2,\ldots,p_{n-1},v$ are co-linear and they appear in that order in the ray $\overrightarrow{uv}$. This means that $\overrightarrow{uv}$ crosses $\overline{a_ib_i}$ from left to right. Thus $\overrightarrow{uv}$ enters $O_{i-1}$ and $O_i$ in the left half-plane. The line $a_ib_i$ divides $O_{i-1}$ and $O_i$ each into the left-portion and the right-portion (w.r.t. $a_ib_i$). The left-portion of $O_i$ is contained in the left-portion of $O_{i-1}$, and hence $\overrightarrow{uv}$ must enter $O_{i-1}$ no later than entering $O_i$. Likewise, $\overrightarrow{uv}$ exits $O_{i-1}$ and $O_i$ in the right half-plane. Since the right-portion of $O_{i-1}$ is contained in the right-portion of $O_i$, $\overrightarrow{uv}$ must exit $O_{i-1}$ no later than exiting $O_i$. This completes the proof.
\end{proof}

\input{fig-v3-stab}

\begin{definition}\label{def:path}\rm
We define the shortest path between $u$ and $v$ in $\mathcal{O}$, denoted by $P_{\mathcal{O}}(u,v)$, to be the shortest path from $u$ to $v$ that consists of arcs in $\{A_1, \ldots, A_n\}\cup\{B_1,\ldots,B_n\}$ and line segments in $\{\overline{a_1b_1},\ldots, \overline{a_{n-1}b_{n-1}}\}$. A more formal definition of $P_{\mathcal{O}}(u,v)$ is given below. Consider a weighted graph representation of $\mathcal{O}$, denoted by $\mathbb{G}_{\mathcal{O}}$, whose vertex set is $\{u=a_0=b_0\}\cup\{a_1,\ldots,a_{n-1}\}\cup\{b_1,\ldots,b_{n-1}\}\cup \{v=a_n=b_n\}$ and whose edge set is $\{(a_{i-1},a_i) ~|~ 1\leq i \leq n\}\cup
\{(b_{i-1},b_i) ~|~ 1\leq i \leq n\}\cup
\{(a_i,b_i) ~|~ 1\leq i \leq n-1\}$. See Figure~\ref{fig:polyline}(c) for an illustration. There is a clear bijection between the edge set of $\mathbb{G}_{\mathcal{O}}$ and the set of the arcs and line segments $\{A_1, \ldots, A_n\}\cup\{B_1,\ldots,B_n\}\cup\{\overline{a_1b_1},\ldots, \overline{a_{n-1}b_{n-1}}\}$  in $\mathcal{O}$. The weight of any edge in $\mathbb{G}_{\mathcal{O}}$ is the length of the corresponding arc or line segment in $\mathcal{O}$, that is, $w(a_{i-1},a_i)=|A_i|$, $w(b_{i-1},b_i)=|B_i|$, and $w(a_i,b_i) = ||a_ib_i||$. If an arc or line segment is degenerated, then the weight is 0. Let $P_{\mathbb{G}_{\mathcal{O}}}(u,v)$ be the shortest path between $u$ and $v$ in $\mathbb{G}_{\mathcal{O}}$. Then $P_{\mathcal{O}}(u,v)$ is defined to be the path in $\mathcal{O}$ that corresponds to $P_{\mathbb{G}_{\mathcal{O}}}(u,v)$. The length of $P_{\mathcal{O}}(u,v)$, denoted $|P_{\mathcal{O}}(u,v)|$, is the total weight of the edges in $P_{\mathbb{G}_{\mathcal{O}}}(u,v)$. Refer to Figure~\ref{fig:polyline} for an illustration. Now we can define the {\em stretch factor} of a chain $\mathcal{O}$, denoted by $C_{\mathcal{O}}$, as the maximum value of $$
|P_{\mathcal{O}}(u,v)|/|D_{\mathcal{O}}(u,v)|,$$ over all terminals $u,v$ of $\mathcal{O}$.
\end{definition}


In this paper, we will prove that 1.998 is an upper bound on the stretch factor of the chain.

\section{An Overview of the Proof of Theorem~\ref{thm:main}} \label{sec:outline}

Due to the complex nature of the proof of Theorem~\ref{thm:main}, in this section we discuss the proof strategy and present an outline of the proof. The main technical details of the proof are captured in two lemmas, whose proofs are given in the subsequent sections.

When $\mathcal{O}$ has only one disk, it is easy to see that for all $\mathcal{O}, u$, and $v$, $|P_{\mathcal{O}}(u,v)|/|D_{\mathcal{O}}(u,v)| \leq \pi/2 < \rho$. So it is natural to attempt an inductive proof based on the number of disks in $\mathcal{O}$.
A simple induction would require us to show that adding a disk to a chain will not increase the stretch factor. However, this is not true because one can always increase the stretch factor of a chain by adding a disk to it, albeit by a very small amount~\cite{fwcg2010}.
We tackle this problem by amortized analysis. Specifically, we introduce a potential function $\Phi_{\mathcal{O}}$ (to be determined later) and define a target function $\Upsilon_{\mathcal{O}}(u,v)$:
\begin{align}
\Upsilon_{\mathcal{O}}(u,v) = |P_{\mathcal{O}}(u,v)| - \lambda|D_{\mathcal{O}}(u,v)|+\Phi_{\mathcal{O}},\label{defupsilon}\end{align} where $\lambda=1.8$ is a parameter whose value is determined by the potential function. Then we will try to prove that $\Upsilon_{\mathcal{O}}(u,v) < 0$ for all $\mathcal{O}, u$, and $v$. This is a sufficient condition for Theorem~\ref{thm:main}, as we will show later in this section.

The key component of the amortized analysis is the selection of an appropriate potential function $\Phi_{\mathcal{O}}$, which is described in the following.

   \input{fig-v2-phi}

\begin{definition}\rm\label{peak}
Let $O_{i-1}$ and $O_i$ be two consecutive disks in a chain $\mathcal{O}$. Without loss of generality, assume that their centers $o_{i-1},o_i$ lie on a horizontal line and that $a_{i-1}$ is on or above the line $o_{i-1}o_i$. See Figure~\ref{fig:phi}. Let $q_{i-1}^{\rightarrow}$ be the ``peak'' of $O_{i-1}$ with regard to $o_{i-1}o_i$, i.e., the point on the upper boundary of $O_{i-1}$ that is the farthest from the line $o_{i-1}o_i$. Likewise, let $q_i^{\leftarrow}$ be the ``peak'' of $O_i$ with regard to $o_{i-1}o_i$ (the sign $^{\rightarrow}$ or $^{\leftarrow}$ indicates whether the peak is defined with the preceding disk or the succeeding disk in $\mathcal{O}$). Let $Q_{i-1}^{\rightarrow}$ be the upper arc between $q_{i-1}^{\rightarrow}$ and $a_{i-1}$ on the boundary of $O_{i-1}$ and let $Q_i^{\leftarrow}$ be the upper arc between $q_i^{\leftarrow}$ and $a_{i-1}$ on the boundary of $O_i$. If $Q_{i-1}^{\rightarrow}$ is inside $O_i$, we say that it is ``heavy'' (colored red\footnote{Red appears as dark-gray in black and white print.} in Figure~\ref{fig:phi}); otherwise we say that $Q_{i-1}^{\rightarrow}$ is ``light'' (colored green\footnote{Green appears as light-gray in black and white print.} in Figure~\ref{fig:phi}). Likewise, we say that $Q_i^{\leftarrow}$ is ``heavy''  or ``light'' depending on whether it is inside $O_{i-1}$. Let $\mathcal{P}_i$ be a path from $q_{i-1}^{\rightarrow}$ to $q_i^{\leftarrow}$ consisting of  $Q_{i-1}^{\rightarrow}$ and $Q_i^{\leftarrow}$. Let $H_i$ be the horizontal distance traveled along the path $\mathcal{P}_i$ with light arcs contributing positively to $H_i$ and heavy arcs contributing negatively to $H_i$. Similarly,  let $V_i$ be the vertical distance traveled along the path $\mathcal{P}_i$ with light arcs contributing positively to $V_i$ and heavy arcs contributing negatively to $V_i$. The potential function is defined as follows: \begin{align}
\Phi_{\mathcal{O}} = \varphi(r_n-r_1)- \frac{\varphi}{3}\sum_{i=2}^n(2H_i+V_i),\label{phi_def}
\end{align}
where $r_i$ is the radius of $O_i$ and $\varphi = \frac{3}{\sqrt{5}}(1-\lambda/\rho)$ is a parameter that determines the ``weight'' of the potential function. When $n=1$, i.e., when $\mathcal{O}$ has only one disk, $\Phi_{\mathcal{O}} = \varphi(r_1-r_1)- \frac{\varphi}{3}\sum_{i=2}^1(2H_i+V_i) = 0.$
\end{definition}

The potential function is constructed with three goals in mind, given  as three lemmas (Lemma~\ref{notincrease}, Lemma~\ref{lemma:main}, and Lemma~\ref{lemma:pot}) in the following. Lemma~\ref{notincrease} is easy to prove, but Lemma~\ref{lemma:main} and Lemma~\ref{lemma:pot} are quite technical and they will be proved in subsequent sections.

First goal, the potential function $\Phi_{\mathcal{O}}$ is constructed such that adding a disk to $\mathcal{O}$ does not increase $\Phi_{\mathcal{O}}$, as shown below.
\begin{lemma}\label{notincrease}
$\Phi_{\mathcal{O}} \leq \Phi_{\mathcal{O}_{1,n-1}}$.
\end{lemma}
\begin{proof}
We have the following observations, whose proofs are given in the appendix (Section~\ref{sec:appendix-hivi}):
\begin{align}
H_i &= ||o_io_{i-1}||, \label{hivi1}\\
V_i &\geq |r_i-r_{i-1}|.\label{hivi2}
\end{align}

By the triangle inequality, $||o_no_{n-1}|| \geq ||o_na_{n-1}||-||o_{n-1}a_{n-1}||=r_n-r_{n-1}$. Combining this with (\ref{hivi1}) and (\ref{hivi2}), we have \begin{align}
2H_n+V_n&\geq 2||o_no_{n-1}||+|r_n-r_{n-1}|\nonumber\\
&\geq 2(r_n-r_{n-1})+(r_n-r_{n-1})\nonumber\\
&=3(r_n-r_{n-1}).\end{align}
Therefore
\begin{align*}
\Phi_{\mathcal{O}}-\Phi_{\mathcal{O}_{1,n-1}} &= [\varphi(r_n-r_1)- \frac{\varphi}{3}\sum_{i=2}^n(2H_i+V_i)] - [\varphi(r_{n-1}-r_1)- \frac{\varphi}{3}\sum_{i=2}^{n-1}(2H_i+V_i)]\\
&= \varphi(r_n-r_{n-1})-\frac{\varphi}{3}(2H_n+V_n) \\
&\leq \varphi(r_n-r_{n-1})-\frac{\varphi}{3}(3(r_n-r_{n-1})) = 0.\end{align*}
\end{proof}

Second goal, the potential function $\Phi_{\mathcal{O}}$ is constructed such that $\Upsilon_{\mathcal{O}}(u,v) = |P_{\mathcal{O}}(u,v)| - \lambda|D_{\mathcal{O}}(u,v)|+\Phi_{\mathcal{O}} < 0$ for all $\mathcal{O}, u$, and $v$, as shown below.

\begin{lemma}[Proof in Section~\ref{section:tech}]\label{lemma:main}
For all $\mathcal{O},u$, and $v$, $\Upsilon_{\mathcal{O}}(u,v) < 0$.
\end{lemma}

Third goal, the potential function $\Phi_{\mathcal{O}}$ is constructed such that its value can be bounded from below as a function of $|P_{\mathcal{O}}(u,v)|$ for some chain $\mathcal{O}$ with certain extremal properties, as shown below.

\begin{lemma}[Proof in Section~\ref{section:main}]\label{lemma:pot}
Let $\mathbb{O}$ be a set of chains whose stretch factor is greater than or equal to a threshold $\tau$. If $\mathbb{O}$ is non-empty, then there exists a chain $\mathcal{O}^* \in \mathbb{O}$ with terminals $u,v$ such that $|P_{\mathcal{O}^*}(u,v)|/|D_{\mathcal{O}^*}(u,v)| \geq \tau$ and
$\Phi_{\mathcal{O}^*} \geq -\frac{\sqrt{5}\varphi}{3}|P_{\mathcal{O}^*}(u,v)|$.
\end{lemma}

Assuming Lemma~\ref{lemma:main} and Lemma~\ref{lemma:pot} are true, we can prove the main theorem.

\begin{proof}[Proof for Theorem~\ref{thm:main}]
We will prove that for all $\mathcal{O}$, the stretch factor $C_{\mathcal{O}}$ is less than $\rho=1.998$.
For the sake of contradiction, suppose that there is a non-empty set $\mathbb{O}$ of chains $\mathcal{O}$ with stretch factor $C_{\mathcal{O}} \geq \rho$. By Lemma~\ref{lemma:pot}, there exists a chain $\mathcal{O}^* \in \mathbb{O}$ with terminals $u$ and $v$ such that \begin{align}
|P_{\mathcal{O}^*}(u,v)|/|D_{\mathcal{O}^*}(u,v)| \geq \rho, \label{lerho}
\end{align} and
\begin{align}
\Phi_{\mathcal{O}^*} \geq -\frac{\sqrt{5}\varphi}{3}|P_{\mathcal{O}^*}(u,v)|.\label{phibound}
\end{align}
By Lemma~\ref{lemma:main},
\begin{align}
\Upsilon_{\mathcal{O}^*}(u,v) = |P_{\mathcal{O}^*}(u,v)| - \lambda|D_{\mathcal{O}^*}(u,v)|+\Phi_{\mathcal{O}^*}
 < 0.\label{upsilon_bound}
\end{align}
Combining (\ref{phibound}) and (\ref{upsilon_bound}), we have
\begin{align}
&|P_{\mathcal{O}^*}(u,v)| - \lambda|D_{\mathcal{O}^*}(u,v)|-\frac{\sqrt{5}\varphi}{3}|P_{\mathcal{O}^*}(u,v)|\nonumber\\
\leq & |P_{\mathcal{O}^*}(u,v)| - \lambda|D_{\mathcal{O}^*}(u,v)|+\Phi_{\mathcal{O}^*}\nonumber\\
= & \Upsilon_{\mathcal{O}^*}(u,v)\nonumber\\
< & 0.\label{last_ineq}
\end{align}
Recall that $\varphi = \frac{3}{\sqrt{5}}(1-\lambda/\rho)$. We have \begin{align}
&|P_{\mathcal{O}^*}(u,v)| - \lambda|D_{\mathcal{O}^*}(u,v)|-\frac{\sqrt{5}\varphi}{3}|P_{\mathcal{O}^*}(u,v)|\nonumber\\
=& (1-\frac{\sqrt{5}\varphi}{3})|P_{\mathcal{O}^*}(u,v)| - \lambda|D_{\mathcal{O}^*}(u,v)|\nonumber\\
=& (1-(1-\lambda/\rho))|P_{\mathcal{O}^*}(u,v)| - \lambda|D_{\mathcal{O}^*}(u,v)|\nonumber\\
=& \frac{\lambda}{\rho}|P_{\mathcal{O}^*}(u,v)| - \lambda|D_{\mathcal{O}^*}(u,v)|\nonumber\\
=& \lambda\left(\frac{|P_{\mathcal{O}^*}(u,v)|}{\rho} - |D_{\mathcal{O}^*}(u,v)|\right).\label{last_ineq_new}
\end{align}

From (\ref{last_ineq}) and (\ref{last_ineq_new}), we have $\lambda\left(\frac{|P_{\mathcal{O}^*}(u,v)|}{\rho} - |D_{\mathcal{O}^*}(u,v)|\right) < 0.$ In other words, $\frac{|P_{\mathcal{O}^*}(u,v)|}{|D_{\mathcal{O}^*}(u,v)|} < \rho$. This is a contradiction to (\ref{lerho}). Therefore $\mathbb{O}$ must be empty and hence $C_{\mathcal{O}} < \rho$ for all $\mathcal{O}$.
\end{proof}

As a a special case of Theorem~\ref{thm:main}, we can derive an improved upper bound on the stretch factor of the Delaunay triangulation.

\begin{proof}[Proof for Corollary~\ref{cor:main}]
We will prove that the stretch factor of a Delaunay triangulation $D$ of a set of points $S$ in the plane is less than $\rho=1.998$.
For any two points $x,y \in S$, let $\mathcal{T}$ be the sequence of triangles in $D$ crossed by the ray $\overrightarrow{xy}$. Let $\mathcal{O}$ be the corresponding sequence of circumscribed circles of the triangles in $\mathcal{T}$. It is clear that $\mathcal{O}$ is a chain and $x,y$ are terminals of $\mathcal{O}$. The shortest path distance from $x$ to $y$ is at most $|P_{\mathcal{O}}(x,y)|$. Since $\overrightarrow{xy}$ stabs through all circles in $\mathcal{O}$, $x,y$ are unobstructed and hence $|D_{\mathcal{O}}(x,y)|= ||xy||$. By Theorem~\ref{thm:main}, the stretch factor of a Delaunay triangulation is at most $|P_{\mathcal{O}}(x,y)|/||xy|| =|P_{\mathcal{O}}(x,y)|/|D_{\mathcal{O}}(x,y)| < \rho$.
\end{proof}

The rest of the paper contains the proofs for Lemma~\ref{lemma:main} and Lemma~\ref{lemma:pot}.

\section{Proof of Lemma~\ref{lemma:main}}\label{section:tech}
In this section, we will prove that for all $\mathcal{O},u$, and $v$, $\Upsilon_{\mathcal{O}}(u,v) < 0$.

Recall that $\Upsilon_{\mathcal{O}}(u,v) = |P_{\mathcal{O}}(u,v)| - \lambda|D_{\mathcal{O}}(u,v)|+\Phi_{\mathcal{O}}$. Proceed by induction on $n$, the number of disks in $\mathcal{O}$. If $n=1$, then $\mathcal{O}$ has a single disk $O_1$. In this case,  $\Phi_{\mathcal{O}} = 0$, $P_{\mathcal{O}}(u,v)$ is the shorter arc between $u$ and $v$ on the boundary of $O_1$, and $D_{\mathcal{O}}(u,v)$ is the chord between $u$ and $v$ inside $O_1$. So $
|P_{\mathcal{O}}(u,v)|\leq \pi/2|D_{\mathcal{O}}(u,v)| < 1.8|D_{\mathcal{O}}(u,v)|=\lambda|D_{\mathcal{O}}(u,v)|
$ and hence $
\Upsilon_{\mathcal{O}}(u,v) = |P_{\mathcal{O}}(u,v)| - \lambda|D_{\mathcal{O}}(u,v)|+\Phi_{\mathcal{O}}  = |P_{\mathcal{O}}(u,v)| - \lambda|D_{\mathcal{O}}(u,v)|
< 0
$ when $n=1$.

Now assuming that the statement is true when there are less than $n$ disks in $\mathcal{O}$, where $n \geq 2$, we will prove that it is true when there are $n$ disks in $\mathcal{O}$.

For the rest of this section, fix an arbitrary chain $\mathcal{O}=(O_1, O_2, \ldots, O_n)$ and an arbitrary terminal $u$ on the boundary of $O_1$.

Before presenting the technical details of the proof for Lemma~\ref{lemma:main}, we give a road-map of the steps involved in the proof. Essentially, this is a step-by-step process of narrowing down the worst case to a specific configuration where the stretch factor of the chain can be bounded by some smooth single-variant functions, whose values can be easily bounded.
\begin{enumerate}
\item First, we eliminate some simple special cases, including the case when $v=a_{n-1}$ or $v=b_{n-1}$ (Proposition~\ref{endsok}) and the case when $u$ and $v$ are obstructed (Proposition~\ref{obsok}).
\item Then by Proposition~\ref{atends} and Proposition~\ref{endofa} we narrow down the worst case to the condition when $v$ is the ``pivot point'' on the boundary of $O_n$, i.e., when the shortest path from $u$ to $v$ including $A_n$ and the the shortest path from $u$ to $v$ including $B_n$ have the same length.
\item Finally we distinguish two cases depending on the angle between $uv$ and $o_{n-1}o_n$. Proposition~\ref{upper} deals with the case when the angle is large, and Proposition~\ref{last} deals with the case when the angle is small, which is by-far the most complicated case.
\end{enumerate}

First consider the special case when $v$ is either $a_{n-1}$ or $b_{n-1}$.
\begin{proposition}
$\Upsilon_{\mathcal{O}}(u,a_{n-1}) < 0$ and $\Upsilon_{\mathcal{O}}(u,b_{n-1}) < 0$.\label{endsok}
\end{proposition}
\begin{proof}
We first prove that $\Upsilon_{\mathcal{O}}(u,a_{n-1}) < 0$. Refer to Figure~\ref{fig:at-end} for an illustration.
Consider the sub-chain $\mathcal{O}_{1,n-1}=(O_1,\ldots,O_{n-1})$. Since $u,a_{n-1}$ are terminals of $\mathcal{O}_{1,n-1}$, by the inductive hypothesis, \begin{align}
&\Upsilon_{\mathcal{O}_{1,n-1}}(u,a_{n-1}) \nonumber\\
=& |P_{\mathcal{O}_{1,n-1}}(u,a_{n-1})| - \lambda|D_{\mathcal{O}_{1,n-1}}(u,a_{n-1})|+\Phi_{\mathcal{O}_{1,n-1}} \nonumber\\ <& 0\label{upsilon_n-1}.\end{align}

We claim that
\begin{align}
|P_{\mathcal{O}}(u,a_{n-1})|\leq |P_{\mathcal{O}_{1,n-1}}(u,a_{n-1})|,\label{ineq:p}
\end{align}  and
\begin{align}
|D_{\mathcal{O}}(u,a_{n-1})| = |D_{\mathcal{O}_{1,n-1}}(u,a_{n-1})|.\label{ineq:d}
\end{align}

\input{fig-v2-chain-at-end}

To verify (\ref{ineq:p}), we note that any arc or line segment in $P_{\mathcal{O}_{1,n-1}}(u,a_{n-1})$ can also be used by $P_{\mathcal{O}}(u,a_{n-1})$, with the exception that the arc $B_{n-1}$ in $\mathcal{O}_{1,n-1}$ between $b_{n-1}$ and $a_{n-1}$ (the dotted arc in Figure~\ref{fig:at-end}) is replaced by a ``shortcut''  in $\mathcal{O}$ via $\overline{b_{n-1}a_{n-1}}$. So clearly $|P_{\mathcal{O}}(u,a_{n-1})|\leq |P_{\mathcal{O}_{1,n-1}}(u,a_{n-1})|$.

Now we verify equality (\ref{ineq:d}). By Definition~\ref{def:polyline}, $D_{\mathcal{O}}(u,a_{n-1})$ is the shortest polyline that consists of line segments $\overline{up_1}$, $\overline{p_1p_2}$, $\ldots$, $\overline{p_{n-1}a_{n-1}}$, where $p_i \in \overline{a_ib_i}$ for $1\leq i \leq n-1$. We can assume $p_{n-1} = a_{n-1}$, because $|\overline{p_{n-2}a_{n-1}}| + |\overline{a_{n-1}a_{n-1}}| = |\overline{p_{n-2}a_{n-1}}| \leq |\overline{p_{n-2}p_{n-1}}| +|\overline{p_{n-1}a_{n-1}}|$ for any $p_{n-1}$ by triangle inequality. So $D_{\mathcal{O}}(u,a_{n-1})$ is the shortest polyline that consists of line segments $\overline{up_1}$, $\overline{p_1p_2}$, $\ldots$, $\overline{p_{n-2}a_{n-1}}$ which is the same as $D_{\mathcal{O}_{1,n-1}}(u,a_{n-1})$. Hence $|D_{\mathcal{O}}(u,a_{n-1})|=|D_{\mathcal{O}_{1,n-1}}(u,a_{n-1})|$ and we have equality (\ref{ineq:d}).

From (\ref{upsilon_n-1}), (\ref{ineq:p}) and (\ref{ineq:d}), we have \begin{align}
\Upsilon_{\mathcal{O}}(u,a_{n-1}) &= |P_{\mathcal{O}}(u,a_{n-1})|- \lambda|D_{\mathcal{O}}(u,a_{n-1})|+\Phi_{\mathcal{O}}\nonumber\\
&\leq |P_{\mathcal{O}_{1,n-1}}(u,a_{n-1})|- \lambda|D_{\mathcal{O}_{1,n-1}}(u,a_{n-1})|+\Phi_{\mathcal{O}}\nonumber\\
&=  |P_{\mathcal{O}_{1,n-1}}(u,a_{n-1})| - \lambda|D_{\mathcal{O}_{1,n-1}}(u,a_{n-1})|+\Phi_{\mathcal{O}_{1,n-1}} +\Phi_{\mathcal{O}}-\Phi_{\mathcal{O}_{1,n-1}} \nonumber\\
&= \Upsilon_{\mathcal{O}_{1,n-1}}(u,a_{n-1})+\Phi_{\mathcal{O}}-\Phi_{\mathcal{O}_{1,n-1}} \nonumber\\
&< \Phi_{\mathcal{O}}-\Phi_{\mathcal{O}_{1,n-1}} \nonumber\\
&\leq 0.
\end{align}
The last inequality is from Lemma~\ref{notincrease}. Similarly, we can prove that $\Upsilon_{\mathcal{O}}(u,b_{n-1}) < 0$. This completes the proof of Proposition~\ref{endsok}.
\end{proof}

\input{fig-v2-chain-obstruct}

Next consider the case when $u$ and $v$ are obstructed. Refer to Figure~\ref{fig:obstruct} for an illustration.
\begin{proposition} If $u$ and $v$ are obstructed, then $\Upsilon_{\mathcal{O}}(u,v) < 0$.\label{obsok}
\end{proposition}
\begin{proof}
If $u$ and $v$ are obstructed, then $D_{\mathcal{O}}(u,v)$ contains a point $p_j$ that is either $a_j$ or $b_j$, for some $1\leq j\leq n-1$. Without loss of generality, assume $p_j=a_j$. Consider two sub-chains of $\mathcal{O}$: $\mathcal{O}_{1,j+1}=(O_1, \ldots, O_{j+1})$ and $\mathcal{O}_{j+1,n}=(O_{j+1}, \ldots, O_n)$. Points $u,a_j$ are terminals of $\mathcal{O}_{1,j+1}$ and points $a_j,v$ are terminals of $\mathcal{O}_{j+1,n}$. If $j=n-1$, by Proposition~\ref{endsok}, \begin{align}
\Upsilon_{\mathcal{O}_{1,j+1}}(u,a_j)=\Upsilon_{\mathcal{O}}(u,a_{n-1}) < 0.\label{uaj1}
\end{align}
If $j<n-1$, $\mathcal{O}_{1,j+1}$ has less than $n$ disks, and by the inductive hypothesis \begin{align}
\Upsilon_{\mathcal{O}_{1,j+1}}(u,a_j) < 0.\label{uaj2}
\end{align}
On the other hand, $\mathcal{O}_{j+1,n}$ has less than $n$ disks, and by the inductive hypothesis, \begin{align}
\Upsilon_{\mathcal{O}_{j+1,n}}(a_j,v) < 0.\label{ajv}
\end{align}
For the similar reasons as those for (\ref{ineq:p}) and (\ref{ineq:d}), we have \begin{align}
|P_{\mathcal{O}}(u,v)| &\leq |P_{\mathcal{O}_{1,j+1}}(u,a_j)|+|P_{\mathcal{O}_{j+1,n}}(a_j,v)|, \mbox{~~~~and}\label{puv}\\
|D_{\mathcal{O}}(u,v)| &= |D_{\mathcal{O}_{1,j+1}}(u,a_j)|+|D_{\mathcal{O}_{j+1,n}}(a_j,v)|.\label{duv}
\end{align} Also,
\begin{align}
\Phi_{\mathcal{O}} &= \varphi(r_n-r_1)- \frac{\varphi}{3}\sum_{i=2}^n(2H_i+V_i) \nonumber\\
&= \varphi(r_n-r_{j+1}+r_{j+1}-r_1)- \frac{\varphi}{3}\sum_{i=2}^{j+1}(2H_i+V_i)-\frac{\varphi}{3}\sum_{i=j+2}^n(2H_i+V_i)\nonumber\\
&= \varphi(r_{j+1}-r_1)- \frac{\varphi}{3}\sum_{i=2}^{j+1}(2H_i+V_i)+\varphi(r_n-r_{j+1})-\frac{\varphi}{3}\sum_{i=j+2}^n(2H_i+V_i)\nonumber\\
&= \Phi_{\mathcal{O}_{1,j+1}}+\Phi_{\mathcal{O}_{j+1,n}}.\label{pppuv}
\end{align} Combining (\ref{uaj1})---(\ref{pppuv}), we have \begin{align}
\Upsilon_{\mathcal{O}}(u,v) &= |P_{\mathcal{O}}(u,v)|-\lambda|D_{\mathcal{O}}(u,v)|+ \Phi_{\mathcal{O}}\nonumber\\
&\leq |P_{\mathcal{O}_{1,j+1}}(u,a_j)|-\lambda|D_{\mathcal{O}_{1,j+1}}(u,a_j)|+\Phi_{\mathcal{O}_{1,j+1}} \nonumber\\
&~~~+|P_{\mathcal{O}_{j+1,n}}(a_j,v)| -\lambda|D_{\mathcal{O}_{j+1,n}}(a_j,v)|+\Phi_{\mathcal{O}_{j+1,n}} \nonumber\\
&=
\Upsilon_{\mathcal{O}_{1,j+1}}(u,a_j)+\Upsilon_{\mathcal{O}_{j+1,n}}(a_j,v) \nonumber\\ &< 0,\end{align}
as desired. This completes the proof of Proposition~\ref{obsok}.
\end{proof}

\input{fig-v2-chain-boundary}

In the rest of the proof for Lemma~\ref{lemma:main}, we assume that $u$ and $v$ are unobstructed, even if it is not explicitly stated. The following definitions are needed. Recall that $u$ is fixed.

\begin{definition}\rm
Consider the set of terminal points $v$ on $O_n$ such that $u$ and $v$ are unobstructed. Such a set forms an arc (illustrated by the think arc in Figure~\ref{fig:pivotal} (a)), denoted by $\widehat{A}$ and referred to as the {\em unobstructed arc}\footnote{For notational convenience, assume that $\widehat{A}$ includes its endpoints. This makes $\widehat{A}$ a closed set.}. Denote by $P_{\mathcal{O}}^{A_n}(u,v)$ and $P_{\mathcal{O}}^{B_n}(u,v)$ the shortest paths from $u$ to $v$ that include the arcs $A_n$ and $B_n$ (on the boundary of $O_n$), respectively. See Figure~\ref{fig:pivotal} (b) and (c) for an illustration. We call $v$ {\em pivotal} if $|P_{\mathcal{O}}^{A_n}(u,v)| = |P_{\mathcal{O}}^{B_n}(u,v)|$. Note that $\widehat{A}$ cannot have more than one pivotal point because $|P_{\mathcal{O}}^{A_n}(u,v)| - |P_{\mathcal{O}}^{B_n}(u,v)|$ varies monotonically as $v$ moves along $\widehat{A}$ in one direction.\label{def:pivot}
\end{definition}

\input{fig-v2-convex}

\begin{proposition} If $u$ is fixed and $u,v$ are unobstructed, then the maximum of $\Upsilon_{\mathcal{O}}(u,v)$ occurs when $v$ is an endpoint of $\widehat{A}$ or when $v$ is pivotal in  $\widehat{A}$.
\label{atends}
\end{proposition}
\begin{proof}
Let $\widehat{A}'$ be an arbitrary sub-arc of $\widehat{A}$ that does not contain a pivotal point in its interior. We will use functional analysis to show that the maximum of $\Upsilon_{\mathcal{O}}(u,v)$ for $v\in \widehat{A}'$ occurs when $v$ is an endpoint of $\widehat{A}'$.

Since $\widehat{A}'$ does not contain a pivotal point in its interior, either $|P_{\mathcal{O}}^{A_n}(u,v)| \leq |P_{\mathcal{O}}^{B_n}(u,v)|$ for every point $v\in \widehat{A}'$ or $|P_{\mathcal{O}}^{B_n}(u,v)| \leq |P_{\mathcal{O}}^{A_n}(u,v)|$ for every point $v\in \widehat{A}'$. Without loss of generality assume that $|P_{\mathcal{O}}^{A_n}(u,v)| \leq |P_{\mathcal{O}}^{B_n}(u,v)|$ for every point $v$ in $\widehat{A}'$. Fixing other parameters, $\Upsilon_{\mathcal{O}}(u,v)$ is a function of $|A_n|$ as $v$ moves along $\widehat{A}'$. One observes the following:
\begin{enumerate}
\item Since the definition of $\Phi_{\mathcal{O}}$ is independent of $v$, $\Phi_{\mathcal{O}}$ remains constant when $v$ moves along $\widehat{A}'$.

\item $|P_{\mathcal{O}}(u,v)|$ is a linear function of $|A_n|$ because $|P_{\mathcal{O}}(u,v)| = \min\{|P_{\mathcal{O}}^{A_n}(u,v)|,|P_{\mathcal{O}}^{B_n}(u,v)|\}= |P_{\mathcal{O}}^{A_n}(u,v)|= |P_{\mathcal{O}}(u,a_{n-1})|+|A_n|$, where $|P_{\mathcal{O}}(u,a_{n-1})|$ remains constant when $v$ moves along $\widehat{A}'$.

\item $-\lambda|D_{\mathcal{O}}(u,v)|$ is a convex function of $|A_n|$. To see why, refer to Figure~\ref{fig:convex}. Let $v_1$ be the location of $v$ after $|A_n|$ is increased by an infinitesimal amount (i.e., $d|A_n|$). Then $\overrightarrow{vv_1}$ is tangent to $O_n$. Let $\sigma$ be the angle from $\overrightarrow{vv_1}$ to $\overrightarrow{uv}$. So $\sigma = \pi/2 - \angle uvo_n$. It is a simple geometric observation that
\begin{align}
\frac{d |D_{\mathcal{O}}(u,v)|}{d|A_n|}=\cos \sigma = \sin(\angle uvo_n).\label{douv}
\end{align}
Also observe that $\angle uvo_n$ decreases as $v$ moves away from $a_{n-1}$ along $\widehat{A}'$ and hence $\frac{d\angle uvo_n}{d|A_n|} \leq 0$. Since $v$ is the exit point of ray $\overrightarrow{uv}$ on $O_n$, $\angle uvo_n\in [-\pi/2,\pi/2]$ and hence $\frac{d\sin(\angle uvo_n)}{d\angle uvo_n} \geq 0$. So
\begin{align}
\frac{d\sin(\angle uvo_n)}{d|A_n|} = \frac{d\sin(\angle uvo_n)}{d\angle uvo_n}\cdot \frac{d\angle uvo_n}{d|A_n|}\leq 0.\label{dsinuvo_n}
\end{align}
Combining (\ref{douv}) and (\ref{dsinuvo_n}), we have
$$\frac{d^2( -\lambda|D_{\mathcal{O}}(u,v)|)}{d|A_n|^2} = -\lambda\cdot \frac{d(\frac{d |D_{\mathcal{O}}(u,v)|}{d|A_n|})}{d|A_n|} = -\lambda\cdot \frac{d(\sin(\angle uvo_n))}{d|A_n|} \geq 0,$$ which implies that $-\lambda|D_{\mathcal{O}}(u,v)|$ is a convex function of $|A_n|$.
\end{enumerate}

Now we know $\Phi_{\mathcal{O}}$, $|P_{\mathcal{O}}(u,v)|$, and $-\lambda|D_{\mathcal{O}}(u,v)|$ are all convex functions of $|A_n|$ (constant and linear functions are also convex). Being a sum of convex functions, $\Upsilon_{\mathcal{O}}(u,v) = |P_{\mathcal{O}}(u,v)| - \lambda|D_{\mathcal{O}}(u,v)|+\Phi_{\mathcal{O}}$ is a convex function of $|A_n|$. This proves that the maximum of $\Upsilon_{\mathcal{O}}(u,v)$ occurs when $v$ is an endpoint of $\widehat{A}'$. Proposition~\ref{atends} follows from this, because:
\begin{itemize}
\item If $\widehat{A}$ does not contain a pivotal point in its interior, then the maximum of $\Upsilon_{\mathcal{O}}(u,v)$ occurs when $v$ is an endpoint of $\widehat{A}$ and hence Proposition~\ref{atends} is true.
\item If $\widehat{A}$ contains a pivotal point in its interior, we can split $\widehat{A}$ at the pivotal point into two sub-arcs $\widehat{A_1}$ and $\widehat{A_2}$ that do not contain a pivotal point in their interiors. So the maximum of $\Upsilon_{\mathcal{O}}(u,v)$ for all $v\in \widehat{A_1}\cup \widehat{A_2}$ occurs when $v$ is an endpoint of $\widehat{A_1}$ or $\widehat{A_2}$.  The endpoints of $\widehat{A_1}$ and $\widehat{A_2}$  are either the pivotal point or the endpoints of $\widehat{A}$. Thus, Proposition~\ref{atends} is also true.
\end{itemize}
This completes the proof of Proposition~\ref{atends}
\end{proof}

In the rest of the proof for Lemma~\ref{lemma:main}, we will show that $\Upsilon_{\mathcal{O}}(u,v) < 0$ if $v$ is an endpoint of $\widehat{A}$ and $\Upsilon_{\mathcal{O}}(u,v) < 0$ if $v$ is pivotal. Once these are proven, then by Proposition~\ref{atends}, we have $\Upsilon_{\mathcal{O}}(u,v) < 0$ for any point $v$ in $\widehat{A}$, which completes the proof of Lemma~\ref{lemma:main}.

\begin{proposition} $\Upsilon_{\mathcal{O}}(u,v) < 0$ if $v$ is an endpoint of $\widehat{A}$.\label{endofa}
\end{proposition}
\begin{proof}
If $v$ is an endpoint of $\widehat{A}$, then there are two cases: (1) $v\in \{a_{n-1}, b_{n-1}\}$ or, (2) $u,v$ are obstructed. See Figure~\ref{fig:pivotal} for an illustration. In either case, we have $\Upsilon_{\mathcal{O}}(u,v)< 0$ by  Proposition~\ref{endsok} and Proposition~\ref{obsok}.
\end{proof}

What remains to be shown is that $\Upsilon_{\mathcal{O}}(u,v) < 0$ when $v$ is pivotal. This requires a careful analysis of the geometry. To start, we fix a coordinate system where the origin is the center point of $\overline{a_{n-1}b_{n-1}}$, the $x$-axis is $o_{n-1}o_n$, and the $y$-axis is $a_{n-1}b_{n-1}$. See Figure~\ref{fig:finalsetup} for an illustration. For any point $p$ in the plane, denote by $X_p$ and $Y_p$ the $x$- and $y$-coordinates of $p$ in the coordinate system. By flipping along the $x$-axis or the $y$-axis (or both) if necessary, we can assume that $X_{o_{n-1}} \leq X_{o_n}$ (i.e., $o_{n-1}$ is to the left of $o_n$) and $Y_{v} \leq Y_u$ (i.e., $u$ is above $v$). Without loss of generality, let $a_{n-1}$ be above $b_{n-1}$ (the analysis will be similar if $b_{n-1}$ is above $a_{n-1}$). By Proposition~\ref{fact}, $\overrightarrow{uv}$ crosses $\overline{a_{n-1}b_{n-1}}$ from left to right. Hence $u$ is to the left of the $y$-axis and $v$ is to the right of the $y$-axis.

Let $q$ be the rightmost intersection between $O_n$ and the $x$-axis. We define the following parameters. Let
$\alpha = \angle qo_na_{n-1}$ and $\beta = \angle vo_nq$. Let $\gamma$ be the angle from $\overrightarrow{uv}$ to the $x$-axis in the {\em counterclockwise} direction.

\input{fig-v2-finalsetup}

The ranges of $\alpha, \beta$ and $\gamma$ are given as follows.
\begin{itemize}
\item Since $a_{n-1}$ is on or above the $x$-axis, $0 \leq \alpha \leq \pi$. If $\alpha=0$ or $\alpha=\pi$, then $a_{n-1}=b_{n-1}$, which means that $D_{\mathcal{O}}(u,v)$ contains $a_{n-1}$ and $\Upsilon_{\mathcal{O}}(u,v) < 0$ by Proposition~\ref{obsok}. So we can assume \begin{align}
0 < \alpha < \pi.\label{range_a}
\end{align}

\item Since $v$ is a pivotal point, we have
$|P_{\mathcal{O}}^{A_n}(u,v)| = |P_{\mathcal{O}}^{B_n}(u,v)|$, where $|P_{\mathcal{O}}^{A_n}(u,v)|=|P_{\mathcal{O}}(u,a_{n-1})|+|A_n|$ and $|P_{\mathcal{O}}^{B_n}(u,v)|=|P_{\mathcal{O}}(u,b_{n-1})|+|B_n|$. Thus,
\begin{align}
|A_n|-|B_n| = |P_{\mathcal{O}}(u,b_{n-1})| - |P_{\mathcal{O}}(u,a_{n-1})|.\label{eq:beta0}
\end{align}
Since $a_{n-1}$ and $b_{n-1}$ are connected by a line segment $\overline{a_{n-1}b_{n-1}}$, the difference between $|P_{\mathcal{O}}(u,a_{n-1})|$ and $|P_{\mathcal{O}}(u,b_{n-1})|$ is at most $||a_{n-1}b_{n-1}||$. In other words,
\begin{align}
-||a_{n-1}b_{n-1}||\leq |P_{\mathcal{O}}(u,b_{n-1})| - |P_{\mathcal{O}}(u,a_{n-1})| \leq ||a_{n-1}b_{n-1}||.\label{eq:beta4}
\end{align}
Combining (\ref{eq:beta0}) and (\ref{eq:beta4}) gives $-||a_{n-1}b_{n-1}||\leq |A_n|-|B_n| \leq ||a_{n-1}b_{n-1}||.$ We further observe that $||a_{n-1}b_{n-1}|| = 2r_n\sin\alpha$ and $|A_n|-|B_n| = 2r_n\beta$. Thus the range of $\beta$ is \begin{align}
-\sin\alpha \leq \beta \leq \sin\alpha.\label{range_b}
\end{align}

\item We have $\gamma \geq 0$ since $Y_{v} \leq Y_u$.

We claim that the largest value of $\gamma$ occurs when $\overline{uv}$ passes through $a_{n-1}$. Here is why.
Refer to Figure~\ref{fig:largestgamma}. Since $\overline{uv}$ crosses $\overline{a_{n-1}b_{n-1}}$ and  $\gamma \geq 0$ , $\overline{uv}$ is sandwiched between the horizontal line passing through $v$ and the line $a_{n-1}v$ (see Figure~\ref{fig:largestgamma}). Within this range, $\gamma$  obtains its maximum value when  $\overline{uv}$ is on the line $a_{n-1}v$, i.e., when $\overline{uv}$ passes through $a_{n-1}$.

\input{fig-v2-finalsetup-1-2}

This means $\gamma \leq \pi/2-\angle b_{n-1}a_{n-1}v = \pi/2-\frac{\angle b_{n-1}o_nv}{2} = \pi/2- (\alpha-\beta)/2$ (as illustrated in Figure~\ref{fig:largestgamma}). So the range of $\gamma$ is \begin{align}
0 \leq \gamma \leq \pi/2-(\alpha-\beta)/2 < \pi/2.\label{range_c}
\end{align}
The last inequality is true because from (\ref{range_b}) $\alpha-\beta \geq \alpha-\sin\alpha > 0.$
\end{itemize}

We proceed by distinguishing two cases depending on the value of $\gamma$ with regard to a threshold value $\gamma^+$, which is defined as  \begin{align}
\gamma^+=\frac{3\sin\alpha-\alpha}{4}+
\arcsin\left(\frac{\alpha+\sin\alpha}{4\lambda\sin(\frac{\alpha+\sin\alpha}{4})}\right).\label{gammaplus}\end{align}
\begin{proposition}\label{upper}
If $v$ is a pivotal point and $\gamma \geq \gamma^+$, then $\Upsilon_{\mathcal{O}}(u,v) < 0.$
\end{proposition}
\begin{proof}
The first step is to show that $|D_{\mathcal{O}}(u,v)|-|D_{\mathcal{O}}(u,a_{n-1})| \geq ||v'v||-||v'a_{n-1}||$, where $v'$ is the entry-point of $\overrightarrow{uv}$ on $O_n$ (see Figure~\ref{fig:finalsetup2}). This requires a careful verification because in general $D_{\mathcal{O}}(u,a_{n-1})$ may not be a straight line segment.

\input{fig-v2-finalsetup-2}

Recall that $u,v$ are unobstructed. So $\overrightarrow{uv}$ crosses all line segments $\overline{a_ib_i}$, $1\leq i \leq n-1$,  in that order. Visualize $D_{\mathcal{O}}(u,v)$ as a ``rubber band'' connecting $u$ and $v$ (illustrated by the dashed line between $u$ and $v$ in Figure~\ref{fig:finalsetup2}). As $v$ moves along the boundary of $O_n$ toward $a_{n-1}$, the line segment $D_{\mathcal{O}}(u,v)$ transforms into a polyline $D_{\mathcal{O}}(u,a_{n-1})$ that ``bends'' around points  $a_j, a_k, \ldots \in \{a
_1,\ldots,a_{n-1}\}$, as illustrated in Figure~\ref{fig:finalsetup2}. Therefore $D_{\mathcal{O}}(u,a_{n-1})$ is a path from $u$ to $a_{n-1}$ that is convex-away from $ua_{n-1}$ and is contained in the area bounded by $\overline{ua_{n-1}}$, $\overline{uv}$, and $A_n$.
Refer to Figure~\ref{fig:finalsetup2}. Let $a_{k}$ be the last turning point in the polyline $D_{\mathcal{O}}(u,a_{n-1})$. So $\overline{a_ka_{n-1}}$ is a part of $D_{\mathcal{O}}(u,a_{n-1})$. By definition, $a_k$ is an intersection between $O_k$ and $O_{k+1}$. Note that $a_k$ and $a_{n-1}$ are terminals of the sub-chain $\mathcal{O}_{k+1,n}$ and $a_k\neq a_{k+1}$ (otherwise we will choose $a_{k+1}$ as the last turning point). Therefore $a_k$ is the entry-point of $\overrightarrow{a_ka_{n-1}}$ on $O_{k+1}$. By Proposition~\ref{fact}, $\overrightarrow{a_ka_{n-1}}$ enters $O_{k+1}$ no later than entering $O_{n}$. This means that $a_k$ appears in the ray $\overrightarrow{a_ka_{n-1}}$ no later than the entry-point of $\overrightarrow{a_ka_{n-1}}$ on $O_n$, denoted henceforth by $v''$. Since $\overline{a_ka_{n-1}}$ is a part of $D_{\mathcal{O}}(u,a_{n-1})$, $\overline{a_ka_{n-1}}$ is contained in the area bounded by $\overline{ua_{n-1}}$, $\overline{uv}$, and $A_n$. This means that $\overrightarrow{a_ka_{n-1}}$ cannot enter $O_n$ via $A_n$ (the blue\footnote{Blue appears as dark-gray in black and white print.} arc in Figure~\ref{fig:finalsetup2}). In other words, the entry-point of $\overrightarrow{a_ka_{n-1}}$ on $O_n$ (i.e., $v''$) is on the arc between $a_{n-1}$ and $v'$ (the green\footnote{Green appears as light-gray in black and white print.} arc in Figure~\ref{fig:finalsetup2}). Since $a_k$ appears in $\overrightarrow{a_ka_{n-1}}$ no later than $v''$, $a_k$ is in $\triangle uv'a_{n-1}$. Now recall that $D_{\mathcal{O}}(u,a_{n-1})$ is a path from $u$ to $a_{n-1}$ that is convex-away from $ua_{n-1}$. Let $$\Pi = \overline{ua_{n-1}}\cup D_{\mathcal{O}}(u,a_{n-1}).$$ Then  $\Pi$ is a convex polygon. Since $\Pi$ contains $\overline{a_ka_{n-1}}$ as an edge, the convex polygon $\Pi$ as a whole is on the same side of the line $a_{n-1}v''$. Combining this with the fact that $\Pi$ is contained in the area bounded by $\overline{ua_{n-1}}$, $\overline{uv}$, and $A_n$, we conclude that $\Pi$ is contained in the triangle $\triangle uv'a_{n-1}$. See Figure~\ref{fig:finalsetup2} for an illustration. It is known that if a convex polygonal body $\Pi$ is contained in another convex polygonal body $\triangle uv'a_{n-1}$, then the length of the boundary of $\Pi$ is less than or equal to the length of the boundary of $\triangle uv'a_{n-1}$ (see~\cite[p. 42]{benson}). Therefore we have $|D_{\mathcal{O}}(u,a_{n-1})|+||ua_{n-1}|| \leq ||uv'||+||v'a_{n-1}||+||ua_{n-1}||$ or, equivalently, $|D_{\mathcal{O}}(u,a_{n-1})| \leq ||uv'||+||v'a_{n-1}||$. Recall that $|D_{\mathcal{O}}(u,v)|=||uv||$. We have \begin{align}
|D_{\mathcal{O}}(u,v)|-|D_{\mathcal{O}}(u,a_{n-1})| &\geq ||uv||-(||uv'||+||v'a_{n-1}||) \nonumber\\
&=||v'v||-||v'a_{n-1}||,\label{ddiff}
\end{align} as claimed in the begin of this proof.

Again refer to Figure~\ref{fig:finalsetup2}. Using simple trigonometric functions, one verifies that $||v'v|| =2r_n\cos(\angle vv'o_n)$ and $||v'a_{n-1}||= 2r_n\cos(\angle o_nv'a_{n-1})$, where $\angle vv'o_n = \angle o_nvv'= \beta-\gamma$ and $\angle o_nv'a_{n-1} = \angle vv'a_{n-1}-\angle vv'o_n = \frac{\angle vo_na_{n-1}}{2}-\angle vv'o_n = (\alpha+\beta)/2-(\beta-\gamma) = \alpha/2-\beta/2+\gamma$. Combining these with (\ref{ddiff}) and applying trigonometric identities, we have \begin{align}
&|D_{\mathcal{O}}(u,v)|-|D_{\mathcal{O}}(u,a_{n-1})| \nonumber\\
&\geq ||v'v||-||v'a_{n-1}|| \nonumber\\
&=
2r_n\cos(\beta-\gamma)-2r_n\cos(\alpha/2-\beta/2+\gamma)\nonumber\\
&=-4r_n\sin(\alpha/4+\beta/4)\sin(3\beta/4-\alpha/4-\gamma).\label{d-d}
\end{align}

Since $v$ is pivotal, $|P_{\mathcal{O}}(u,v)|=|P_{\mathcal{O}}(u,a_{n-1})|+|A_n|=|P_{\mathcal{O}}(u,b_{n-1})|+|B_n|$. Note that $|A_n|=r_n(\alpha+\beta)$. By Proposition~\ref{endsok}, $\Upsilon_{\mathcal{O}}(u,a_{n-1}) = |P_{\mathcal{O}}(u,a_{n-1})|- \lambda|D_{\mathcal{O}}(u,a_{n-1})|+\Phi_{\mathcal{O}} < 0$. Combining these with (\ref{d-d}), we have
\begin{align}
\Upsilon_{\mathcal{O}}(u,v) &= |P_{\mathcal{O}}(u,v)|-\lambda|D_{\mathcal{O}}(u,v)|+\Phi_{\mathcal{O}}\nonumber\\
&= |P_{\mathcal{O}}(u,a_{n-1})|+|A_n|-\lambda|D_{\mathcal{O}}(u,v)|+\Phi_{\mathcal{O}}\nonumber\\
&= (|P_{\mathcal{O}}(u,a_{n-1})|-\lambda|D_{\mathcal{O}}(u,a_{n-1})|+\Phi_{\mathcal{O}}\nonumber) + |A_n| -\lambda|D_{\mathcal{O}}(u,v)|+\lambda|D_{\mathcal{O}}(u,a_{n-1})|\\
&= \Upsilon_{\mathcal{O}}(u,a_{n-1}) +|A_n| - \lambda(|D_{\mathcal{O}}(u,v)|-|D_{\mathcal{O}}(u,a_{n-1})|)\nonumber\\
&<  |A_n| -\lambda(|D_{\mathcal{O}}(u,v)|-|D_{\mathcal{O}}(u,a_{n-1})|)\nonumber\\
&\leq r_n(\alpha+\beta)  + 4\lambda r_n\sin(\alpha/4+\beta/4)\sin(3\beta/4-\alpha/4-\gamma).\label{habc}
\end{align}

Define a function $$h(\alpha,\beta,\gamma)=r_n(\alpha+\beta)+4\lambda r_n\sin(\alpha/4+\beta/4)\sin(3\beta/4-\alpha/4-\gamma).$$

Then from (\ref{habc})
\begin{align}
\Upsilon_{\mathcal{O}}(u,v) < h(\alpha,\beta,\gamma).\label{replacement}
\end{align}

We have \begin{align}
\frac{\partial h}{\partial \gamma}=-4\lambda r_n\sin(\alpha/4+\beta/4)\cos(3\beta/4-\alpha/4-\gamma).\label{hder}
\end{align}
By (\ref{range_b}), we have $
\alpha/4+\beta/4 \geq (\alpha-\sin\alpha)/4 > 0,\label{abl}
$ and $
\alpha/4+\beta/4 \leq (\alpha+\sin\alpha)/4 < \pi/4\label{abh}$. The last inequality is true because $(\alpha+\sin\alpha)/4$ is an increasing function for $\alpha \in (0, \pi)$. So \begin{align}
0 < \alpha/4+\beta/4 < \pi/4.\label{abboth}
\end{align}
By (\ref{range_c}), we have
$
3\beta/4-\alpha/4-\gamma \geq 3\beta/4-\alpha/4-\pi/2+(\alpha-\beta)/2
=\alpha/4+\beta/4-\pi/2
> -\pi/2.\label{cl}$
Also $
3\beta/4-\alpha/4-\gamma < 3\beta/4 \leq (3\sin\alpha)/4 \leq 3/4.\label{ch}
$ So
\begin{align}
-\pi/2 < 3\beta/4-\alpha/4-\gamma < 3/4.\label{bboth}
\end{align}
From (\ref{abboth}) and (\ref{bboth}),
$\sin(\alpha/4+\beta/4) > 0$ and $\cos(3\beta/4-\alpha/4-\gamma) > 0$. Plugging these into (\ref{hder}), we have \begin{align}
\frac{\partial h}{\partial \gamma} < 0.\label{hgamma}
\end{align}

Let  $$\gamma^*=
3\beta/4-\alpha/4+\arcsin\left(\frac{\alpha+\beta}{4\lambda\sin(\alpha/4+\beta/4)}\right).$$ It is easy to verify that $h(\alpha,\beta,\gamma^*) = 0$. By a careful calculation that is given in the appendix (Section~\ref{sec:appendix-gamma}), we can show that
\begin{align}
\gamma^* \leq \gamma^+.\label{gammagamma}
\end{align}

By (\ref{hgamma}), $h$ is a decreasing function of $\gamma$. So for any $\gamma \geq \gamma^+ \geq \gamma^*$, we have $$h(\alpha,\beta,\gamma) \leq h(\alpha,\beta,\gamma^+) \leq h(\alpha,\beta,\gamma^*)=0.$$ Combining this with (\ref{replacement}), we have $\Upsilon_{\mathcal{O}}(u,v) < h(\alpha,\beta,\gamma) \leq 0$ as desired. This completes the proof of Proposition~\ref{upper}.
\end{proof}

\input{fig-v2-finalsetup-2-2}

The only remaining case is when $\gamma$ is below the threshold $\gamma^+$.
\begin{proposition}\label{last}
If $v$ is a pivotal point and $\gamma < \gamma^+$, then $\Upsilon_{\mathcal{O}}(u,v) < 0$.
\end{proposition}
\begin{proof}

Consider the following transformation (see Figure~\ref{fig:finalsetup22}):
\begin{itemize}
\item Transform $O_n$ by fixing two ``anchor'' points $a_{n-1},$  $b_{n-1}$ on its boundary and moving its center $o_n$ along the $x$-axis toward $o_{n-1}$, the center of $O_{n-1}$. During the transformation, the points $a_{n-1}$ and $b_{n-1}$ remain the  intersections between the boundaries of $O_n$ and $O_{n-1}$. The radius of $O_n$ changes during the transformation. In Figure~\ref{fig:finalsetup22}, the dotted disks show the process of the transformation.

\item Adjust the location of $v$ on the boundary of $O_n$ so that $v$ stays pivotal. In Figure~\ref{fig:finalsetup22}, the sequence of large dots shows the locations of $v$ during the transformation.

\item Stop the transformation when $\gamma \geq \gamma^+$ or when $o_n$ reaches $o_{n-1}$.
\end{itemize}

There are two cases that the transformation will end: (1) if $\gamma \geq \gamma^+$ at the end of the transformation, then $\Upsilon_{\mathcal{O}}(u,v) < 0$ by Proposition~\ref{upper}; (2) if $o_n$ reaches $o_{n-1}$ at the end of the transformation, then $\Upsilon_{\mathcal{O}}(u,v) = \Upsilon_{\mathcal{O}_{1,n-1}}(u,v)< 0$ by the inductive hypothesis. So in any case, $\Upsilon_{\mathcal{O}}(u,v)< 0$ at the end of the transformation.

In order to prove the proposition, we need to show that $\Upsilon_{\mathcal{O}}(u,v)< 0$ {\em before} the transformation. Since $\Upsilon_{\mathcal{O}}(u,v)< 0$ at the end of the transformation, it suffices to show that $\Upsilon_{\mathcal{O}}(u,v)$ does not decrease during the transformation. Since the transformation is effected by varying $X_{o_n}$ --- the x-coordinate of $o_n$, and $X_{o_n}$ decreases during the transformation, we only need to show that $$\frac{\partial{\Upsilon_{\mathcal{O}}(u,v)}}{\partial{X_{o_n}}} \leq 0,$$ during the transformation. We will prove this by analyzing the geometry. Recall that we have fixed a coordinate system where the origin is the center point of $\overline{a_{n-1}b_{n-1}}$, the $x$-axis is $o_{n-1}o_n$ where $o_{n-1}$ is to the left of $o_n$, and the $y$-axis is $a_{n-1}b_{n-1}$ where $a_{n-1}$ is above $b_{n-1}$.

Since $
\Upsilon_{\mathcal{O}}(u,v) = |P_{\mathcal{O}}(u,v)| - \lambda|D_{\mathcal{O}}(u,v)|+\Phi_{\mathcal{O}}$ by (\ref{defupsilon}), we need to calculate the partial derivatives $\frac{\partial{|P_{\mathcal{O}}(u,v)|}}{\partial{X_{o_n}}}$, $\frac{\partial{\Phi_{\mathcal{O}}}}{\partial{X_{o_n}}}$, and $\frac{\partial{|D_{\mathcal{O}}(u,v)|}}{\partial{X_{o_n}}}$. The calculations for them are routine but technical, and are given in the appendix (Section~\ref{sec:appendix-calc}).

\begin{align}
\frac{\partial{|P_{\mathcal{O}}(u,v)|}}{\partial{X_{o_n}}} &= \sin\alpha-\alpha\cos\alpha.\label{pder}\\
\frac{\partial{\Phi_{\mathcal{O}}}}{\partial{X_{o_n}}} &= \begin{cases} -\frac{2\varphi}{3}-\frac{4\varphi}{3}\cos\alpha, & \mbox{if } 0< \alpha <\pi/2. \\ -\frac{2\varphi}{3}-\frac{2\varphi}{3}\cos\alpha, & \mbox{if } \pi/2 \leq \alpha <\pi. \end{cases}\label{phider}\\
\frac{\partial{|D_{\mathcal{O}}(u,v)|}}{\partial{X_{o_n}}} &= \cos{\gamma}-\cos\alpha(\cos(\beta-\gamma)+\beta\sin(\beta-\gamma)).\label{dder}
\end{align}

Based on (\ref{dder}), we define a function \begin{align}
f(\alpha, \beta,\gamma)&= -\lambda\frac{\partial{|D_{\mathcal{O}}(u,v)|}}{\partial{X_{o_n}}}\label{fdeforig}\\
&= -\lambda(\cos{\gamma}-\cos\alpha(\cos(\beta-\gamma)+\beta\sin(\beta-\gamma))).\label{fdef}
\end{align}

We can bound $f(\alpha, \beta,\gamma)$ by single-variant functions in $\alpha$, as shown in the following inequality whose proof is quite technical and is given in the appendix  (Section~\ref{sec:appendix-inequ-f}).
\begin{align}
f(\alpha,\beta,\gamma)\leq \begin{cases}  \max\{f(\alpha,\sin\alpha,0), f(\alpha,\sin\alpha,\gamma^+)\}, & \mbox{if } 0< \alpha <\pi/2. \\ \max\{f(\alpha,0,0), f(\alpha,0,\gamma^+)\}, & \mbox{if } \pi/2 \leq \alpha <\pi. \end{cases}\label{gggbbb}
\end{align}

Combining (\ref{pder}), (\ref{phider}), (\ref{dder}), (\ref{fdeforig}), and (\ref{gggbbb}), we have
\begin{align}
\frac{\partial{\Upsilon_{\mathcal{O}}(u,v)}}{\partial{X_{o_n}}} &= \frac{\partial{|P_{\mathcal{O}}(u,v)|}}{\partial{X_{o_n}}} - \lambda\frac{\partial{|D_{\mathcal{O}}(u,v)|}}{\partial{X_{o_n}}}
+\frac{\partial{|\Phi_{\mathcal{O}}|}}{\partial{X_{o_n}}}\nonumber\\
&= \frac{\partial{|P_{\mathcal{O}}(u,v)|}}{\partial{X_{o_n}}} + f(\alpha, \beta,\gamma)
+\frac{\partial{|\Phi_{\mathcal{O}}|}}{\partial{X_{o_n}}}\nonumber\\
&= \begin{cases}  \sin\alpha-\alpha\cos\alpha +f(\alpha, \beta,\gamma) -\frac{2\varphi}{3}-\frac{4\varphi}{3}\cos\alpha, & \mbox{if } 0< \alpha <\pi/2 \\ \sin\alpha-\alpha\cos\alpha +f(\alpha, \beta,\gamma) -\frac{2\varphi}{3}-\frac{2\varphi}{3}\cos\alpha, & \mbox{if } \pi/2 \leq \alpha <\pi \end{cases}\nonumber\\
&\leq \begin{cases}  \sin\alpha-\alpha\cos\alpha +\max\{f(\alpha,\sin\alpha,0), f(\alpha,\sin\alpha,\gamma^+)\} -\frac{2\varphi}{3}-\frac{4\varphi}{3}\cos\alpha, & \mbox{if } 0< \alpha <\pi/2. \\ \sin\alpha-\alpha\cos\alpha +\max\{f(\alpha,0,0), f(\alpha,0,\gamma^+)\} -\frac{2\varphi}{3}-\frac{2\varphi}{3}\cos\alpha, & \mbox{if } \pi/2 \leq \alpha <\pi. \end{cases}\label{finaleq}
\end{align}

By (\ref{finaleq}), we only need to verify the following four inequalities, where $f(\alpha, \beta,\gamma) = -\lambda(\cos{\gamma}-\cos\alpha(\cos(\beta-\gamma)+\beta\sin(\beta-\gamma)))$ and $\gamma^+=\frac{3\sin\alpha-\alpha}{4}+\arcsin\left(\frac{\alpha+\sin\alpha}{4\lambda\sin(\frac{\alpha+\sin\alpha}{4})}\right).$ \begin{align}
g_1(\alpha)&=\sin\alpha-\alpha\cos\alpha -\frac{2\varphi}{3}-\frac{2\varphi}{3}\cos\alpha +f(\alpha,0,0)  < 0,\mbox{~ when $\pi/2 \leq \alpha < \pi$}.\label{ineq1}\\
g_2(\alpha)&=\sin\alpha-\alpha\cos\alpha -\frac{2\varphi}{3}-\frac{2\varphi}{3}\cos\alpha +f(\alpha,0,\gamma^+)  < 0,\mbox{~ when $\pi/2 \leq \alpha < \pi$}.\label{ineq2}\\
g_3(\alpha)&=\sin\alpha-\alpha\cos\alpha -\frac{2\varphi}{3}-\frac{4\varphi}{3}\cos\alpha +f(\alpha,\sin\alpha,0) < 0,\mbox{~ when $0 < \alpha < \pi/2$}.\label{ineq3}\\
g_4(\alpha)&=\sin\alpha-\alpha\cos\alpha -\frac{2\varphi}{3}-\frac{4\varphi}{3}\cos\alpha +f(\alpha,\sin\alpha,\gamma^+) < 0,\mbox{~ when $0 < \alpha < \pi/2$}.\label{ineq4}
\end{align}
Since $g_1,g_2,g_3$ and $g_4$ are smooth single-variant functions on small intervals of $\alpha$, one can easily verify the above inequalities using numerical computing software, such as Mathematica. For completeness, a more formal verification is given in the appendix (Section~\ref{sec:appendix-all-ineqs}). There, we show that $g_1,g_2,g_3$ and $g_4$ have small Lipschitz constants, and then use a program that implements a simplified Piyavskii's algorithm~\cite{Vanderbei97extensionof} to verify that their upper bounds are less than 0. This completes the proof of Proposition~\ref{last}
\end{proof}

This is the end of the proof of Lemma~\ref{lemma:main}.

\section{Proof of Lemma~\ref{lemma:pot}} \label{section:main}

Let $\mathbb{O}$ be a set of chains whose stretch factor is greater than or equal to a threshold $\tau$. In this section, we will prove that if $\mathbb{O}$ is non-empty, then there exists a chain $\mathcal{O}^* \in \mathbb{O}$ with terminals $u,v$ such that $|P_{\mathcal{O}^*}(u,v)|/|D_{\mathcal{O}^*}(u,v)| \geq \tau$ and
$\Phi_{\mathcal{O}^*} \geq -\frac{\sqrt{5}\varphi}{3}|P_{\mathcal{O}^*}(u,v)|$.

Suppose that $\mathbb{O}$ is non-empty. Let $\mathbb{E}$ be the subset of $\mathbb{O}$ consisting of chains in $\mathbb{O}$ with a minimum number of disks. $\mathbb{E}$ is non-empty because $\mathbb{O}$ is non-empty. The number of disks in a chain of $\mathbb{E}$ is denoted $n$.  Next, we will choose a chain $\mathcal{O}^*$ from $\mathbb{E}$ such that the total radii of the disks in $\mathcal{O}^*$ is minimized. We justify below that such a $\mathcal{O}^*$ always exists when $\mathbb{E}$ is non-empty.

Note that the definition of the chain (Definition~\ref{def:circles}) includes the boundary cases: the case when two consecutive disks are tangent and the case when two connecting arcs of the same disk share an end point. Also note that the definition of $\mathbb{O}$ includes the boundary case, i.e., the case when the stretch factor is equal to the threshold $\tau$. Therefore, $\mathbb{E}$ also contains its boundary cases. In other words, $\mathbb{E}$ is a closed  set (but $\mathbb{E}$ is not necessarily bounded).
Associate with every chain $\mathcal{O} \in \mathbb{E}$  a pair of terminals $u,v$ that yields the worst stretch factor of $\mathcal{O}$. Every $\mathcal{O}\in \mathbb{E}$ can be represented\footnote{modulo rotating and scaling, which do not affect the stretch factor} by a vector $\mathbf{x} \in \mathcal{R}^{3n}$ that specifies, for each of the $n$ disks in $\mathcal{O}$, its radius and the $x$- and $y$-coordinates of its center in a (normalized) coordinate system where $u$ is the origin and $v$ is $(1,0)$. Therefore $\mathbb{E}$ can be mapped to a non-empty closed set $\mathbb{S} \subseteq \mathcal{R}^{3n}$. Define a function $\mathcal{H}: \mathbb{S} \rightarrow \mathcal{R}$ as $\mathcal{H}(\mathbf{x})=\sum_{i=1}^n r_i$, where $\mathbf{x}\in \mathbb{S}$ and $r_1, \ldots,r_n$ are the radii of the disks in the chain $\mathcal{O}$ represented by $\mathbf{x}$. $\mathcal{H}(\mathbf{x})$ is a continuous (linear) function on $\mathbb{S}$. Observe that when the $L_2$-norm of $\mathbf{x}$, $||\mathbf{x}||_2$, approaches infinity, the length of the {\em centered polyline} (the polyline connecting the centers of the disks in $\mathcal{O}$, see Definition~\ref{def:terminals}) also approaches infinity. Therefore, to prevent the chain $\mathcal{O}$ from being broken, the total radii of the disks in the chain (i.e., $\mathcal{H}(\mathbf{x})$) must also approach infinity. In other words, $$\lim_{||\mathbf{x}||_2\rightarrow \infty}\mathcal{H}(\mathbf{x}) \rightarrow \infty.$$ This means that $\mathcal{H}$ is a {\em coercive} function on $\mathbb{S}$. It is known that a continuous coercive function on a non-empty closed set $\mathbb{S} \subseteq \mathcal{R}^{3n}$ has a global minimum, regardless whether $\mathbb{S}$ is bounded or not (see \cite[p.60]{extremabook}). So $\mathcal{H}$ has a global minimum on $\mathbb{S}$. Therefore we can choose $\mathcal{O}^*$ to be the chain in $\mathbb{E}$ that achieves the global minimum of $\mathcal{H}$. Let $u,v$ be the terminals associated with $\mathcal{O}^*$ that yield the worst stretch factor. Then $\mathcal{O}^*$ satisfies three conditions:
\begin{enumerate}
    \item Since $\mathcal{O}^* \in \mathbb{O}$, $\mathcal{O}^*$ has stretch factor $\geq \tau$, i.e., $|P_{\mathcal{O}^*}(u,v)|/|D_{\mathcal{O}^*}(u,v)| \geq \tau$.

    \item Since $\mathcal{O}^* \in \mathbb{E}$, the number of disks in $\mathcal{O}^*$ is minimized among all chains in $\mathbb{O}$.

    \item Since $\mathcal{O}^*$ minimizes $\mathcal{H}$, the sum of the radii $\sum_{O_i\in\mathcal{O}^*}r_i$ is minimized among all chains in $\mathbb{E}$.
\end{enumerate}

Because $\mathcal{O}^*$ satisfies conditions 1---3, we can prove that it has the following two properties (Propositions~\ref{vunobs} and \ref{equal}), which may not hold for a general chain.

\begin{proposition}
$u$ and $v$ are unobstructed in $\mathcal{O}^*$.\label{vunobs}
\end{proposition}
\begin{proof}
Suppose that $D_{\mathcal{O}^*}(u,v)$ is obstructed. Then $D_{\mathcal{O}^*}(u,v)$ contains a point $p_j$ which is either $a_j$ or $b_j$, for some $1\leq j \leq n-1$. Consider two sub-chains of $\mathcal{O}^*$: $\mathcal{O}_{1,j}^*=(O_1,\ldots,O_j)$ and $\mathcal{O}_{j+1,n}^*=(O_{j+1},\ldots,O_n)$. So $u,p_j$ are terminals of $\mathcal{O}_{1,j}^*$ and $p_j,v$ are terminals of $\mathcal{O}_{j+1,n}^*$.
For the similar reasons as those for (\ref{ineq:p}) and (\ref{ineq:d}), $|P_{\mathcal{O}^*}(u,v)| \leq |P_{\mathcal{O}_{1,j}^*}(u,p_j)|+|P_{\mathcal{O}_{j+1,n}^*}(p_j,v)|$ and $|D_{\mathcal{O}^*}(u,v)| = |D_{\mathcal{O}_{1,j}^*}(u,p_j)|+|D_{\mathcal{O}_{j+1,n}^*}(p_j,v)|$. Since $|P_{\mathcal{O}^*}(u,v)|/|D_{\mathcal{O}^*}(u,v)| \geq \tau$, we have either $|P_{\mathcal{O}_{1,j}^*}(u,p_j)|/|D_{\mathcal{O}_{1,j}^*}(u,p_j)|\geq \tau$ or $|P_{\mathcal{O}_{j+1,n}^*}(p_j,v)|/|D_{\mathcal{O}_{j+1,n}^*}(p_j,v)|\geq \tau$. This means that either $\mathcal{O}_{1,j}^* \in \mathbb{O}$ or $\mathcal{O}_{j+1,n}^* \in \mathbb{O}$ and both have less number of disks than $\mathcal{O}^*$---a contradiction to condition 2 of $\mathcal{O}^*$. So $u,v$ must be unobstructed in $\mathcal{O}^*$. This completes the proof of Proposition~\ref{vunobs}.
\end{proof}

\begin{proposition}
Both $A_1A_2\ldots A_n$ and $B_1B_2\ldots B_n$ are shortest paths between $u$ and $v$ in $\mathcal{O}^*$.\label{equal}
\end{proposition}
\begin{proof}
Suppose that the statement is not true. Then at least of one of these two paths, say $B_1B_2\ldots B_n$, is not a shortest path between $u$ and $v$ in $\mathcal{O}^*$. Let $Q$ be a shortest path between $u$ and $v$ that contains the longest prefix $B_1B_2\ldots B_{j-1}$ of $B_1B_2\ldots B_n$. If $B_1B_2\ldots B_{j-1}$ is empty, then $j=1$. We have a few observations:
\begin{itemize}
\item $B_j$ is not in any shortest path between $u$ and $v$, because if there is a shortest path $Q'$ that contains $B_j$, then replacing the subpath of $Q'$ before $B_j$ by $B_1B_2\ldots B_{j-1}$ yields another shortest path $Q''$ which contains a longer prefix $B_1B_2\ldots B_j$ than $Q$.

\item Since every shortest path between $u$ and $v$ cannot contain $B_j$, every shortest path between $u$ and $v$ must contain $A_j$ because $\{A_j,B_j\}$ is a cut that separates $u$ and $v$.

\item Since $Q$ contains $B_{j-1}$ and $A_j$ (but not $B_j$), it must also contain $\overline{a_{j-1}b_{j-1}}$. In the case when $j=1$, $B_0$ is degenerated and $a_0 = b_0 = u$

\item $B_j$ is not degenerated because: (i) if both $A_j$ and $B_j$ are degenerated, then $O_j$ can be removed from $\mathcal{O}^*$ without changing the stretch factor---a contradiction to condition 2 of $\mathcal{O}^*$; (ii) if $B_j$ is degenerated and $A_j$ is not degenerated, then $Q$ can be further shortened by taking $\overline{a_jb_j}$ as a short-cut from $b_{j-1}$ (which equals $b_j$) to $a_j$ instead of taking $\overline{a_{j-1}b_{j-1}}$ and $A_j$---a contradiction to the fact that $Q$ is a shortest path.
\end{itemize}

We perform a transformation that shrinks $O_j$ by reducing $r_j$, as follows: fix $a_{j-1}$ and $a_j$ on the boundary of $O_j$ and reduce $O_j$'s radius $r_j$ by a small amount; in case where $A_j$ is degenerated, fix $a_j$ (which equals $a_{j-1}$) on the boundary of $O_j$ and move the center $o_j$ toward $a_j$ by a small amount along the line segment $\overline{o_ja_j}$. See Figure~\ref{fig:shrink} for illustrations. Such transformation can be performed while maintaining the following three properties:

\input{fig-v2-shrink}

\begin{enumerate}
\item $\mathcal{O}^*$ remains a chain. Since $a_{j-1}$ and $a_j$ stay on the boundary of $O_j$, $O_j$ remains intersected with $O_{j-1}$ and $O_{j+1}$ during the transformation. So property (1) of a chain (in Definition~\ref{def:circles}) is satisfied. Shrinking $O_j$ will shrink the connecting arcs $C_{j-1}^{(j)}$ and $C_{j+1}^{(j)}$, while $C_{j-1}^{(j-2)}$ and $C_{j+1}^{(j+2)}$ remain the same. Therefore the two connecting arcs on $C_{j-1}$, namely $C_{j-1}^{(j-2)}$ and $C_{j-1}^{(j)}$ will not overlap after shrinking $O_j$. For the same reason, $C_{j+1}^{(j)}$ and $C_{j+1}^{(j+2)}$ will not overlap after shrinking $O_j$. The only other possible way that property (2) of a chain can be violated is for $C_{j}^{(j-1)}$ and $C_{j}^{(j+1)}$ to overlap; but we know that $B_j$ is not degenerated. So shrinking $O_j$ by a sufficiently small amount will not make $C_{j}^{(j-1)}$ and $C_{j}^{(j+1)}$ overlap.

\item $B_j$ is still not in any shortest path between $u$ and $v$. This is because the lengths of the paths between $u$ and $v$ changes continuously when $O_j$ shrinks. So we can shrink $O_j$ by a sufficiently small amount such that $B_j$ is still not in any shortest path between $u$ and $v$.

\item $|D_{\mathcal{O}^*}(u,v)|$ remains unchanged. This is because by Proposition~\ref{vunobs}, $u,v$ are unobstructed in $\mathcal{O}^*$. So we can shrink $O_j$ by a sufficiently small amount such that $u,v$ remain unobstructed, which means that $|D_{\mathcal{O}^*}(u,v)|=||uv||$ remains the same.
\end{enumerate}

We claim that shrinking $O_j$ by a sufficiently small amount will not decrease  $|P_{\mathcal{O}^*}(u,v)|$. We will prove the claim by showing that the length of any shortest path between $u$ and $v$ in $\mathcal{O}^*$ is not decreased after shrinking $O_j$. Since shrinking $O_j$ does not affect $\mathcal{O}^*$ except for $A_j$ and the arcs and line segments containing $b_{j-1}$ or $b_j$ (including $B_{j-1}$, $B_j$, $B_{j+1}$,  $\overline{a_{j-1}b_{j-1}}$, and $\overline{a_jb_j}$), we only need to consider the effect of shrinking $O_j$ on shortest paths that contain $A_j$, $b_{i-1}$ or $b_j$.

First consider the effect of shrinking $O_j$ on $|A_j|$. We have $|A_j| \leq ||a_{j-1}b_{j-1}||+|B_j|+||a_jb_j||$ because otherwise $A_j$ will not be in a shortest path. Hence $|A_j|$ is less than half of the circumference of $O_j$. Therefore, shrinking $O_j$ while keeping $a_{j-1}$ and $a_j$ on the boundary of $O_j$ will increase $|A_j|$. For an illustration, see Figure~\ref{fig:shrink} (a) which shows that shrinking $O_j$ increases $|A_j|$.

Next consider the effect of shrinking $O_j$ on any shortest path containing $b_{j-1}$ or $b_j$. Let $P$ be an arbitrary shortest path between $u$ and $v$ in $\mathcal{O}^*$ that contains $b_j$. Since $P$ contains $b_j$ but not $B_j$, $P$ includes $A_j$, $\overline{a_jb_j}$, and $B_{j+1}$. Let $b_j'$ be the new location of $b_j$ after the transformation of $O_j$ (see Figure~\ref{fig:shrink}). Then the arc $\wideparen{b_{j}b_{j}'}$ becomes an extension of $B_{j+1}$. So after shrinking $O_j$, $|B_{j+1}|$ becomes $|B_{j+1}|+|\wideparen{b_{j}b_{j}'}|$ and $||a_jb_j||$ becomes $||a_jb_j'||$. Hence, $|P|$ becomes $|P|+|\wideparen{b_{j}b_{j}'}|+||a_jb_j'||-||a_jb_j||$ after shrinking $O_j$. Since $||a_jb_j||$ is the shortest distance between $a_j$ and $b_j$, we have
\begin{align}
|\wideparen{b_{j}b_{j}'}|+||a_jb_j'||-||a_jb_j|| > 0.\label{ajbj}
\end{align} For an illustration, see Figure~\ref{fig:shrink} (b) which shows that $b_4'$ is the new location of $b_4$ after shrinking $O_4$ and $|\wideparen{b_{4}b_{4}'}|+||a_4b_4'||-||a_4b_4||>0$.  This implies that shrinking $O_j$ increases $|P|$. By a similar argument, shrinking $O_j$ increases the length of any shortest path containing $b_{j-1}$.

So in any case, shrinking $O_j$ will not decrease  $|P_{\mathcal{O}^*}(u,v)|$.

In summary, we have proven that if $B_1B_2\ldots B_n$ is not a shortest path, then by shrinking a disk $O_j\in \mathcal{O}^*$ we have a new chain that satisfies condition 1 and 2 of $\mathcal{O}^*$, and has a smaller sum of radii than $\mathcal{O}^*$, which is a contradiction to condition 3 of $\mathcal{O}^*$. Therefore the statement of Proposition~\ref{equal} is true.
\end{proof}

In order to prove Lemma~\ref{lemma:pot}, we then only need to prove the following.
\begin{proposition}\label{prop:lasatone}
$\Phi_{\mathcal{O}^*} \geq -\frac{\sqrt{5}\varphi}{3}|P_{\mathcal{O}^*}(u,v)|.$
\end{proposition}
\begin{proof}
Recall that $\Phi_{\mathcal{O}^*}=\varphi(r_n-r_1)- \frac{\varphi}{3}\sum_{i=2}^n(2H_i+V_i)$. Without loss of generality, we can assume that $r_n\geq r_1$, because if this is not the case we can reverse the labels of the disks in $\mathcal{O}^*$. So it suffices to show that
\begin{align}
\sum_{i=2}^n(2H_i+V_i)\leq \sqrt{5}|P_{\mathcal{O}^*}(u,v)|.\label{ineq:lastone2}
\end{align}

See Figure~\ref{fig:equals}. We create a chain $\overline{\mathcal{O}^*}$ from $\mathcal{O}^*$ where the centers $o_1\ldots o_{n}$ are aligned on a straight line, as follows: first rotate $O_1$ around $o_1$ such that $u$ comes to a location where $|A_1|=|B_1|$; then rotate $O_1$ and $O_2$ as a whole around $o_2$ such that $|A_2|=|B_2|$; then rotate $O_1,O_2$, and $O_3$ as a whole around $o_3$ such that $|A_3|=|B_3|$; and so on. $\overline{\mathcal{O}^*}$ is the resulting chain after the rotation operations.

Observe that the rotation operations do not change the total length of the boundary of the chain. So $\sum_{i=1}^n (|A_i|+|B_i|)$ in $\overline{\mathcal{O}^*}$ remains the same as in $\mathcal{O}^*$, which is $2|P_{\mathcal{O}^*}(u,v)|$ by Proposition~\ref{equal}. In $\overline{\mathcal{O}^*}$, since $A_i=B_i$ for $1\leq i \leq n$, we have \begin{align}
|A_1\ldots A_n|=|B_1\ldots B_n| = |P_{\mathcal{O}^*}(u,v)|.\label{aaa}
\end{align}

\input{fig-v2-equals}

Recall from Definition~\ref{peak} that $H_i$ and $V_i$ are the horizontal and vertical distance traveled along the path $\mathcal{P}_i$, where $\mathcal{P}_i$ is a path from $q_{i-1}^{\rightarrow}$, a peak of $O_{i-1}$ with regard to $o_{i-1}o_i$, to $q_i^{\leftarrow}$, a peak of $O_i$ with regard to $o_{i-1}o_i$. Also recall that the arcs in $\mathcal{P}_i$ are ``light'' or ``heavy'' depending on whether the arc is on the boundary of the chain or not. The light arcs in $\mathcal{P}_i$ contribute positively to $H_i$ and $V_i$, and the heavy arcs in $\mathcal{P}_i$ contribute negatively to $H_i$ and $V_i$. Observe that $H_i$ and $V_i$ are determined by the sizes of $O_{i-1},O_i$ and the distance between them. The rotation operations do not affect the sizes and the distance of $O_{i-1}$ and $O_i$. So $\sum_{i=2}^n(2H_i+V_i)$ in $\overline{\mathcal{O}^*}$ remains the same as in $\mathcal{O}^*$.

Since the centers $o_1\ldots o_{n}$ are all aligned on a straight line  in $\overline{\mathcal{O}^*}$, the peaks $q_i^{\leftarrow}$ and $q_i^{\rightarrow}$ of every disk $O_i$ overlap, for $2\leq i\leq n-1$. So the paths $\mathcal{P}_2,\ldots,\mathcal{P}_n$ join at the peaks to form a single path from $q_1^{\rightarrow}$ to $q_n^{\leftarrow}$; denote it by $\mathcal{P}_{\overline{\mathcal{O}^*}}$ (see Figure~\ref{fig:equals}). Let $H_{\overline{\mathcal{O}^*}}$ and $V_{\overline{\mathcal{O}^*}}$ the horizontal and vertical distance traveled along the path $\mathcal{P}_{\overline{\mathcal{O}^*}}$. Then $H_{\overline{\mathcal{O}^*}} = \sum_{i=2}^n H_i$ and $V_{\overline{\mathcal{O}^*}} = \sum_{i=2}^n V_i$.

We further refine $\mathcal{P}_{\overline{\mathcal{O}^*}}$ into another path $\mathcal{P}_{\overline{\mathcal{O}^*}}'$ as follows (see Figure~\ref{fig:equals}). First we let the overlapping portions of heavy arcs and light arcs in $\mathcal{P}_{\overline{\mathcal{O}^*}}$ cancel each other out; this will not affect $H_{\overline{\mathcal{O}^*}}$ and $V_{\overline{\mathcal{O}^*}}$. All heavy arcs in $\mathcal{P}_{\overline{\mathcal{O}^*}}$ will be canceled in this manner except for those at the beginning or the end of $\mathcal{P}_{\overline{\mathcal{O}^*}}$, which can then be safely removed since they contribute negatively to $H_{\overline{\mathcal{O}^*}}$ and $V_{\overline{\mathcal{O}^*}}$. The resulting path is $\mathcal{P}_{\overline{\mathcal{O}^*}}'$. Let $s$ and $t$ be the end points of $\mathcal{P}_{\overline{\mathcal{O}^*}}'$.

Let $H_{\overline{\mathcal{O}^*}}'$ and $V_{\overline{\mathcal{O}^*}}'$ be the horizontal and vertical distance traveled by $\mathcal{P}_{\overline{\mathcal{O}^*}}'$. Then \begin{align}
\sum_{i=2}^n H_i = H_{\overline{\mathcal{O}^*}} \leq H_{\overline{\mathcal{O}^*}}'\label{hibound}
\end{align}
and
\begin{align}
\sum_{i=2}^n V_i = V_{\overline{\mathcal{O}^*}} \leq V_{\overline{\mathcal{O}^*}}'.\label{vibound}
\end{align}
Since $\mathcal{P}_{\overline{\mathcal{O}^*}}'$ is a path between $s$ and $t$,
by the generalized triangle inequality, we have
\begin{align}
\sqrt{(H_{\overline{\mathcal{O}^*}}')^2+(V_{\overline{\mathcal{O}^*}}')^2} &= \sqrt{\left(\int_s^t |dx|\right)^2+\left(\int_s^t |dy|\right)^2} \nonumber\\
&\leq \int_s^t\sqrt{(dx)^2+(dy)^2} \nonumber\\
&= |\mathcal{P}_{\overline{\mathcal{O}^*}}'|.\label{ineq:gri}
\end{align}
Since $\mathcal{P}_{\overline{\mathcal{O}^*}}'$ is a subpath of $A_1\ldots A_n$, from (\ref{aaa}), we have
\begin{align}
|\mathcal{P}_{\overline{\mathcal{O}^*}}'|\leq |A_1\ldots A_n|= |P_{\mathcal{O}^*}(u,v)|.\label{paaa}
\end{align}

Combining (\ref{hibound}), (\ref{vibound}), (\ref{ineq:gri}) and (\ref{paaa}) gives
\begin{align}
\sqrt{(\sum_{i=2}^n H_i)^2+(\sum_{i=2}^n V_i)^2} &\leq \sqrt{(H_{\overline{\mathcal{O}^*}}')^2+(V_{\overline{\mathcal{O}^*}}')^2} \nonumber\\
&\leq |\mathcal{P}_{\overline{\mathcal{O}^*}}'| \nonumber\\
&\leq |P_{\mathcal{O}^*}(u,v)|.\label{ineq:lastone3}
\end{align}
Applying the Cauchy-Schwarz inequality and using (\ref{ineq:lastone3}), we have \begin{align}
2\sum_{i=2}^n H_i+\sum_{i=2}^n V_i &\leq \sqrt{2^2+1^2}\cdot\sqrt{(\sum_{i=2}^n H_i)^2+(\sum_{i=2}^n V_i)^2}\nonumber\\
&= \sqrt{5}\cdot\sqrt{(\sum_{i=2}^n H_i)^2+(\sum_{i=2}^n V_i)^2}\nonumber\\
&\leq \sqrt{5}\cdot|P_{\mathcal{O}^*}(u,v)|,
\end{align} as required by (\ref{ineq:lastone2}). This proves Proposition~\ref{prop:lasatone}.
\end{proof}
This completes the proof of Lemma~\ref{lemma:pot}.

\section{Conclusions}\label{sec:conclu}
In this paper, we showed that the stretch factor of the Delaunay triangulation is less than 1.998 by proving the same upper bound on the stretch factor of the chain.

There are a few places where our approach can be further improved. Firstly, the potential function can be improved to yield a better upper bound. For example, if we define the potential function $\Phi_{\mathcal{O}}$ to be the length of the segment of $\overline{uv}$ inside $O_n$, then we can improve the upper bound to 1.98, although the analysis is quite complicated. Secondly, the key components of our proof are Proposition~\ref{upper} and Proposition~\ref{last}, whose proofs rely largely on functional analysis. We hope to gain insight of the underlying geometry that will help us simplify the proofs and push the upper bound closer to the tight bound.
~\\

\noindent{\bf Acknowledgement:} This work is supported in part by a Lafayette College research grant. The author thanks the anonymous reviewers for their valuable comments to improve the quality of the paper. The author is also grateful to Iyad Kanj and Shiliang Cui for helpful discussions related to the paper.

\bibliographystyle{plain}
\bibliography{ref}


\newpage

\section{Appendix}

\subsection{Proofs for Equality (\ref{hivi1}) and Inequality (\ref{hivi2})}\label{sec:appendix-hivi}

For this proof, fix a coordinate system where the origin is $o_i$, $x$-axis is $\overrightarrow{o_{i-1}o_i}$, and $a_{i-1}$ is on or above the $x$-axis. Let $X_{q_{i-1}^{\rightarrow}}$ and $Y_{q_{i-1}^{\rightarrow}}$ be the $x$- and $y$-coordinates of $q_{i-1}^{\rightarrow}$. Let $X_{q_i^{\leftarrow}}$ and $Y_{q_i^{\leftarrow}}$ be the $x$- and $y$-coordinates of $q_i^{\leftarrow}$. Let $X_{a_{i-1}}$ and $Y_{a_{i-1}}$ be the $x$- and $y$-coordinates of $a_{i-1}$. We distinguish three cases.
\begin{itemize}
\item Case 1. $Q_{i-1}^{\rightarrow}$ is light and $Q_i^{\leftarrow}$ is heavy, as in Figure~\ref{fig:phi}(a). In this case, $r_i < r_{i-1}$, $X_{a_{i-1}} > X_{q_{i-1}^{\rightarrow}}$ and $X_{a_{i-1}} > X_{q_i^{\leftarrow}}$. The horizontal distance traveled by $Q_{i-1}^{\rightarrow}$ and $Q_i^{\leftarrow}$ are $X_{a_{i-1}}-X_{q_{i-1}^{\rightarrow}}$ and $X_{a_{i-1}}-X_{q_i^{\leftarrow}}$, respectively. So
 $$H_i=(X_{a_{i-1}}-X_{q_{i-1}^{\rightarrow}})-(X_{a_{i-1}}-X_{q_i^{\leftarrow}})=
X_{q_i^{\leftarrow}}-X_{q_{i-1}^{\rightarrow}}=||o_io_{i-1}||.$$ In this case, the vertical distance traveled by $Q_{i-1}^{\rightarrow}$ and $Q_i^{\leftarrow}$ are $Y_{q_{i-1}^{\rightarrow}}-Y_{a_{i-1}}$ and $Y_{q_i^{\leftarrow}}-Y_{a_{i-1}}$, respectively (this is the same for all three cases). We have $$V_i=(Y_{q_{i-1}^{\rightarrow}}-Y_{a_{i-1}})-(Y_{q_i^{\leftarrow}}-Y_{a_{i-1}})= Y_{q_{i-1}^{\rightarrow}}-Y_{q_i^{\leftarrow}}=r_{i-1}-r_i=|r_i-r_{i-1}|.$$

\item Case 2. $Q_{i-1}^{\rightarrow}$ is heavy and $Q_i^{\leftarrow}$ is light, as in Figure~\ref{fig:phi}(b). In this case, $r_i > r_{i-1}$, $X_{q_{i-1}^{\rightarrow}} > X_{a_{i-1}}$ and $X_{q_i^{\leftarrow}} > X_{a_{i-1}}$. The horizontal distance traveled by $Q_{i-1}^{\rightarrow}$ and $Q_i^{\leftarrow}$ are $X_{q_{i-1}^{\rightarrow}}-X_{a_{i-1}}$ and $X_{q_i^{\leftarrow}}-X_{a_{i-1}}$, respectively. So $$H_i=-(X_{q_{i-1}^{\rightarrow}}-X_{a_{i-1}})+(X_{q_i^{\leftarrow}}-X_{a_{i-1}})=
X_{q_i^{\leftarrow}}-X_{q_{i-1}^{\rightarrow}}=||o_io_{i-1}||.$$ Same as in case 1, the vertical distance traveled by $Q_{i-1}^{\rightarrow}$ and $Q_i^{\leftarrow}$ are $Y_{q_{i-1}^{\rightarrow}}-Y_{a_{i-1}}$ and $Y_{q_i^{\leftarrow}}-Y_{a_{i-1}}$, respectively. We have $$V_i=-(Y_{q_{i-1}^{\rightarrow}}-Y_{a_{i-1}})+(Y_{q_i^{\leftarrow}}-Y_{a_{i-1}}) = Y_{q_i^{\leftarrow}}-Y_{q_{i-1}^{\rightarrow}}=r_i-r_{i-1}=|r_i-r_{i-1}|.$$

\item Case 3. Both $Q_{i-1}^{\rightarrow}$ and $Q_i^{\leftarrow}$ are light, as in Figure~\ref{fig:phi}(c). In this case, $X_{a_{i-1}} > X_{q_{i-1}^{\rightarrow}}$ and $X_{q_i^{\leftarrow}} > X_{a_{i-1}}$. The horizontal distance traveled by $Q_{i-1}^{\rightarrow}$ and $Q_i^{\leftarrow}$ are $X_{a_{i-1}}-X_{q_{i-1}^{\rightarrow}}$ and $X_{q_i^{\leftarrow}}-X_{a_{i-1}}$, respectively. So $$
H_i=(X_{a_{i-1}}-X_{q_{i-1}^{\rightarrow}})+(X_{q_i^{\leftarrow}}-X_{a_{i-1}}) =
X_{q_i^{\leftarrow}}-X_{q_{i-1}^{\rightarrow}}=||o_io_{i-1}||.$$ Same as in case 1, the vertical distance traveled by $Q_{i-1}^{\rightarrow}$ and $Q_i^{\leftarrow}$ are $Y_{q_{i-1}^{\rightarrow}}-Y_{a_{i-1}}$ and $Y_{q_i^{\leftarrow}}-Y_{a_{i-1}}$, respectively. We have $$
V_i=(Y_{q_i^{\leftarrow}}-Y_{a_{i-1}})+(Y_{q_{i-1}^{\rightarrow}}-Y_{a_{i-1}}).$$ Note that $Y_{q_{i-1}^{\rightarrow}}-Y_{a_{i-1}} \geq 0$ and $Y_{q_i^{\leftarrow}}-Y_{a_{i-1}} \geq 0$. If $r_i\geq r_{i-1}$ (i.e., $Y_{q_i^{\leftarrow}} \geq Y_{q_{i-1}^{\rightarrow}}$) then
\begin{align*}
V_i&=(Y_{q_i^{\leftarrow}}-Y_{a_{i-1}})+(Y_{q_{i-1}^{\rightarrow}}-Y_{a_{i-1}})\\
&\geq (Y_{q_i^{\leftarrow}}-Y_{a_{i-1}})-(Y_{q_{i-1}^{\rightarrow}}-Y_{a_{i-1}}) \\
&= Y_{q_i^{\leftarrow}}-Y_{q_{i-1}^{\rightarrow}} \\
&= r_i-r_{i-1} =|r_i-r_{i-1}|.
\end{align*}
If $r_i < r_{i-1}$ (i.e., $Y_{q_i^{\leftarrow}} < Y_{q_{i-1}^{\rightarrow}}$) then
\begin{align*}
V_i&=(Y_{q_i^{\leftarrow}}-Y_{a_{i-1}})+(Y_{q_{i-1}^{\rightarrow}}-Y_{a_{i-1}})\\
&\geq -(Y_{q_i^{\leftarrow}}-Y_{a_{i-1}})+(Y_{q_{i-1}^{\rightarrow}}-Y_{a_{i-1}})\\
&= Y_{q_{i-1}^{\rightarrow}}-Y_{q_i^{\leftarrow}}\\
&= r_{i-1}-r_i =|r_i-r_{i-1}|.
\end{align*}

\item It is impossible for both $Q_{i-1}^{\rightarrow}$ and $Q_i^{\leftarrow}$ to be heavy.
\end{itemize}
So in any case, $H_i = ||o_io_{i-1}||$ and
$V_i \geq |r_i-r_{i-1}|,$ as required for Equality (\ref{hivi1}) and Inequality (\ref{hivi2}).



\subsection{Proof for Inequality~(\ref{gammagamma})}\label{sec:appendix-gamma}

Recall that  $$\gamma^* =
3\beta/4-\alpha/4+\arcsin\left(\frac{\alpha+\beta}{4\lambda\sin(\alpha/4+\beta/4)}\right).$$ Let $\mu=\alpha/4+\beta/4$ and $\nu=\frac{\mu}{\lambda \sin \mu} = \frac{\alpha+\beta}{4\lambda\sin(\alpha/4+\beta/4)}$. So
\begin{align}
\frac{\partial \mu}{\partial \beta}=\frac{\partial (\alpha/4+\beta/4)}{\partial \beta} > 0.\label{zbeta}
\end{align}
By (\ref{abboth}), $0 < \mu < \pi/4$. In this range, $\mu < \tan \mu$ and $\sin \mu> 0$. We have
\begin{align}
\frac{\partial \nu}{\partial \mu} = \frac{\partial \frac{\mu}{\lambda \sin \mu}}{\partial \mu} = (1-\mu/\tan \mu)/(\lambda\sin \mu)> 0.\label{zeta}
\end{align}
This means that $\nu$ is an increasing function of $\mu$. Therefore $\nu=\frac{\mu}{\lambda \sin \mu} \leq \frac{\pi/4}{\lambda \sin(\pi/4)} < 0.618.$ Also $\nu > 0$. In the range $0 < \nu < 0.618$, we have \begin{align}
\frac{\partial \arcsin \nu}{\partial \nu} = 1/\sqrt{1-\nu^2} > 0,\label{arcsinzeta}
\end{align}
By (\ref{zbeta}), (\ref{zeta}), and (\ref{arcsinzeta}),
$$\frac{\partial \gamma^*}{\partial \beta} = 3/4+\frac{\partial \arcsin \nu}{\partial \nu}\cdot \frac{\partial \nu}{\partial \mu}\cdot \frac{\partial \mu}{\partial \beta} > 0.$$
This means $\gamma^*$ is an increasing function of $\beta$, and since $\beta \leq \sin\alpha$, we have  $$\gamma^* \leq \gamma^* |_{\beta = \sin\alpha} =\frac{3\sin\alpha-\alpha}{4}+
\arcsin\left(\frac{\alpha+\sin\alpha}{4\lambda\sin(\frac{\alpha+\sin\alpha}{4})}\right) = \gamma^+,$$ as required by Inequality~(\ref{gammagamma}).

\subsection{Calculations of the Partial Derivatives in Equalities (\ref{pder}), (\ref{phider}) and (\ref{dder})}\label{sec:appendix-calc}

\input{fig-v2-finalsetup-3}

First, let's define some parameters (refer to Figure~\ref{fig:finalsetup3} for an illustration):
\begin{itemize}
\item Let $X_u, Y_u$ and $X_v, Y_v$ be the $x$- and $y$-coordinates of $u$ and $v$, respectively.

\item Let $Y_{a_{n-1}}$ be the $y$-coordinate of $a_{n-1}$.

\item Let $\eta=\beta r_n$.
\end{itemize}

We claim that the parameters $X_u$, $Y_u$, $Y_{a_{n-1}}$ and $\eta$ remain constant during the transformation. Here is why. Firstly, $X_u$ and $Y_u$ remain constant during the transformation because the location of $u$ is not changed by the transformation. Secondly, $Y_{a_{n-1}}$ remains constant during the transformation because $a_{n-1}$ is fixed during the transformation. Thirdly, observe that $\eta = \beta r_n$ is the length of the arc between $v$ and $q$ on the boundary of $O_n$ (the thick arc in Figure~\ref{fig:finalsetup3}). Recall that $A_n$ is the arc on the boundary of $O_n$ between $a_{n-1}$ and $v$, and $B_n$ is the arc on the boundary of $O_n$ between $b_{n-1}$ and $v$. So $\eta=(|A_n|-|B_n|)/2$. During the transformation, $v$ stays pivotal and hence $|P_{\mathcal{O}}^{A_n}(u,v)| = |P_{\mathcal{O}}^{B_n}(u,v)|$, where $|P_{\mathcal{O}}^{A_n}(u,v)| = |P_{\mathcal{O}}(u,a_{n-1})|+|A_n|$ and $|P_{\mathcal{O}}^{B_n}(u,v)| = |P_{\mathcal{O}}(u,b_{n-1})|+|B_n|$. Therefore, $\eta = (|A_n|-|B_n|)/2=(|P_{\mathcal{O}}(u,b_{n-1})| - |P_{\mathcal{O}}(u,a_{n-1})|)/2$. Since the transformation does not affect $O_1,\ldots,O_{n-1}$, clearly $|P_{\mathcal{O}}(u,b_{n-1})|$ and $|P_{\mathcal{O}}(u,a_{n-1})|$ remain constant during the transformation. In other words, $\eta$ remains constant during the transformation.

In what follows, we express all other parameters as functions of ($X_{o_n}$, $\eta$, $X_u$, $Y_u$,  $Y_{a_{n-1}}$) and then calculate their partial derivatives with respect to $X_{o_n}$. This is achievable because, as mentioned above, the parameters $\eta$, $X_u$, $Y_u$, and $Y_{a_{n-1}}$ are all independent of $X_{o_n}$.

Refer to Figure~\ref{fig:finalsetup3}. Let $o$ be the origin of the coordinate system. Consider the triangle $\triangle oo_na_{n-1}$. We have $||o_na_{n-1}||=r_n$, $||oo_n||=|X_{o_n}|$ ($X_{o_n} < 0$ if $o_n$ is to the left of $o$ and $X_{o_n} \geq 0$ otherwise), and $||oa_{n-1}||=Y_{a_{n-1}}$ ($Y_{a_{n-1}}$ is always non-negative since $a_{n-1}$ is above $b_{n-1}$). Therefore
\begin{align}
r_n &= ||o_na_{n-1}|| = \sqrt{||oo_n||^2+||oa_{n-1}||^2} = \sqrt{X_{o_n}^2+Y_{a_{n-1}}^2},\label{r_n_eq}
\end{align}
and hence
\begin{align}
\frac{\partial{r_n}}{\partial{X_{o_n}}} &= \frac{\partial{\sqrt{X_{o_n}^2+Y_{a_{n-1}}^2}}}{\partial{X_{o_n}}} = \frac{2X_{o_n}}{2\sqrt{X_{o_n}^2+Y_{a_{n-1}}^2}} = \frac{X_{o_n}}{r_n}= -\cos\alpha.\label{r_n_der}
\end{align}
The last equality is based on the geometric observation that $\cos{\alpha}=-\frac{X_{o_n}}{r_n}$ (refer to Figure~\ref{fig:finalsetup3}).

Since $\alpha = \pi/2+\angle oa_{n-1}o_n$ and $\angle oa_{n-1}o_n = \arctan(\frac{X_{o_n}}{Y_{a_{n-1}}})$ (note that $\angle oa_{n-1}o_n \geq 0$ if and only if $X_{o_n} \geq 0$), we have
\begin{align}
\alpha &= \pi/2+\arctan(\frac{X_{o_n}}{Y_{a_{n-1}}}),\label{alpha_eq}
\end{align}
and hence
\begin{align}
\frac{\partial{\alpha}}{\partial{X_{o_n}}} &= \frac{\partial{(\pi/2+\arctan(\frac{X_{o_n}}{Y_{a_{n-1}}})})}{\partial{X_{o_n}}} = \frac{Y_{a_{n-1}}^2}{X_{o_n}^2+Y_{a_{n-1}}^2}\cdot \frac{1}{Y_{a_{n-1}}} = \frac{Y_{a_{n-1}}}{X_{o_n}^2+Y_{a_{n-1}}^2} = \frac{Y_{a_{n-1}}}{r_n^2}= \frac{\sin\alpha}{r_n}.\label{alpha_der}
\end{align} The last equality is based on the geometric observation that $\sin{\alpha}=\frac{Y_{a_{n-1}}}{r_n}$ (refer to Figure~\ref{fig:finalsetup3}).

Since $\eta=\beta r_n$, we have
\begin{align}
\beta &= \eta/r_n,\label{beta_eq}\end{align}
and hence, from (\ref{r_n_der}), we have
\begin{align}
\frac{\partial{\beta}}{\partial{X_{o_n}}} &= \frac{\partial{(\eta/r_n)}}{\partial{X_{o_n}}} = -\frac{\eta}{r_n^2}\cdot\frac{\partial{r_n}}{\partial{X_{o_n}}} = -\frac{\beta}{r_n}\cdot\frac{\partial{r_n}}{\partial{X_{o_n}}} = \frac{\beta\cos\alpha}{r_n}.\label{beta_der}
\end{align}

Refer to Figure~\ref{fig:finalsetup3}. Let $v_x$ be the projection of $v$ on the x-axis. Consider the triangle $\triangle o_nvv_x$. We have $X_{v} - X_{o_n}= ||o_nv||\cos\beta = r_n\cos\beta$. Therefore
\begin{align}
X_{v} &= X_{o_n}+r_n\cos\beta,\label{x_v_eq}
\end{align}
and hence, from (\ref{r_n_der}) and (\ref{beta_der}), we have
\begin{align}
\frac{\partial{X_{v}}}{\partial{X_{o_n}}} &=  \frac{\partial{(X_{o_n}+r_n\cos\beta)}}{\partial{X_{o_n}}}= 1-r_n\sin\beta\frac{\partial{\beta}}{\partial{X_{o_n}}}+\cos\beta\frac{\partial{r_n}}{\partial{X_{o_n}}} \nonumber \\
&= 1-\beta\cos\alpha\sin\beta-\cos\alpha\cos\beta.\label{x_v_der}
\end{align}

Again consider the triangle $\triangle o_nvv_x$. We have
\begin{align}
Y_{v} &= -||o_nv||\sin\beta = -r_n\sin\beta,\label{y_v_eq}
\end{align}
and hence, from (\ref{r_n_der}) and (\ref{beta_der}), we have
\begin{align}
\frac{\partial{Y_{v}}}{\partial{X_{o_n}}} &=  \frac{\partial{(-r_n\sin\beta)}}{\partial{X_{o_n}}} = -r_n\cos\beta\frac{\partial{\beta}}{\partial{X_{o_n}}}-\sin\beta\frac{\partial{r_n}}{\partial{X_{o_n}}} \nonumber \\
&= -\beta\cos\alpha\cos\beta+\cos\alpha\sin\beta.\label{y_v_der}
\end{align}

Since $A_n$ is the arc on the boundary of $O_n$ between $a_{n-1}$ and $v$, we have
\begin{align}
|A_n| = (\alpha+\beta)r_n = \alpha r_n+\beta r_n=\alpha r_n+\eta,\label{a_n_eq}
\end{align}
and hence, from (\ref{r_n_der}) and (\ref{alpha_der}), we have \begin{align} \frac{\partial{|A_{n}|}}{\partial{X_{o_n}}}
=\frac{\partial{(\alpha r_n+\eta)}}{\partial{X_{o_n}}}
=\frac{\partial{(\alpha r_n)}}{\partial{X_{o_n}}}
= \alpha\frac{\partial{r_n}}{\partial{X_{o_n}}}+ r_n\frac{\partial{\alpha}}{\partial{X_{o_n}}} = \sin\alpha-\alpha\cos\alpha.\label{alphar_n_der}
\end{align}

Next we will express $\frac{\partial{\Upsilon_{\mathcal{O}}(u,v)}}{\partial{X_{o_n}}}$ as a function of $\alpha, \beta$ and $\gamma$.

By the definition of $\Upsilon_{\mathcal{O}}(u,v)$, see (\ref{defupsilon}), we have
\begin{align}
\frac{\partial{\Upsilon_{\mathcal{O}}(u,v)}}{\partial{X_{o_n}}} &= \frac{\partial{(|P_{\mathcal{O}}(u,v)| - \lambda|D_{\mathcal{O}}(u,v)|+\Phi_{\mathcal{O}})}}{\partial{X_{o_n}}} = \frac{\partial{|P_{\mathcal{O}}(u,v)|}}{\partial{X_{o_n}}} - \lambda\frac{\partial{|D_{\mathcal{O}}(u,v)|}}{\partial{X_{o_n}}}
+\frac{\partial{|\Phi_{\mathcal{O}}|}}{\partial{X_{o_n}}}\label{upsilon_der}
\end{align}

Since $v$ is the pivotal point, we have $|P_{\mathcal{O}}(u,v)| = |P_{\mathcal{O}}(u,a_{n-1})|+|A_{n}|$, where  $|P_{\mathcal{O}}(u,a_{n-1})|$ is independent of $X_{o_n}$ because the transformation does not affect $O_1,\ldots,O_{n-1}$. Hence from (\ref{alphar_n_der}),
\begin{align}
\frac{\partial{|P_{\mathcal{O}}(u,v)|}}{\partial{X_{o_n}}}=
\frac{\partial{(P_{\mathcal{O}}(u,a_{n-1})+|A_{n}|)}}{\partial{X_{o_n}}}
=\frac{\partial{|A_{n}|}}{\partial{X_{o_n}}}
= \sin\alpha-\alpha\cos\alpha,\label{pder-appendix}
\end{align}
which gives Equality (\ref{pder}).\\
\\

From (\ref{phi_def}),
\begin{align}
\frac{\partial{\Phi_{\mathcal{O}}}}{\partial{X_{o_n}}} &= \frac{\partial{(\varphi (r_{n}-r_1)-\frac{\varphi}{3}\sum_{i=2}^{n-1}(2H_i+V_i))}}{\partial{X_{o_n}}}.
\end{align}
Since the transformation does not affect $O_1,\ldots,O_{n-1}$, the variables $r_1$, $H_1,\ldots,H_{n-1}$, $V_1,\ldots,V_{n-1}$ are independent of $X_{o_n}$. We have
\begin{align}
\frac{\partial{\Phi_{\mathcal{O}}}}{\partial{X_{o_n}}} &= \frac{\partial{(\varphi (r_{n}-r_1)-\frac{\varphi}{3}\sum_{i=2}^{n-1}(2H_i+V_i))}}{\partial{X_{o_n}}}\nonumber\\
&=\varphi\frac{\partial{r_n}}{\partial{X_{o_n}}}-\frac{\varphi}{3}(2\frac{\partial{H_n}}{\partial{X_{o_n}}}+\frac{\partial{V_n}}{\partial{X_{o_n}}})
\end{align}

From (\ref{hivi1}) in Section~\ref{sec:outline}, $H_n=||o_no_{n-1}||=X_{o_n}-X_{o_{n-1}}$. Since $X_{o_{n-1}}$ is constant during the transformation,
\begin{align}
\frac{\partial{H_n}}{\partial{X_{o_n}}} =\frac{\partial{(X_{o_n}-X_{o_{n-1}})}}{\partial{X_{o_n}}} =\frac{\partial{X_{o_n}}}{\partial{X_{o_n}}}= 1.\label{h_der}
\end{align}

Now let's consider $V_n$, which is the vertical distance traveled by $Q_{n-1}^{\rightarrow}$ and $Q_n^{\leftarrow}$ with light arcs contributing positively and heavy arcs contributing negatively (refer to Figure~\ref{fig:phi}).
The vertical distance traveled by $Q_{n}^{\leftarrow}$ is $r_n-Y_{a_{n-1}}$ and the vertical distance traveled by $Q_{n-1}^{\rightarrow}$ is $r_{n-1}-Y_{a_{n-1}}$.

When $0< \alpha <\pi/2$ (as in case a of Figure~\ref{fig:phi}), $Q_{n-1}^{\rightarrow}$ is light and $Q_n^{\leftarrow}$ is heavy. So $V_n=(r_{n-1}-Y_{a_{n-1}})-(r_n-Y_{a_{n-1}})=r_{n-1}-r_n$. Since $r_{n-1}$ is constant during the transformation,
\begin{align}
\frac{\partial{V_n}}{\partial{X_{o_n}}} = \frac{\partial{(r_{n-1}-r_n)}}{\partial{X_{o_n}}}=-\frac{\partial{r_n}}{\partial{X_{o_n}}}=\cos\alpha.\label{v_n_der1}
\end{align}

When $\pi/2 \leq \alpha <\pi$, there are two possibilities: (1) $Q_n^{\leftarrow}$ is light and $Q_{n-1}^{\rightarrow}$ is heavy (as in case b of Figure~\ref{fig:phi}), and $V_n=-(r_{n-1}-Y_{a_{n-1}})+(r_n-Y_{a_{n-1}}) = r_n - r_{n-1}$. Since $r_{n-1}$ is constant during the transformation,
\begin{align}
\frac{\partial{V_n}}{\partial{X_{o_n}}} = \frac{\partial{(r_n - r_{n-1})}}{\partial{X_{o_n}}}=\frac{\partial{r_n}}{\partial{X_{o_n}}}=-\cos\alpha.\label{v_n_der2}
\end{align}
(2) $Q_n^{\leftarrow}$ and $Q_{n-1}^{\rightarrow}$ are both light (as in case c of Figure~\ref{fig:phi}), and $V_n=(r_{n-1}-Y_{a_{n-1}})+(r_n-Y_{a_{n-1}}) = r_n+r_{n-1}-2Y_{a_{n-1}}$. Since $r_{n-1}$ and $Y_{a_{n-1}}$ are constant during the transformation,
\begin{align}
\frac{\partial{V_n}}{\partial{X_{o_n}}} = \frac{\partial{(r_n+r_{n-1}-2Y_{a_{n-1}})}}{\partial{X_{o_n}}}=\frac{\partial{r_n}}{\partial{X_{o_n}}}=-\cos\alpha.\label{v_n_der3}
\end{align}
Therefore, combining (\ref{r_n_der}), (\ref{h_der}), (\ref{v_n_der1}), (\ref{v_n_der2}), and (\ref{v_n_der3}), we have
\begin{align}
\frac{\partial{\Phi_{\mathcal{O}}}}{\partial{X_{o_n}}} &= \frac{\partial{(\varphi (r_{n}-r_1)-\frac{\varphi}{3}\sum_{i=2}^{n-1}(2H_i+V_i))}}{\partial{X_{o_n}}}\nonumber\\
&=\varphi\frac{\partial{r_n}}{\partial{X_{o_n}}}-\frac{\varphi}{3}(2\frac{\partial{H_n}}{\partial{X_{o_n}}}+\frac{\partial{V_n}}{\partial{X_{o_n}}})\nonumber\\
&= \begin{cases} -\varphi\cos\alpha-\frac{2\varphi}{3}-\frac{\varphi}{3}\cos\alpha, & \mbox{if } 0< \alpha <\pi/2 \\ -\varphi\cos\alpha-\frac{2\varphi}{3}+\frac{\varphi}{3}\cos\alpha, & \mbox{if } \pi/2 \leq \alpha <\pi \end{cases}\nonumber\\
&= \begin{cases} -\frac{2\varphi}{3}-\frac{4\varphi}{3}\cos\alpha, & \mbox{if } 0< \alpha <\pi/2, \\ -\frac{2\varphi}{3}-\frac{2\varphi}{3}\cos\alpha, & \mbox{if } \pi/2 \leq \alpha <\pi, \end{cases}\label{phider-appendix}
\end{align}
which gives Equality (\ref{phider}).\\
\\

For the case under consideration, $|D_{\mathcal{O}}(u,v)| = ||uv||$ and hence
\begin{align}
\frac{\partial{|D_{\mathcal{O}}(u,v)|}}{\partial{X_{o_n}}} &= \frac{\partial{||uv||}}{\partial{X_{o_n}}} = \frac{\partial{\sqrt{(X_{v}-X_u)^2+(Y_{v}-Y_u)^2}}}{\partial{X_{o_n}}}\nonumber\\
& = \frac{X_{v}-X_u}{\sqrt{(X_{v}-X_u)^2+(Y_{v}-Y_u)^2}}\cdot\frac{\partial{(X_{v}-X_u)}}{\partial{X_{o_n}}}+\frac{Y_{v}-Y_u}{\sqrt{(X_{v}-X_u)^2+(Y_{v}-Y_u)^2}}\cdot\frac{\partial{(Y_{v}-Y_u)}}{\partial{X_{o_n}}} \nonumber\\
& = \frac{X_{v}-X_u}{||uv||}\cdot\frac{\partial{(X_{v}-X_u)}}{\partial{X_{o_n}}}+\frac{Y_{v}-Y_u}{||uv||}\cdot\frac{\partial{(Y_{v}-Y_u)}}{\partial{X_{o_n}}} \nonumber\\
& = \cos{\gamma}\frac{\partial{X_{v}}}{\partial{X_{o_n}}}-\sin{\gamma}\frac{\partial{Y_{v}}}{\partial{X_{o_n}}}. \label{theta-angle}
\end{align}
The last equality is true because $\cos{\gamma} = \frac{X_v-X_u}{||uv||}$, $-\sin{\gamma} = \frac{Y_v-Y_u}{||uv||}$, and $X_u,Y_u$ are independent of $X_{o_n}$.

Combining (\ref{x_v_der}), (\ref{y_v_der}), and (\ref{theta-angle}) and by the trigonometric identities, we have
\begin{align}
\frac{\partial{|D_{\mathcal{O}}(u,v)|}}{\partial{X_{o_n}}} &= \cos{\gamma}\frac{\partial{X_{v}}}{\partial{X_{o_n}}}-\sin{\gamma}\frac{\partial{Y_{v}}}{\partial{X_{o_n}}}\\
& = \cos{\gamma}(1-\beta\cos\alpha\sin\beta-\cos\alpha\cos\beta)
-\sin{\gamma}(-\beta\cos\alpha\cos\beta+\cos\alpha\sin\beta) \nonumber\\
& = \cos{\gamma}-\cos\alpha(\beta\sin\beta\cos\gamma-\beta\cos\beta\sin\gamma
+\cos\beta\cos\gamma+\sin\beta\sin\gamma) \nonumber\\
& = \cos{\gamma}-\cos\alpha(\cos(\beta-\gamma)+\beta\sin(\beta-\gamma)),\label{dder-appendix}
\end{align}
which gives Equality (\ref{dder}).

\subsection{Proof for Inequality (\ref{gggbbb})}\label{sec:appendix-inequ-f}

Recall that by definition \begin{align}
f(\alpha, \beta,\gamma)&= -\lambda\frac{\partial{|D_{\mathcal{O}}(u,v)|}}{\partial{X_{o_n}}}\\
&= -\lambda(\cos{\gamma}-\cos\alpha(\cos(\beta-\gamma)+\beta\sin(\beta-\gamma))),
\end{align} where $\lambda=1.8$ is a constant.

We will first analyze the occurrence of the maximum value of $f(\alpha, \beta,\gamma)$ with respect to $\gamma$.
Let $v_2$ be the location of $v$ when $X_{o_n}$ increases by $\partial{X_{o_n}}$ (i.e., $o_n$ moves to the right by $\partial{X_{o_n}}$)\footnote{Note that here $o_n$ is moving in the opposite direction to the transformation of $O_n$ defined in the beginning of the proof for Proposition~\ref{last}. While it is necessary for the correctness of the proof to move $o_n$ towards $o_{n-1}$ when defining the transformation of $O_n$, here we move $o_n$ away from $o_{n-1}$ to be consistent with the sign of $\partial{X_{o_n}}$ for the sake of notational convenience.}. Let $\partial\ell$ be the distance from $v$ to $v_2$. Let $\omega$ be the angle from the $x$-axis to $\overrightarrow{vv_2}$. See Figure~\ref{fig:finalsetup4} for an illustration\footnote{In Figure~\ref{fig:finalsetup4}, both $\partial\ell$ and $\omega$ are exaggerated for the illustrative purpose.}. So

\input{fig-v2-finalsetup-4}

\begin{align}
\cos\omega=\frac{\partial X_v}{\partial\ell}, \label{cosomega}
\end{align}
and \begin{align}
\sin\omega=\frac{\partial Y_v}{\partial\ell}, \label{sinomega}
\end{align}
where $X_{v}$ and $Y_{v}$ are the $x$- and $y$-coordinates of $v$. Recall that $\gamma$ is the angle from $\overrightarrow{uv}$ to the $x$-axis. Let $\theta=\omega+\gamma$.

From (\ref{theta-angle}), (\ref{cosomega}) and (\ref{sinomega}) we have \begin{align}
\frac{\partial{|D_{\mathcal{O}}(u,v)|}}{\partial{X_{o_n}}}&=
\cos{\gamma}\frac{\partial{X_{v}}}{\partial{X_{o_n}}}-\sin{\gamma}\frac{\partial{Y_{v}}}{\partial{X_{o_n}}} \nonumber\\
&=(\cos\gamma\frac{\partial X_v}{\partial\ell}-\sin\gamma\frac{\partial Y_v}{\partial\ell})\frac{\partial{\ell}}{\partial{X_{o_n}}}\nonumber\\
&=(\cos\gamma\cos\omega-\sin\gamma\sin\omega)\frac{\partial{\ell}}{\partial{X_{o_n}}}\nonumber\\
&=
\cos(\omega+\gamma)\frac{\partial{\ell}}{\partial{X_{o_n}}}\nonumber\\
&= \cos\theta\frac{\partial{\ell}}{\partial{X_{o_n}}}.\label{observ}\end{align} Therefore, from (\ref{fdeforig}) and (\ref{observ}), $f=-\lambda\frac{\partial{|D_{\mathcal{O}}(u,v)|}}{\partial{X_{o_n}}}=-\lambda\cos\theta\frac{\partial\ell}{\partial{X_{o_n}}}$. Recall that $\partial\ell$ was defined to be $||vv_2||$ where $v_2$ is the location of $v$ when $X_{o_n}$ increases by $\partial{X_{o_n}}$. Hence $\frac{\partial\ell}{\partial{X_{o_n}}}>0$. Also note that $\frac{\partial\ell}{\partial{X_{o_n}}}$ is determined by $\alpha$ and $\beta$. In other words, $\frac{\partial\ell}{\partial{X_{o_n}}}$ is independent of $\gamma$. Therefore, with $\alpha$ and $\beta$ fixed, $f(\alpha, \beta,\gamma)$ is maximized when $\cos\theta$ is minimized, i.e., when $\theta$ is minimized or maximized for $\theta\in [-\pi, \pi]$. Since $\theta=\omega+\gamma$ where $\omega$ is independent of $\gamma$, $\theta$ is minimized when $\gamma = 0$ and is maximized when $\gamma=\gamma^+$. Therefore the maximum of $f$ occurs when $\gamma = 0$ or when $\gamma=\gamma^+$. In other words,
\begin{align}
f(\alpha, \beta,\gamma) \leq \max\{f(\alpha,\beta,0), f(\alpha,\beta,\gamma^+)\},\label{ggg}
\end{align} where $\gamma^+=\frac{3\sin\alpha-\alpha}{4}+\arcsin\left(\frac{\alpha+\sin\alpha}{4\lambda\sin(\frac{\alpha+\sin\alpha}{4})}\right).$ 

We then analyze the occurrence of the maximum value of $f(\alpha, \beta,\gamma)$ with respect to $\beta$. We have
\begin{align*}
\frac{\partial{f}}{\partial{\beta}} &= \frac{\partial{(-\lambda(\cos{\gamma}-\cos\alpha(\cos(\beta-\gamma)+\beta\sin(\beta-\gamma))))}}{\partial{\beta}} \\ &= -\lambda\frac{\partial{(\cos{\gamma}-\cos\alpha(\cos(\beta-\gamma)+\beta\sin(\beta-\gamma)))}}{\partial{\beta}} \\ &= \lambda\cos\alpha\frac{\partial{(\cos(\beta-\gamma)+\beta\sin(\beta-\gamma))}}{\partial{\beta}} \\ &= \lambda\cos\alpha(-\sin(\beta-\gamma)+\sin(\beta-\gamma)+\beta\cos(\beta-\gamma)) \\
&=\lambda\beta\cos\alpha\cos(\beta-\gamma),
\end{align*}
where $\beta-\gamma = \angle o_nvu$. Since $v$ is the exit point of $\overrightarrow{uv}$ on the boundary of $O_n$, The angle $\angle o_nvu$ is in the range $[-\pi/2, \pi/2]$. In other words, $-\pi/2 \leq \beta-\gamma \leq \pi/2$ and hence  $\cos(\beta-\gamma) \geq 0$. We distinguish two cases:
\begin{enumerate}
\item $\pi/2 \leq \alpha < \pi$. In this case, $\cos\alpha\leq 0$. Since $\cos(\beta-\gamma) \geq 0$, we have $
    \frac{\partial{f}}{\partial{\beta}} = \lambda\beta\cos\alpha\cos(\beta-\gamma)
    \leq 0$ when $\beta \geq 0$, and $
    \frac{\partial{f}}{\partial{\beta}} \geq 0$ when $\beta \leq 0.$ This means that $f$ is a non-decreasing function of $\beta$ when $\beta \leq 0$ and a non-increasing function of $\beta$ when $\beta \geq 0.$ So $f$ obtains its maximum value when $\beta = 0$, i.e.,
    \begin{align}
    f(\alpha,\beta,\gamma)\leq f(\alpha,0,\gamma).\label{ge}
    \end{align}

\item $0 < \alpha < \pi/2$. In this case, $\cos\alpha\geq 0$. Since $\cos(\beta-\gamma) \geq 0$, we have $
    \frac{\partial{f}}{\partial{\beta}} = \lambda\beta\cos\alpha\cos(\beta-\gamma) \geq 0$ when $\beta \geq 0,$ and $
    \frac{\partial{f}}{\partial{\beta}} \leq 0$ when $\beta \leq 0.$ From (\ref{fdef}), we have
    \begin{align}
    f(\alpha,\beta,\gamma) - f(\alpha,-\beta,\gamma) = &
    -\lambda(\cos{\gamma}-\cos\alpha(\cos(\beta-\gamma)+\beta\sin(\beta-\gamma))) \nonumber \\ & +\lambda(\cos{\gamma}-\cos\alpha(\cos(-\beta-\gamma)-\beta\sin(-\beta-\gamma)))\nonumber\\
    = & \lambda\cos\alpha[(\cos(\beta-\gamma)+\beta\sin(\beta-\gamma))-(\cos(-\beta-\gamma)-\beta\sin(-\beta-\gamma))]\nonumber\\
    = & \lambda\cos\alpha[\cos(\beta-\gamma)+\beta\sin(\beta-\gamma)-\cos(\beta+\gamma)-\beta\sin(\beta+\gamma)]\nonumber\\
    = & \lambda\cos\alpha[\cos(\beta-\gamma)-\cos(\beta+\gamma)+\beta\sin(\beta-\gamma)-\beta\sin(\beta+\gamma)]\nonumber\\
    = & \lambda\cos\alpha[-2\sin\beta\sin(-\gamma)+2\beta\cos\beta\sin(-\gamma)] \mbox{~~~~(by trigonometric identities)}\nonumber\\
    = & \lambda\cos\alpha[2\sin\beta\sin\gamma-2\beta\cos\beta\sin\gamma] \nonumber\\
    = & 2\lambda\cos\alpha\sin\gamma(\sin\beta-\beta\cos\beta).
    \end{align}

    Recall that $\cos\alpha\geq 0$. By (\ref{range_c}), $0\leq \gamma < \pi/2$ and hence
    $\sin\gamma \geq 0$. By (\ref{range_b}), $\beta \leq \sin\alpha \leq 1$. If $\beta \geq 0$, then $0 \leq \beta \leq 1$ and within this range, one verifies that $\sin\beta-\beta\cos\beta \geq 0$. So when $\beta \geq 0$, we have $f(\alpha,\beta,\gamma) - f(\alpha,-\beta,\gamma) \geq 0$. In other words,
    \begin{align}
    f(\alpha,\beta,\gamma) \leq f(\alpha,|\beta|,\gamma).\label{f_1}
    \end{align}
    As mentioned in the beginning of this case, we have $\frac{\partial{f}}{\partial{\beta}} \geq 0$ for $\beta \geq 0$, and from (\ref{range_b}), $0 \leq |\beta| \leq \sin\alpha$. Hence, we have \begin{align}f(\alpha,|\beta|,\gamma) \leq f(\alpha,\sin\alpha,\gamma).\label{f_2}\end{align} From (\ref{f_1}) and (\ref{f_2}),
    \begin{align}
    f(\alpha,\beta,\gamma) \leq f(\alpha,\sin\alpha,\gamma).\label{le}
    \end{align}

\end{enumerate}

Combining  (\ref{ggg}) with (\ref{ge}) and (\ref{le}), we have
\begin{align}
f(\alpha,\beta,\gamma)\leq \max\{f(\alpha,\beta,0), f(\alpha,\beta,\gamma^+)\} \leq \begin{cases}  \max\{f(\alpha,\sin\alpha,0), f(\alpha,\sin\alpha,\gamma^+)\}, & \mbox{if } 0< \alpha <\pi/2, \\ \max\{f(\alpha,0,0), f(\alpha,0,\gamma^+)\}, & \mbox{if } \pi/2 \leq \alpha <\pi, \end{cases}\label{gggbbb-appendix}
\end{align}
which yields Inequality~(\ref{gggbbb}).

\subsection{Verify Inequalities (\ref{ineq1}), (\ref{ineq2}), (\ref{ineq3}), and (\ref{ineq4})}\label{sec:appendix-all-ineqs}

Let $z=(\alpha+\sin\alpha)/4$. Since $0 < z < \pi/4$, we have $1< \frac{z}{\sin z}< 1.111$. Furthermore,
$$\frac{dz}{d\alpha} = \frac{d((\alpha+\sin\alpha)/4)}{d\alpha} =1/4+\cos\alpha/4,$$ and
$$\frac{d(\frac{z}{\sin z})}{dz} = \frac{1}{\sin z}-\frac{z\cos z}{(\sin z)^2} \leq \frac{1}{\sin z}-\frac{\cos z}{\sin z} = \tan(z/2) < 0.415.$$ Recall that $$ \gamma^+=\frac{3\sin\alpha-\alpha}{4}+\arcsin\left(\frac{\alpha+\sin\alpha}{4\lambda\sin(\frac{\alpha+\sin\alpha}{4})}\right)= \frac{3\sin\alpha-\alpha}{4}+\arcsin(\frac{z}{\lambda\sin z}).$$ So \begin{align*}
\left|\frac{d(\gamma^+)}{d\alpha}\right|&\leq \left|\frac{3\cos\alpha-1}{4}+\frac{1}{\lambda\sqrt{1-(\frac{z}{\lambda\sin z})^2}} \frac{d(\frac{z}{\sin z})}{dz}\frac{dz}{d\alpha} \right| \\
&\leq 1+\left|\frac{0.415}{\lambda\sqrt{1-(\frac{z}{\lambda\sin z})^2}} (1/4+\cos\alpha/4)\right| \\ &\leq 1+\left|\frac{0.415}{\lambda\sqrt{1-(\frac{1.111}{\lambda})^2}} (1/4+1/4)\right| \\ &< 1.15.\end{align*}

Recall that
\begin{align}
f(\alpha, \beta,\gamma) =-\lambda(\cos{\gamma}-\cos\alpha(\cos(\beta-\gamma)+\beta\sin(\beta-\gamma))).\end{align}

We have
\begin{align*}
\left|\frac{d f(\alpha,0,0)}{d\alpha}\right| = \left|\frac{d (-\lambda(1-\cos\alpha))}{d\alpha}\right|
=|-\lambda\sin\alpha|
< 2.
\end{align*}

\begin{align*}
\left|\frac{d f(\alpha,0,\gamma^+)}{d\alpha}\right| &= \left|\frac{d (-\lambda(\cos{\gamma^+}-\cos\alpha\cos(-\gamma^+)))}{d\alpha}\right| \\
 &= \left|\frac{d (-\lambda\cos{\gamma^+}(1-\cos\alpha))}{d\alpha}\right| \\
&\leq  \lambda(|\frac{d(\gamma^+)}{d\alpha}\sin\gamma^+(1-\cos\alpha)|+|\sin\alpha\cos\gamma^+|) \\
&\leq  1.8(1.15+1) < 4.
\end{align*}

\begin{align*}
\left|\frac{d f(\alpha,\sin\alpha,0)}{d\alpha}\right| &= \left|\frac{d (-\lambda(1-\cos\alpha\cos(\sin\alpha)-\cos\alpha\sin\alpha\sin(\sin\alpha)))}{d\alpha}\right| \\
&\leq  \lambda(|\sin\alpha\cos(\sin\alpha)|+|(\cos\alpha)^2\sin\alpha\cos(\sin\alpha)| +|(\sin\alpha)^2\sin(\sin\alpha)|)\\
&\leq  1.8(1+1+1) < 6.
\end{align*}

\begin{align*}
\left|\frac{d f(\alpha,\sin\alpha,\gamma^+)}{d\alpha}\right| &= \left|\frac{d (-\lambda(\cos\gamma^+-\cos\alpha\cos(\sin\alpha-\gamma^+)-\cos\alpha\sin\alpha\sin(\sin\alpha-\gamma^+)))}{d\alpha}\right| \\
&=\left|\frac{d (-\lambda(\cos\gamma^+-\cos\alpha\cos(\sin\alpha-\gamma^+)-\frac{\sin(2\alpha)}{2}\sin(\sin\alpha-\gamma^+)))}{d\alpha}\right| \\
&\leq  \lambda(|\frac{d(\gamma^+)}{d\alpha}\sin\gamma^+|+|\sin\alpha\cos(\sin\alpha-\gamma^+)|+|\cos\alpha\sin(\sin\alpha-\gamma^+)(\cos\alpha-\frac{d(\gamma^+)}{d\alpha})|\\
&~~~+|\cos(2\alpha)\sin(\sin\alpha-\gamma^+)|+|\frac{\sin(2\alpha)}{2}\cos(\sin\alpha-\gamma^+)(\cos\alpha-\frac{d(\gamma^+)}{d\alpha})|)\\
&<  1.8(1.15+1+(1+1.15)+1+(1+1.15)/2) < 12.
\end{align*}

Also
\begin{align*}
\left|\frac{d(\sin\alpha-\alpha\cos\alpha -\frac{2\varphi}{3}-\frac{2\varphi}{3}\cos\alpha)}{d\alpha}\right| =|\alpha\sin\alpha+\frac{2\varphi}{3}\sin\alpha|
\leq  \pi+\frac{2\varphi}{3} < 4.\\
\left|\frac{d(\sin\alpha-\alpha\cos\alpha -\frac{2\varphi}{3}-\frac{4\varphi}{3}\cos\alpha)}{d\alpha}\right| =|\alpha\sin\alpha+\frac{4\varphi}{3}\sin\alpha|
\leq  \pi+\frac{4\varphi}{3} < 4.
\end{align*}

Therefore the Lipschitz constants are \begin{align*}
\left|\frac{dg_1(\alpha)}{d\alpha}\right| &= \left|\frac{d(\sin\alpha-\alpha\cos\alpha -\frac{2\varphi}{3}-\frac{2\varphi}{3}\cos\alpha)}{d\alpha}\right|+ \left|\frac{d f(\alpha,0,0)}{d\alpha}\right| < 4+2 < 16.\\
\left|\frac{dg_2(\alpha)}{d\alpha}\right| &= \left|\frac{d(\sin\alpha-\alpha\cos\alpha -\frac{2\varphi}{3}-\frac{2\varphi}{3}\cos\alpha)}{d\alpha}\right|+ \left|\frac{d f(\alpha,0,\gamma^+)}{d\alpha}\right| < 4+4 < 16.\\
\left|\frac{dg_3(\alpha)}{d\alpha}\right| &= \left|\frac{d(\sin\alpha-\alpha\cos\alpha -\frac{2\varphi}{3}-\frac{2\varphi}{3}\cos\alpha)}{d\alpha}\right|+ \left|\frac{d f(\alpha,\sin\alpha,0)}{d\alpha}\right| < 4+6 < 16.\\
\left|\frac{dg_4(\alpha)}{d\alpha}\right| &= \left|\frac{d(\sin\alpha-\alpha\cos\alpha -\frac{2\varphi}{3}-\frac{2\varphi}{3}\cos\alpha)}{d\alpha}\right|+ \left|\frac{d f(\alpha,\sin\alpha,\gamma^+)}{d\alpha}\right| < 4+12 = 16.
\end{align*}

Now we can use a simplified version of the Piyavskii's algorithm~\cite{Vanderbei97extensionof} for Lipschitz optimization to verify that $g_i(\alpha) < 0$, $1\leq i \leq 4$. The following Algorithm {\bf Bound}$(g_i,s,t)$ will either find a value $g_i(\alpha) \geq 0$ in the given range or return an upper bound on the value of $g_i$ in range $[s,t]$ that is less than 0. For $1\leq i \leq 4$, we run {\bf Bound}$(g_i,s,t)$ with $s$ and $t$ set to be the appropriate lower and upper bound on the range of $\alpha$, where $g_i(0)$ is set to be $\lim_{\alpha \rightarrow 0} g_i(\alpha)$. We verify that indeed $g_i(\alpha) < 0$, $1\leq i \leq 4$. This completes the verification of inequalities (\ref{ineq1}), (\ref{ineq2}), (\ref{ineq3}), and (\ref{ineq4}).

\begin{algorithm}{The Algorithm {\bf Bound}($g_i, s, t$)}{alg}
\algtitle{{\bf Bound}($g_i, s, t$), for $1\leq i \leq 4$}
\begin{codeblock}
\step {\bf if} $g_i(s) \geq 0$ or $g_i(t) \geq 0$ {\bf return} $\max\{g_i(s),g_i(t)\}$.
\step apex = $\max\{g_i(s),g_i(t)\}+16*(t-s)/2$
\step {\bf if} apex $\geq 0$, {\bf do}:
    \begin{enumerate}
    \step apex = $\max\{\mbox{\bf Bound}(g_i, s, (s+t)/2), \mbox{\bf Bound}(g_i, (s+t)/2,t)\}$
    \end{enumerate}
\step {\bf return} apex
\end{codeblock}
\end{algorithm}


\end{document}

%% file: fig-v2-chain-only.tex
\begin{figure}[tbph]
\psset{unit=0.8pt}
\begin{center} \small
\begin{pspicture}(-120,-130)(150,120)
    \psset{labelsep=6pt}
\newrgbcolor{lightyellow}{1 1 0.8}


\pscircle[fillstyle=solid,fillcolor=lightyellow,linestyle=none](-107,4){35}

    \psarc[linewidth=2pt](-107,4){35}{10}{302}

\pscircle[fillstyle=solid,fillcolor=lightyellow,linestyle=none](-27,-31){61}
    \psarc[linewidth=2pt](-27,-31){61}{101}{138}
    \psarc[linewidth=2pt](-27,-31){61}{175}{372}

\pscircle[fillstyle=solid,fillcolor=lightyellow,linestyle=none](21,41){61}
    \psarc[linewidth=2pt](21,41){61}{20}{192}

\pscircle[fillstyle=solid,fillcolor=lightyellow,linestyle=none](68,14){49}
    \psarc[linewidth=2pt](68,14){49}{59}{77}
    \psarc[linewidth=2pt](68,14){49}{224}{280}

\pscircle[fillstyle=solid,fillcolor=lightyellow,linestyle=none](123,4){60}
    \psarc[linewidth=2pt](123,4){60}{-140}{120}

    \uput[90](-110,35){$O_1$}

    \uput[90](-60,20){$O_2$}

    \uput[90](0,95){$O_3$}

    \uput[90](90,55){$O_4$}

    \uput[90](150,55){$O_5$}

    \psarc[linecolor=blue,linewidth=0.5pt](-107,4){35}{302}{10}
    \psarc[linecolor=blue,linewidth=0.5pt](-27,-31){61}{372}{101}
    \psarc[linecolor=blue,linewidth=0.5pt](-27,-31){61}{138}{175}

    \psarc[linecolor=blue,linewidth=0.5pt](21,41){61}{192}{20}

    \psarc[linecolor=blue,linewidth=0.5pt](68,14){49}{280}{59}
    \psarc[linecolor=blue,linewidth=0.5pt](68,14){49}{77}{224}

    \psarc[linecolor=blue,linewidth=0.5pt](123,4){60}{120}{-140}

\uput[-90](0,-100){(a) Example 1.}
\end{pspicture}

\psset{unit=1.4pt}
\begin{pspicture}(-120,-70)(150,90)
    \psset{labelsep=6pt}

\pscircle[fillstyle=solid,fillcolor=lightyellow,linestyle=none](0,0){40}
\pscircle[fillstyle=solid,fillcolor=lightyellow,linestyle=none](60,0){25}
\pscircle[fillstyle=solid,fillcolor=lightyellow,linestyle=none](63,37){20}
\pscircle[fillstyle=solid,fillcolor=lightyellow,linestyle=none](37,55){18}
\pscircle[fillstyle=solid,fillcolor=lightyellow,linestyle=none](11,45){16}
\pscircle[fillstyle=solid,fillcolor=lightyellow,linestyle=none](0,22.4){14}

    \psarc[linewidth=2pt](0,0){40}{18}{342}

    \psarc[linewidth=2pt](60,0){25}{114.7}{150}
    \psarc[linewidth=2pt](60,0){25}{209.8}{53.7}

    \psarc[linewidth=2pt](63,37){20}{176}{224.8}
    \psarc[linewidth=2pt](63,37){20}{303.6}{114.7}

    \psarc[linewidth=2pt](37,55){18}{0}{168.6}
    \psarc[linewidth=2pt](37,55){18}{233.6}{290.6}

    \psarc[linewidth=2pt](11,45){16}{58.4}{213.1}
    \psarc[linewidth=2pt](11,45){16}{275}{343.9}

    \psarc[linewidth=2pt](0,22.4){14}{100}{28}

    \uput[90](11,-55){$O_1$}

    \uput[90](57,-40){$O_2$}

    \uput[90](90,44){$O_3$}

    \uput[90](55,65){$O_4$}

    \uput[90](3, 60){$O_5$}

    \uput[90](-20,10){$O_6$}

    \psarc[linecolor=blue,linewidth=0.5pt](0,0){40}{342}{18}
    \psarc[linecolor=blue,linewidth=0.5pt](60,0){25}{150}{209.8}
    \psarc[linecolor=blue,linewidth=0.5pt](60,0){25}{53.7}{114.7}
    \psarc[linecolor=blue,linewidth=0.5pt](63,37){20}{224.8}{303.6}
    \psarc[linecolor=blue,linewidth=0.5pt](63,37){20}{114.7}{176}
    \psarc[linecolor=blue,linewidth=0.5pt](37,55){18}{168.6}{233.6}
    \psarc[linecolor=blue,linewidth=0.5pt](37,55){18}{290.6}{0}
    \psarc[linecolor=blue,linewidth=0.5pt](11,45){16}{213.1}{275}
    \psarc[linecolor=blue,linewidth=0.5pt](11,45){16}{343.9}{58.4}
    \psarc[linecolor=blue,linewidth=0.5pt](0,22.4){14}{28}{100}

\uput[-90](0,-60){(b) Example 2. }
\end{pspicture}

\caption{Examples of chains. The connecting arcs are thin. A chain may self-intersect, as in Example 2. Although $O_1$ and $O_5$ intersect in Example 2, there are no connecting arcs between them because they are not consecutive in the sequence.}\label{fig:G}
\end{center}
\end{figure}

%% file: fig-v2-chain-terminal.tex
\begin{figure}[tbph]
\psset{unit=0.8pt}
\begin{center} \small
\begin{pspicture}(-120,-110)(150,100)
    \psset{labelsep=6pt}


\pscircle[fillstyle=solid,fillcolor=lightyellow,linestyle=none](-107,4){35}

    \psarc(-107,4){35}{10}{220}
    \psarc[linewidth=2pt,linecolor=blue](-107,4){35}{220}{302}

\pscircle[fillstyle=solid,fillcolor=lightyellow,linestyle=none](-27,-31){61}
    \psarc[linewidth=2pt,linecolor=blue](-27,-31){61}{101}{138}
    \psarc(-27,-31){61}{175}{372}

\pscircle[fillstyle=solid,fillcolor=lightyellow,linestyle=none](21,41){61}
    \psarc(21,41){61}{20}{192}

\pscircle[fillstyle=solid,fillcolor=lightyellow,linestyle=none](68,14){49}
    \psarc[linewidth=2pt,linecolor=blue](68,14){49}{59}{77}
    \psarc(68,14){49}{224}{280}

\pscircle[fillstyle=solid,fillcolor=lightyellow,linestyle=none](123,4){60}
    \psarc[linewidth=2pt,linecolor=blue](123,4){60}{26}{120}
    \psarc(123,4){60}{-140}{26}

    \psline[linewidth=2pt,linecolor=blue](-72,10)(-88,-26)
    \psline[linewidth=2pt,linecolor=blue](-39,28)(32,-19)
    \psline[linewidth=2pt,linecolor=blue](79,62)(32,-19)
    \psline(93,56)(76,-34)

    \uput[170](-72,10){$a_1$}
    \uput[-30](-88,-26){$b_1$}

    \uput[90](-110,35){$A_1$}
    \uput[-90](-110,-30){$B_1$}

    \uput[0](-39,28){$a_2$}
    \uput[220](32,-19){$b_2$($b_3$)}

    \uput[90](-60,20){$A_2$}
    \uput[-90](-50,-85){$B_2$}

    \uput[180](79,62){$a_3$}

    \uput[90](0,95){$A_3$}
    \uput[-90](42,-30){$B_3$}

    \uput[-20](93,56){$a_4$}
    \uput[0](76,-34){$b_4$}

    \uput[90](90,55){$A_4$}
    \uput[-90](65,-30){$B_4$}



    \uput[90](150,55){$A_5$}
    \uput[-90](135,-55){$B_5$}


\psline[linewidth=2pt,linestyle=dotted,dotsep=2pt](-134,-19)(-107,4)(-27,-31)(21,41)(68,14)(123,4)(178,30)

    \psline[linestyle=dashed,linecolor=red](-134,-19)(178,30)

    \uput[-190](-134,-19){$u$}
    \uput[-120](-125,-19){($a_0$,$b_0$)}

    \uput[30](178,30){$v$}
    \uput[-30](178,30){($a_5$,$b_5$)}

\uput[-90](0,-105){(a) Example 1. }
\end{pspicture}

\psset{unit=1.4pt}
\begin{pspicture}(-50,-80)(110,100)
    \psset{labelsep=6pt}

\pscircle[fillstyle=solid,fillcolor=lightyellow,linestyle=none](0,0){40}
\pscircle[fillstyle=solid,fillcolor=lightyellow,linestyle=none](60,0){25}
\pscircle[fillstyle=solid,fillcolor=lightyellow,linestyle=none](63,37){20}
\pscircle[fillstyle=solid,fillcolor=lightyellow,linestyle=none](37,55){18}
\pscircle[fillstyle=solid,fillcolor=lightyellow,linestyle=none](11,45){16}
\pscircle[fillstyle=solid,fillcolor=lightyellow,linestyle=none](0,22.4){14}

    \psarc[linewidth=2pt,linecolor=blue](0,0){40}{16}{180}
    \psarc(0,0){40}{180}{342}

    \psarc[linewidth=2pt,linecolor=blue](60,0){25}{114.7}{150}
    \psarc(60,0){25}{209.8}{53.7}

    \psarc[linewidth=2pt,linecolor=blue](63,37){20}{176}{224.8}
    \psarc(63,37){20}{303.6}{114.7}

    \psarc(37,55){18}{0}{168.6}
    \psarc[linewidth=2pt,linecolor=blue](37,55){18}{233.6}{290.6}

    \psarc(11,45){16}{58.4}{213.1}
    \psarc[linewidth=2pt,linecolor=blue](11,45){16}{275}{343.9}

    \psarc[linewidth=2pt,linecolor=blue](0,22.4){14}{100}{172.8}
    \psarc(0,22.4){14}{172.8}{28}

    \psline(38,11)(38,-12.5)
    \psline(49.2,22.7)(74.3,20.2)
    \psline(43.4,38.2)(55,55)
    \psline(26.4,40.5)(19.4,58.6)
    \psline[linewidth=2pt,linecolor=blue](12.4,29)(-2.4,36.2)

    \psline[linewidth=2pt,linestyle=dotted,dotsep=2pt](-40,0)(0,0)(60,0)(63,37)(37,55)(11,45)(0,22.4)(-13.9, 24.2)

    \psline[linestyle=dashed,linecolor=red](-40,0)(38,11)(50,22.7)(43.4,38.2)(26.4,40.5)(-13.9, 24.2)


    \uput[150](-40,0){$u$}
    \uput[210](-37,0){($a_0$,$b_0$)}

    \uput[150](-13.9, 24.2){$v$}
    \uput[210](-9.5, 24.2){($a_6$,$b_6$)}

    \uput[0](38,11){$a_1$}
    \uput[0](38,-11){$b_1$}

    \uput[45](50,22.7){$a_2$}
    \uput[0](74.3,20.2){$b_2$}

    \uput[0](43.4,38.2){$a_3$}
    \uput[45](55,55){$b_3$}

    \uput[60](26.4,40.5){$a_4$}
    \uput[120](19.4,58.6){$b_4$}

    \uput[-20](12.4,29){$a_5$}
    \uput[170](-2.4,36.2){$b_5$}

\uput[-90](0,-60){(b) Example 2. }
\end{pspicture}
\psset{unit=0.5pt}
\begin{pspicture}(-160,-200)(150,120)

    \psline[linewidth=2pt,linecolor=blue](-150,0)(-100,60)
    \psline(-150,0)(-100,-60)
    \psline(-100,60)(-100,-60)

    \psline[linewidth=2pt,linecolor=blue](-100,60)(-50,60)
    \psline(-100,-60)(-50,-60)

    \psline(-50,60)(-50,-60)

    \psline[linewidth=2pt,linecolor=blue](-50,60)(0,60)
    \psline(-50,-60)(0,-60)

    \psline(0,60)(0,-60)

    \psline[linewidth=2pt,linecolor=blue](0,60)(50,60)
    \psline(0,-60)(50,-60)

    \psline(50,60)(50,-60)

    \psline[linewidth=2pt,linecolor=blue](50,60)(100,60)
    \psline(50,-60)(100,-60)

    \psline[linewidth=2pt,linecolor=blue](100,60)(100,-60)

    \psline(100,60)(150,0)
    \psline[linewidth=2pt,linecolor=blue](100,-60)(150,0)

    \cnode*[linecolor=red](-150,0){2pt}{u}
    \cnode*[linecolor=red](-100,60){2pt}{a_1}
    \cnode*[linecolor=red](-100,-60){2pt}{b_1}
    \cnode*[linecolor=red](-50,60){2pt}{a_2}
    \cnode*[linecolor=red](-50,-60){2pt}{b_2}
    \cnode*[linecolor=red](150,0){2pt}{v}
    \cnode*[linecolor=red](100,60){2pt}{a_5}
    \cnode*[linecolor=red](100,-60){2pt}{b_5}
    \cnode*[linecolor=red](50,60){2pt}{a_4}
    \cnode*[linecolor=red](50,-60){2pt}{b_4}
    \cnode*[linecolor=red](0,60){2pt}{a_3}
    \cnode*[linecolor=red](0,-60){2pt}{b_3}

    \psset{labelsep=6pt}

    \uput[90](-100,60){$a_1$}
    \uput[90](-50,60){$a_2$}
    \uput[90](0,60){$a_3$}
    \uput[90](50,60){$a_4$}
    \uput[90](100,60){$a_5$}


    \uput[-90](-100,-60){$b_1$}
    \uput[-90](-50,-60){$b_2$}
    \uput[-90](0,-60){$b_3$}
    \uput[-90](50,-60){$b_4$}
    \uput[-90](100,-60){$b_5$}


    \uput[-190](-150,0){$u$}
    \uput[-120](-150,0){($a_0$,$b_0$)}
    \uput[-10](150,0){$v$}
    \uput[-60](150,0){($a_6$,$b_6$)}

\uput[-90](0,-140){(c) Graph representation of Example 2.}
\end{pspicture}
\caption{Illustrations for the definitions of $D_{\mathcal{O}}(u,v)$ and $P_{\mathcal{O}}(u,v)$. In both Example 1 and Example 2, $D_{\mathcal{O}}(u,v)$ is the dashed (poly)line, $P_{\mathcal{O}}(u,v)$ is the thick path, and the dotted polyline is the centered polyline between $u$ and $v$. In Example 1, the terminals $u,v$ are unobstructed and in Example 2, $u,v$ are obstructed. Figure (c) is the graph representation $\mathbb{G}_{\mathcal{O}}$ of the chain $\mathcal{O}$ in Example 2, where the weight of each edge is the length of the corresponding arc or line segment in $\mathcal{O}$. The shortest path between $u$ and $v$ in (c) is the thick path $P_{\mathbb{G}_{\mathcal{O}}}(u,v)$, which corresponds to the thick path $P_{\mathcal{O}}(u,v)$ in (b). }\label{fig:polyline}
\end{center}
\end{figure}

%% file: fig-v3-stab.tex
\begin{figure}[tbhp]
\psset{unit=0.6pt}
\begin{center} \small
\begin{pspicture}(-160,-150)(160,140)
    \psset{labelsep=5pt}

    \pscircle(-100,20){70}

    \pscircle(115,-20){55}

    \pscircle[linecolor=blue](0,0){100}

    \pscircle[linecolor=blue](60,0){87.5}

    \psline[linecolor=red](50,150)(50,-150)

  \psline[linestyle=dashed]{->}(-160,55)(159,-52)

  \uput[110](-160,55){$u$}

\uput[-70](159,-52){$v$}

\uput[50](50,87){$a_i$}
\uput[-50](50,-87){$b_i$}

\uput[50](-80,80){$a_{i-1}$}
\uput[-50](-110,-50){$b_{i-1}$}

\uput[50](139,25){$a_{i+1}$}
\uput[-50](90,-75){$b_{i+1}$}

\uput[-45](-35,125){$O_{i-1}$}
\uput[-45](90,100){$O_i$}

\uput[-45](-155,108){$O_{i-2}$}
\uput[-45](165,5){$O_{i+1}$}

\uput[-45](-75,-130){left half-plane}
\uput[-45](65,-130){right half-plane}

\cnode*[linecolor=red](-160,55){2pt}{u}

\cnode*[linecolor=red](159,-52){2pt}{v}

\end{pspicture}

\caption{An illustration for the proof of Fact~\ref{fact}. The left-portion of $O_i$ (w.r.t. $a_ib_i$) is contained in the left-portion of $O_{i-1}$, and the right-portion of $O_{i-1}$ is contained in the right-portion of $O_i$.}\label{fig:stab}
\end{center}
\end{figure}
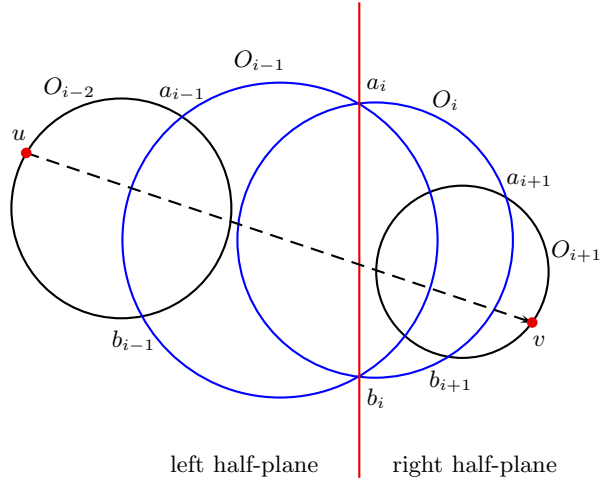

%% file: fig-v2-phi.tex
\begin{figure}[tbp]
\psset{unit=0.7pt}
\psset{labelsep=5pt}

\begin{center} \small
\begin{pspicture}(-150,-130)(150,130)

    \psline{->}(-130,0)(130,0)



    \pscircle(0,0){100}
    \pscircle(40,0){70}

    \psline[linestyle=dashed](0,0)(0,100)
    \psline[linestyle=dashed](40,0)(40,70)
    \psline[linestyle=dashed](0,0)(84,55)
    \psline[linestyle=dashed](40,0)(84,55)

    \psline[linestyle=dotted](0,70)(40,70)




    \uput[30](84,55){$a_{i-1}$}

    \psarc[linecolor=green,linewidth=3pt](0,0){100}{33}{90}
    \psarc[linecolor=red,linewidth=3pt](40,0){70}{51}{90}


    \uput[30](-120,-80){$O_{i-1}$}
    \uput[-90](120,-20){$O_{i}$}

    \uput[90](0,100){$q_{i-1}^{\rightarrow}$}
    \uput[90](40,65){$q_{i}^{\leftarrow}$}

    \uput[100](65,90){$\mathcal{P}_i$}

\uput[-90](0,-110){(a)}

    \uput[-90](50,0){$o_i$}
    \uput[-60](-10,0){$o_{i-1}$}

\end{pspicture}
\begin{pspicture}(-150,-130)(150,130)

    \psline{->}(-130,0)(130,0)


    \pscircle(0,0){100}
    \pscircle(-60,0){50}

    \psline[linestyle=dashed](0,0)(0,100)
    \psline[linestyle=dashed](-60,0)(-60,50)
    \psline[linestyle=dashed](0,0)(-92,38)
    \psline[linestyle=dashed](-60,0)(-92,38)

    \psline[linestyle=dotted](-60,50)(0,50)

    \psarc[linecolor=green,linewidth=3pt](0,0){100}{90}{156}
    \psarc[linecolor=red,linewidth=3pt](-60,0){50}{90}{131}

    \uput[0](80,-60){$O_i$}
    \uput[-90](-120,-20){$O_{i-1}$}

    \uput[-90](0,0){$o_i$}
    \uput[-60](-70,0){$o_{i-1}$}

    \uput[180](-92,38){$a_{i-1}$}

    \uput[90](0,100){$q_i^{\leftarrow}$}
    \uput[70](-60,45){$q_{i-1}^{\rightarrow}$}

    \uput[100](-45,90){$\mathcal{P}_i$}

\uput[-90](0,-110){(b)}
\end{pspicture}
\begin{pspicture}(-140,-130)(140,150)

    \psline{->}(-130,0)(130,0)


    \pscircle(40,0){70}
    \pscircle(-60,0){50}

    \psline[linestyle=dashed](40,0)(40,70)
    \psline[linestyle=dashed](-60,0)(-60,50)
    \psline[linestyle=dashed](40,0)(-21,31)
    \psline[linestyle=dashed](-60,0)(-21,31)


    \psline[linestyle=dotted](-60,50)(40,50)

    \psarc[linecolor=green,linewidth=3pt](40,0){70}{90}{152}
    \psarc[linecolor=green,linewidth=3pt](-60,0){50}{38}{90}

    \uput[0](90,-50){$O_i$}
    \uput[-90](-80,-40){$O_{i-1}$}

    \uput[-90](50,0){$o_i$}
    \uput[-60](-70,0){$o_{i-1}$}

    \uput[180](25,35){$a_{i-1}$}

    \uput[90](40,70){$q_i^{\leftarrow}$}
    \uput[90](-60,45){$q_{i-1}^{\rightarrow}$}

    \uput[100](-5,60){$\mathcal{P}_i$}

\uput[-90](0,-110){(c)}
\end{pspicture}
\begin{pspicture}(-180,-130)(160,150)

    \psline(0,0)(40,0)

    \psline(0,0)(-110,-61)

    \pscircle(0,0){100}
    \pscircle(40,0){70}
    \pscircle(-110,-61){50}

    \psline[linestyle=dashed](0,0)(0,100)
    \psline[linestyle=dashed](40,0)(40,70)
    \psline[linestyle=dashed](0,0)(84,55)
    \psline[linestyle=dashed](40,0)(84,55)

    \psline[linestyle=dashed](-110,-61)(-134,-18)
    \psline[linestyle=dashed](0,0)(-49,87)

    \psline[linestyle=dashed](-110,-61)(-99,-12)
    \psline[linestyle=dashed](0,0)(-99,-12)


    \uput[30](84,55){$a_{i}$}
    \uput[45](-99,-12){$a_{i-1}$}

    \psarc[linecolor=green,linewidth=3pt](0,0){100}{33}{90}
    \psarc[linecolor=red,linewidth=3pt](40,0){70}{51}{90}

    \psarc[linecolor=green,linewidth=3pt](0,0){100}{119}{187}
    \psarc[linecolor=green,linewidth=3pt](-110,-61){50}{78}{119}

    \uput[30](50,-100){$O_i$}
    \uput[-90](-165,-80){$O_{i-1}$}
    \uput[-90](125,-20){$O_{i+1}$}

    \uput[90](-49,87){$q_i^{\leftarrow}$}
    \uput[90](0,100){$q_i^{\rightarrow}$}
    \uput[100](40,65){$q_{i+1}^{\leftarrow}$}
    \uput[100](-130,-18){$q_{i-1}^{\rightarrow}$}

    \uput[100](-100,40){$\mathcal{P}_i$}
    \uput[100](70,90){$\mathcal{P}_{i+1}$}

\uput[-90](0,-110){(d)}

\end{pspicture}

\caption{Illustrations for Definition~\ref{peak}. The path $\mathcal{P}_i$ consists of ``heavy'' (red or dark-gray) arcs and/or ``light'' (green or light-gray) arcs. There are three possible cases for two consecutive disks $O_{i-1}$ and $O_i$ to intersect, as shown in case (a), (b), and (c).
In case (a), $Q_{i-1}^{\rightarrow}$ is light and $Q_i^{\leftarrow}$ is heavy. In this case, $r_{i-1} > r_i$. In case (b), $Q_{i-1}^{\rightarrow}$ is heavy and $Q_i^{\leftarrow}$ is light. In this case, $r_i > r_{i-1}$.  In case (c), both $Q_{i-1}^{\rightarrow}$ and $Q_i^{\leftarrow}$ are light. In this case, $r_{i-1} > r_i$, $r_i > r_{i-1}$ and $r_{i-1} = r_i$ are all possible.
Figure (d) shows an example of $\mathcal{P}_i$ and $\mathcal{P}_{i+1}$ defined on three consecutive disks $O_{i-1},O_i$ and $O_{i+1}$. Note that $O_i$ has two different peaks. The peak $q_i^{\leftarrow}$ is defined with regard to the preceding disk $O_{i-1}$. The peak $q_i^{\rightarrow}$ is defined with regard to the succeeding disk $O_{i+1}$.
}\label{fig:phi}
\end{center}
\end{figure}
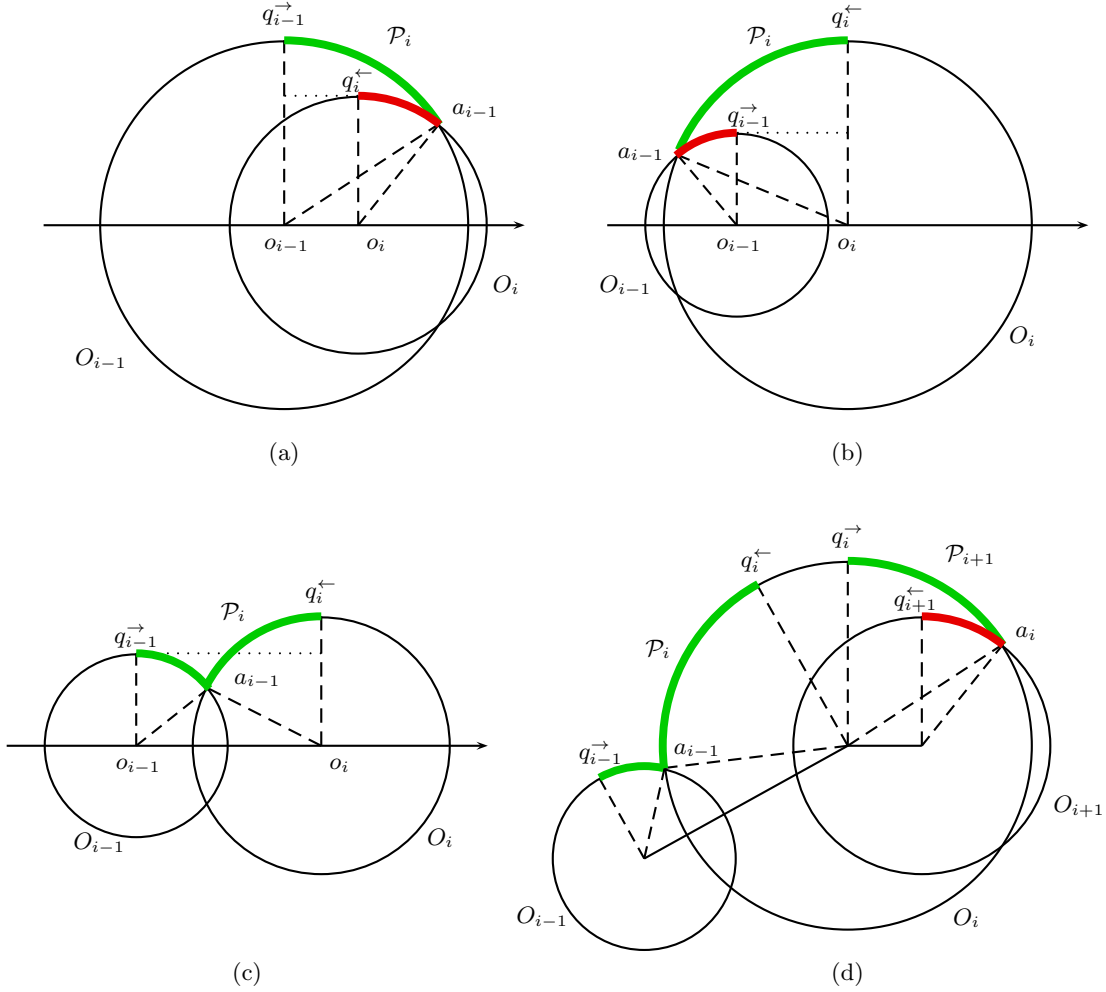

%% file: fig-v2-chain-at-end.tex
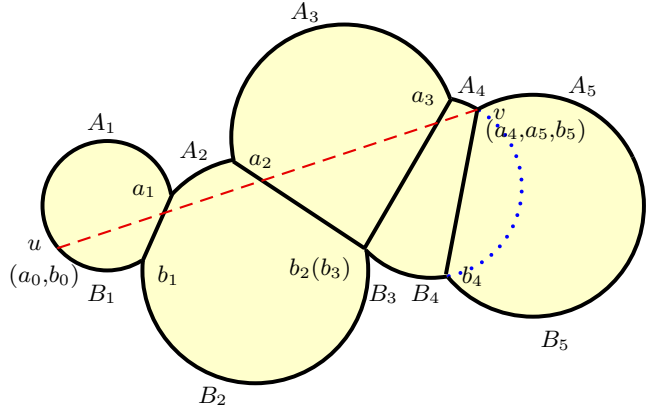
\begin{figure}[tbp]
\psset{unit=0.7pt}
\begin{center} \small
\begin{pspicture}(-120,-110)(150,130)
    \psset{labelsep=6pt}


\pscircle[fillstyle=solid,fillcolor=lightyellow,linestyle=none](-107,4){35}

    \psarc[linewidth=1.5pt](-107,4){35}{10}{220}
    \psarc[linewidth=1.5pt](-107,4){35}{220}{302}

\pscircle[fillstyle=solid,fillcolor=lightyellow,linestyle=none](-27,-31){61}
    \psarc[linewidth=1.5pt](-27,-31){61}{101}{138}
    \psarc[linewidth=1.5pt](-27,-31){61}{175}{372}

\pscircle[fillstyle=solid,fillcolor=lightyellow,linestyle=none](21,41){61}
    \psarc[linewidth=1.5pt](21,41){61}{20}{192}

\pscircle[fillstyle=solid,fillcolor=lightyellow,linestyle=none](68,14){49}
    \psarc[linewidth=1.5pt](68,14){49}{59}{77}
    \psarc[linewidth=1.5pt](68,14){49}{224}{280}

\pscircle[fillstyle=solid,fillcolor=lightyellow,linestyle=none](123,4){60}

\psarc[linewidth=1.5pt,linecolor=blue,linestyle=dotted](68,14){49}{280}{59}

    \psarc[linewidth=1.5pt](123,4){60}{26}{120}
    \psarc[linewidth=1.5pt](123,4){60}{-140}{26}

    \psline[linewidth=1.5pt](-72,10)(-88,-26)
    \psline[linewidth=1.5pt](-39,28)(32,-19)
    \psline[linewidth=1.5pt](79,62)(32,-19)
    \psline[linewidth=1.5pt](93,56)(76,-33)

    \uput[170](-72,10){$a_1$}
    \uput[-30](-88,-26){$b_1$}

    \uput[90](-110,35){$A_1$}
    \uput[-90](-110,-30){$B_1$}

    \uput[0](-39,28){$a_2$}
    \uput[220](32,-19){$b_2$($b_3$)}

    \uput[90](-60,20){$A_2$}
    \uput[-90](-50,-85){$B_2$}

    \uput[180](79,62){$a_3$}

    \uput[90](0,95){$A_3$}
    \uput[-90](42,-30){$B_3$}

    \uput[-10](93,56){$v$}
    \uput[-40](90,56){($a_4$,$a_5$,$b_5$)}
    
    \uput[0](76,-34){$b_4$}

    \uput[90](90,55){$A_4$}
    \uput[-90](65,-30){$B_4$}



    \uput[90](150,55){$A_5$}
    \uput[-90](135,-55){$B_5$}



    \psline[linecolor=red,linestyle=dashed](-134,-19)(93,56)

    \uput[-190](-134,-19){$u$}
    \uput[-120](-125,-19){($a_0$,$b_0$)}


\end{pspicture}

\caption{An illustration for Proposition~\ref{endsok}. }\label{fig:at-end}
\end{center}
\end{figure}

%% file: fig-v2-chain-obstruct.tex
\begin{figure}[tbp]
\psset{unit=0.7pt}
\begin{center} \small
\begin{pspicture}(-120,-110)(150,130)
    \psset{labelsep=6pt}


\pscircle[fillstyle=solid,fillcolor=lightyellow,linestyle=none](-107,4){35}

    \psarc[linewidth=1.5pt](-107,4){35}{10}{220}
    \psarc[linewidth=1.5pt](-107,4){35}{220}{302}

\pscircle[fillstyle=solid,fillcolor=lightyellow,linestyle=none](-27,-31){61}
    \psarc[linewidth=1.5pt](-27,-31){61}{101}{138}
    \psarc[linewidth=1.5pt](-27,-31){61}{175}{372}

\pscircle[fillstyle=solid,fillcolor=lightyellow,linestyle=none](21,41){61}
    \psarc[linewidth=1.5pt](21,41){61}{20}{192}

\pscircle[fillstyle=solid,fillcolor=lightyellow,linestyle=none](68,14){49}
    \psarc[linewidth=1.5pt](68,14){49}{59}{77}
    \psarc[linewidth=1.5pt](68,14){49}{224}{280}

\pscircle[fillstyle=solid,fillcolor=lightyellow,linestyle=none](123,4){60}
    \psarc[linewidth=1.5pt](123,4){60}{26}{120}
    \psarc[linewidth=1.5pt](123,4){60}{-140}{26}

    \psline[linewidth=1.5pt](-72,10)(-88,-26)
    \psline[linewidth=1.5pt](-39,28)(32,-19)
    \psline[linewidth=1.5pt](79,62)(32,-19)
    \psline[linewidth=1.5pt](93,56)(76,-34)

    \uput[-50](-72,10){$a_1$}
    \uput[-30](-88,-26){$b_1$}

    \uput[90](-110,35){$A_1$}
    \uput[-90](-110,-30){$B_1$}

    \uput[0](-39,28){$a_2$}
    \uput[220](32,-19){$b_2$($b_3$)}

    \uput[90](-60,20){$A_2$}
    \uput[-90](-50,-85){$B_2$}

    \uput[180](79,62){$a_3$}

    \uput[90](0,95){$A_3$}
    \uput[-90](42,-30){$B_3$}

    \uput[-20](93,56){$a_4$}
    \uput[0](76,-34){$b_4$}

    \uput[90](90,55){$A_4$}
    \uput[-90](65,-30){$B_4$}



    \uput[90](150,55){$A_5$}
    \uput[-90](135,-55){$B_5$}

    \psline[linecolor=red,linestyle=dashed](168,44)(-72,10)
    \psline[linecolor=red,linestyle=dashed](-140,19)(-72,10)



    \uput[120](-140,19){$u$}
    \uput[180](-140,19){($a_0$,$b_0$)}

    \uput[60](168,44){$v$}
    \uput[0](168,44){($a_5$,$b_5$)}

\end{pspicture}

\caption{An illustration for Proposition~\ref{obsok}. }\label{fig:obstruct}
\end{center}
\end{figure}
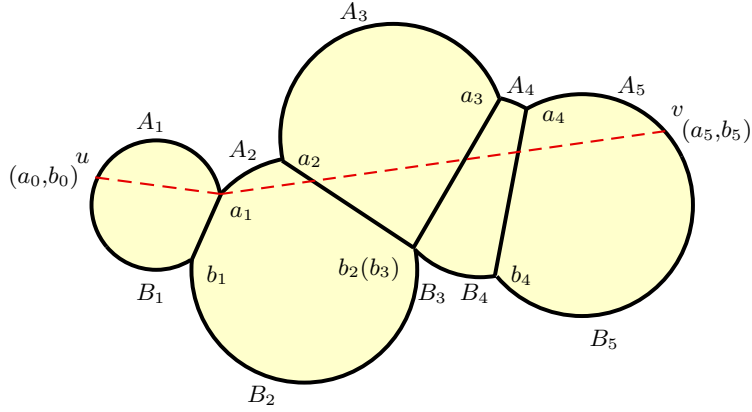

%% file: fig-v2-chain-boundary.tex
\begin{figure}[tbp]
\psset{unit=0.7pt}
\begin{center} \small
\begin{pspicture}(-120,-160)(150,130)
    \psset{labelsep=6pt}


\pscircle[fillstyle=solid,fillcolor=lightyellow,linestyle=none](-107,4){35}

    \psarc(-107,4){35}{10}{220}
    \psarc(-107,4){35}{220}{302}

\pscircle[fillstyle=solid,fillcolor=lightyellow,linestyle=none](-27,-31){61}
    \psarc(-27,-31){61}{101}{138}
    \psarc(-27,-31){61}{175}{372}

\pscircle[fillstyle=solid,fillcolor=lightyellow,linestyle=none](21,41){61}
    \psarc(21,41){61}{20}{192}

\pscircle[fillstyle=solid,fillcolor=lightyellow,linestyle=none](68,14){49}
    \psarc(68,14){49}{55}{77}
    \psarc(68,14){49}{224}{262}

\pscircle[fillstyle=solid,fillcolor=lightyellow,linestyle=none](93,4){50}
    \psarc[linewidth=3pt,linecolor=magenta](93,4){50}{-30}{88}
    \psarc(93,4){50}{-130}{-30}

    \psline(-72,10)(-88,-26)
    \psline(-39,28)(32,-19)
    \psline(79,62)(32,-19)
    \psline(95,54)(61,-34)

    \uput[-30](-98,18){$a_1$}
    \uput[-30](-88,-26){$b_1$}

    \uput[90](-110,35){$A_1$}
    \uput[-90](-110,-30){$B_1$}

    \uput[0](-39,30){$a_2$}
    \uput[220](32,-19){$b_2$($b_3$)}

    \uput[90](-60,20){$A_2$}
    \uput[-90](-50,-85){$B_2$}

    \uput[180](79,62){$a_3$}

    \uput[90](0,95){$A_3$}
    \uput[-90](42,-25){$B_3$}

    \uput[-20](90,46){$a_4$}
    \uput[0](60,-34){$b_4$}

    \uput[90](90,55){$A_4$}
    \uput[-90](60,-30){$B_4$}



    \uput[90](125,44){$A_5$}
    \uput[-90](125,-38){$B_5$}

    \psline[linestyle=dashed](-136,-15)(95,54)
    \psline[linestyle=dashed](-136,-15)(136,-21)



    \uput[150](-136,-30){$u$}

    \uput[0](142,0){$\widehat{A}$}

    \uput[10](138,25){$v$}




\uput[-90](0,-115){(a) The trick arc is the unobstructed arc $\widehat{A}$. }
\end{pspicture}

\begin{pspicture}(-120,-160)(160,130)
    \psset{labelsep=6pt}


\pscircle[fillstyle=solid,fillcolor=lightyellow,linestyle=none](-107,4){35}

    \psarc(-107,4){35}{10}{220}
    \psarc(-107,4){35}{220}{302}

\pscircle[fillstyle=solid,fillcolor=lightyellow,linestyle=none](-27,-31){61}
    \psarc(-27,-31){61}{101}{138}
    \psarc(-27,-31){61}{175}{372}

\pscircle[fillstyle=solid,fillcolor=lightyellow,linestyle=none](21,41){61}
    \psarc(21,41){61}{20}{192}

\pscircle[fillstyle=solid,fillcolor=lightyellow,linestyle=none](68,14){49}
    \psarc(68,14){49}{55}{77}
    \psarc(68,14){49}{224}{262}

\pscircle[fillstyle=solid,fillcolor=lightyellow,linestyle=none](93,4){50}
    \psarc(93,4){50}{-130}{88}

    \psline(-72,10)(-88,-26)
    \psline(-39,28)(32,-19)
    \psline(79,62)(32,-19)
    \psline(95,54)(61,-34)

\psarc[linewidth=3pt,linecolor=red](-107,4){35}{215}{305}
\psline[linewidth=3pt,linecolor=red](-73,10)(-89,-26)
\psarc[linewidth=3pt,linecolor=red](-27,-31){61}{101}{138}
\psline[linewidth=3pt,linecolor=red](-39,29)(32,-18)
\psline[linewidth=3pt,linecolor=red](79,62)(32,-19)
\psarc[linewidth=3pt,linecolor=red](68,14){49}{55}{77}
\psarc[linewidth=3pt,linecolor=red](93,4){50}{24}{88}

    \uput[-30](-98,18){$a_1$}
    \uput[-30](-88,-26){$b_1$}

    \uput[90](-110,35){$A_1$}
    \uput[-90](-110,-30){$B_1$}

    \uput[0](-39,30){$a_2$}
    \uput[220](32,-19){$b_2$($b_3$)}

    \uput[90](-60,20){$A_2$}
    \uput[-90](-50,-85){$B_2$}

    \uput[180](79,62){$a_3$}

    \uput[90](0,95){$A_3$}
    \uput[-90](42,-25){$B_3$}

    \uput[-20](90,46){$a_4$}
    \uput[0](60,-34){$b_4$}

    \uput[90](90,55){$A_4$}
    \uput[-90](60,-30){$B_4$}



    \uput[90](125,44){$A_5$}
    \uput[-90](125,-38){$B_5$}

    \psline[linestyle=dashed](-136,-15)(138,25)



    \uput[150](-136,-30){$u$}

    \uput[0](142,0){$\widehat{A}$}

    \uput[10](138,25){$v$}




\uput[-90](0,-115){(b) The thick path is $P_{\mathcal{O}}^{A_n}(u,v)$, which is the shortest paths from $u$ to $v$ that include the arcs $A_n$. }
\end{pspicture}

\begin{pspicture}(-120,-150)(150,130)
    \psset{labelsep=6pt}


\pscircle[fillstyle=solid,fillcolor=lightyellow,linestyle=none](-107,4){35}

    \psarc(-107,4){35}{10}{220}
    \psarc(-107,4){35}{220}{302}

\pscircle[fillstyle=solid,fillcolor=lightyellow,linestyle=none](-27,-31){61}
    \psarc(-27,-31){61}{101}{138}
    \psarc(-27,-31){61}{175}{372}

\pscircle[fillstyle=solid,fillcolor=lightyellow,linestyle=none](21,41){61}
    \psarc(21,41){61}{20}{192}

\pscircle[fillstyle=solid,fillcolor=lightyellow,linestyle=none](68,14){49}
    \psarc(68,14){49}{55}{77}
    \psarc(68,14){49}{224}{262}

\pscircle[fillstyle=solid,fillcolor=lightyellow,linestyle=none](93,4){50}
    \psarc(93,4){50}{-130}{88}

    \psline(-72,10)(-88,-26)
    \psline(-39,28)(32,-19)
    \psline(79,62)(32,-19)
    \psline(95,54)(61,-34)

\psarc[linewidth=3pt,linecolor=blue](-107,4){35}{215}{305}
\psline[linewidth=3pt,linecolor=blue](-71,10)(-87,-26)
\psarc[linewidth=3pt,linecolor=blue](-27,-31){61}{101}{138}
\psline[linewidth=3pt,linecolor=blue](-39,27)(32,-20)
\psarc[linewidth=3pt,linecolor=blue](68,14){49}{224}{262}
\psarc[linewidth=3pt,linecolor=blue](93,4){50}{-130}{24}

    \uput[-30](-98,18){$a_1$}
    \uput[-30](-88,-26){$b_1$}

    \uput[90](-110,35){$A_1$}
    \uput[-90](-110,-30){$B_1$}

    \uput[0](-39,30){$a_2$}
    \uput[220](32,-19){$b_2$($b_3$)}

    \uput[90](-60,20){$A_2$}
    \uput[-90](-50,-85){$B_2$}

    \uput[180](79,62){$a_3$}

    \uput[90](0,95){$A_3$}
    \uput[-90](42,-25){$B_3$}

    \uput[-20](90,46){$a_4$}
    \uput[0](60,-34){$b_4$}

    \uput[90](90,55){$A_4$}
    \uput[-90](60,-30){$B_4$}



    \uput[90](125,44){$A_5$}
    \uput[-90](125,-38){$B_5$}

    \psline[linestyle=dashed](-136,-15)(138,25)



    \uput[150](-136,-30){$u$}

    \uput[0](142,0){$\widehat{A}$}

    \uput[10](138,25){$v$}




\uput[-90](0,-115){(c) The thick path is $P_{\mathcal{O}}^{B_n}(u,v)$, which is the shortest paths from $u$ to $v$ that include the arcs $B_n$. }
\end{pspicture}

\caption{
An illustration for Definition~\ref{def:pivot}. $v$ is pivotal if $|P_{\mathcal{O}}^{A_n}(u,v)|=|P_{\mathcal{O}}^{B_n}(u,v)|$.}\label{fig:pivotal}
\end{center}
\end{figure}

%% file: fig-v2-convex.tex
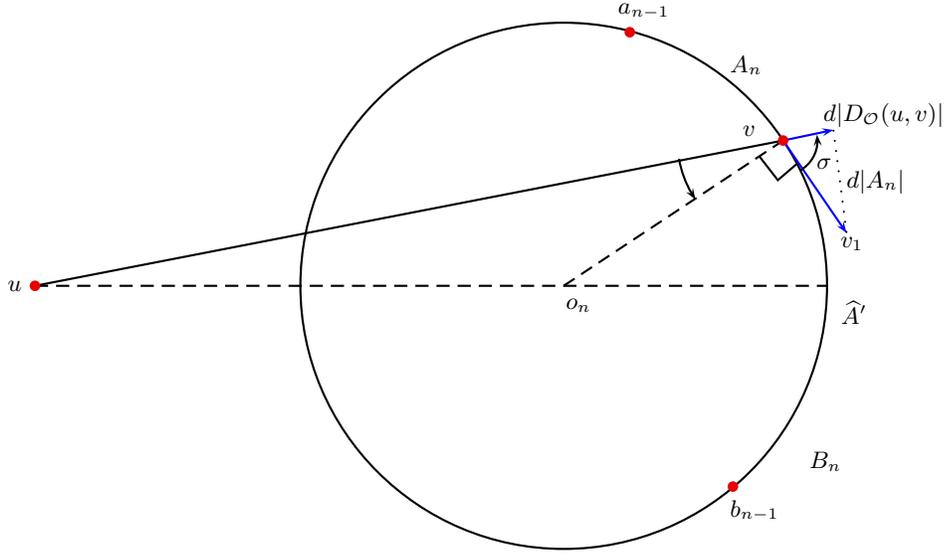
\begin{figure}[tbp]
\psset{unit=1pt}
\begin{center} \small
\begin{pspicture}(-220,-120)(140,120)
    \psset{labelsep=5pt}

    \pscircle(0,0){100}
  \psline[linestyle=dashed](-200,0)(100,0)

  \psline[linestyle=dashed](0,0)(83,55)
  \psline(-200,0)(83,55)

\uput[0](100,-10){$\widehat{A}'$}


  \psline[linecolor=blue]{->}(83,55)(107,20)

  \uput[180](-200,0){$u$}
  \uput[300](0,0){$o_n$}

\uput[60](66,52){$v$}
\uput[70](25,96){$a_{n-1}$}

\cnode*[linecolor=red](64,-76){2pt}{b}
\uput[-60](64,-76){$b_{n-1}$}

\uput[70](65,75){$A_n$}
\uput[70](95,-75){$B_n$}

\psline[linecolor=blue]{->}(83,55)(102,59)

\psline[linestyle=dotted](107,20)(102,59)

\psarc{->}(83,55){13}{-60}{10}

\uput[-30](91,50){$\sigma$}

\uput[-30](100,20){$v_1$}

\uput[-30](102,45){$d|A_n|$}

\uput[60](105,55){$d|D_{\mathcal{O}}(u,v)|$}

\psarc{->}(83,55){40}{190}{214}

\cnode*[linecolor=red](-200,0){2pt}{z}
\cnode*[linecolor=red](25,96){2pt}{a}
\cnode*[linecolor=red](83,55){2pt}{y}

\psline(74,49)(81,40)(88,46)

\end{pspicture}

\caption{An illustration for Proposition~\ref{atends}. The angle $\angle uvo_n$ is in range $(-\pi/2,\pi/2)$. As $v$ moves away from $a_{n-1}$ along $\widehat{A}'$, $\angle uvo_n$ decreases. When $v$ moves below $uo_n$, $\angle uvo_n$ becomes negative. }\label{fig:convex}
\end{center}
\end{figure}

%% file: fig-v2-finalsetup.tex
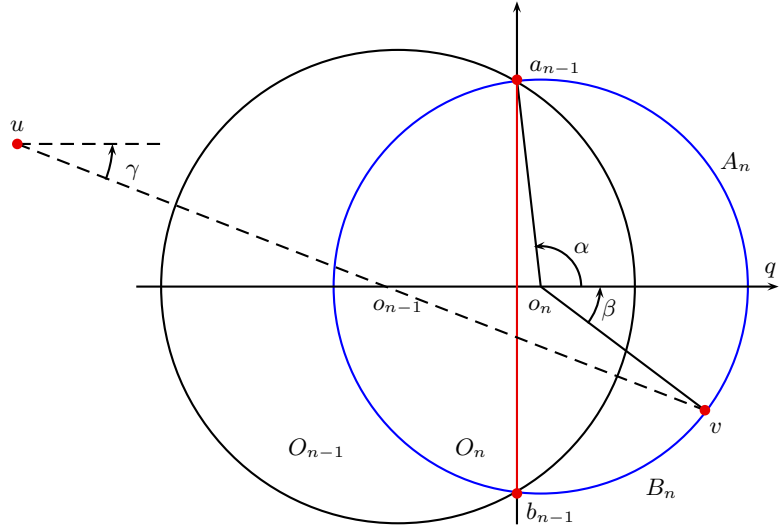
\begin{figure}[tbp]
\psset{unit=0.9pt}
\begin{center} \small
\begin{pspicture}(-160,-130)(160,150)
    \psset{labelsep=5pt}

    \pscircle(0,0){100}

    \pscircle[linecolor=blue](60,0){87.5}


    \psline(60,0)(129,-52)

  \psline{->}(-110,0)(160,0)
  \psline(60,0)(50,87)
    \psline[linecolor=red](50,87)(50,-87)
  \psline[linestyle=dashed](-160,60)(129,-52)
    \psline[linestyle=dashed](-160,60)(-100,60)
\psline{->}(50,87)(50,120)
\psline(50,-87)(50,-100)







  \uput[90](-160,60){$u$}
  \uput[-90](0,0){$o_{n-1}$}

    \uput[-90](60,0){$o_n$}

  \uput[45](150,0){$q$}


\uput[-60](129,-52){$v$}

\uput[30](50,87){$a_{n-1}$}
\uput[-45](50,-87){$b_{n-1}$}

\uput[30](130,47){$A_n$}
\uput[-45](100,-77){$B_n$}

\uput[-45](-50,-60){$O_{n-1}$}
\uput[-45](20,-60){$O_n$}


\uput[45](70,10){$\alpha$}

\uput[0](80,-10){$\beta$}

\uput[0](-120,48){$\gamma$}




\psarc{->}(60,0){17}{0}{100}
\psarc{->}(60,0){25}{-38}{0}
\psarc{->}(-160,60){40}{-22}{0}



\cnode*[linecolor=red](-160,60){2pt}{z}
\cnode*[linecolor=red](50,87){2pt}{a}
\cnode*[linecolor=red](50,-87){2pt}{b}

\cnode*[linecolor=red](129,-52){2pt}{v}

\end{pspicture}

\caption{An illustration of the coordinate system and the definition of the parameters $\alpha,\beta,\gamma$.}\label{fig:finalsetup}
\end{center}
\end{figure}

%% file: fig-v2-finalsetup-1-2.tex
\begin{figure}[tbp]
\psset{unit=0.7pt}
\begin{center} \small
\begin{pspicture}(-160,-150)(160,210)
    \psset{labelsep=5pt}

    \pscircle(0,0){100}

    \pscircle[linecolor=blue](60,0){87.5}


    \psline(60,0)(129,-52)

  \psline{->}(-110,0)(160,0)
  \psline(60,0)(50,-87)
    \psline[linecolor=red](50,87)(50,-87)
  \psline[linestyle=dashed](-160,60)(129,-52)

  \psline[linestyle=dashed,linecolor=blue](-20,210)(129,-52)
  \psline[linestyle=dashed](-20,210)(40,210)
  \psarc{->}(-20,210){20}{-60}{0}
  \uput[-45](0,205){$\gamma'$}

  \psline[linestyle=dashed,linecolor=blue](-180,-52)(129,-52)

    \psline[linestyle=dashed](-160,60)(-100,60)
\psline{->}(50,87)(50,120)
\psline(50,-87)(50,-100)







  \uput[90](-160,60){$u$}


  \uput[45](150,0){$q$}


\uput[-90](129,-52){$v$}

\uput[30](50,87){$a_{n-1}$}
\uput[-45](50,-87){$b_{n-1}$}

\uput[30](130,47){$A_n$}
\uput[-45](100,-77){$B_n$}

\uput[-45](-50,-60){$O_{n-1}$}
\uput[-45](20,-60){$O_n$}


\uput[-70](70,-10){$\alpha$}

\uput[0](80,-10){$\beta$}

\uput[0](-120,50){$\gamma$}




\uput[80](60,0){$o_n$}

\psarc{->}(60,0){17}{-100}{0}
\psarc{->}(60,0){25}{-38}{0}
\psarc{->}(-160,60){40}{-22}{0}



\cnode*[linecolor=red](-160,60){2pt}{z}
\cnode*[linecolor=red](50,87){2pt}{a}
\cnode*[linecolor=red](50,-87){2pt}{b}

\cnode*[linecolor=red](129,-52){2pt}{v}


\end{pspicture}

\caption{An illustration for the range of $\gamma$. This shows that $0 < \gamma < \gamma'$, where $\gamma' = \pi/2-\angle b_{n-1}a_{n-1}v = \pi/2-\frac{\angle b_{n-1}o_nv}{2} = \pi/2- (\alpha-\beta)/2$.}\label{fig:largestgamma}
\end{center}
\end{figure}
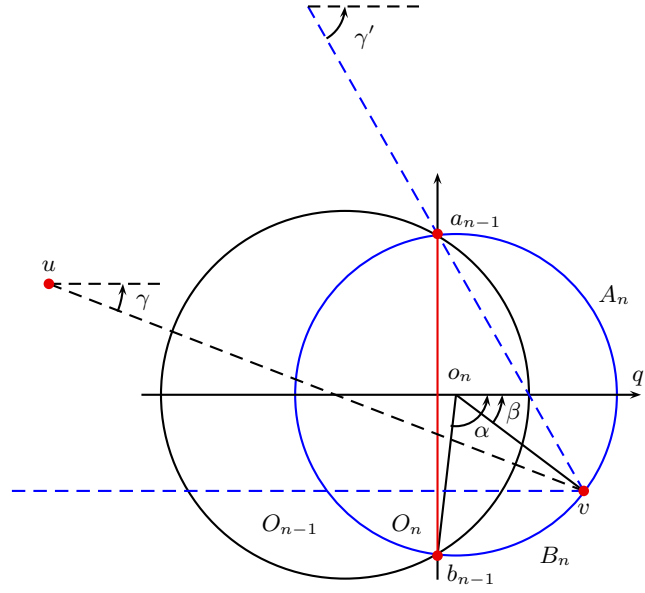

%% file: fig-v2-finalsetup-2.tex
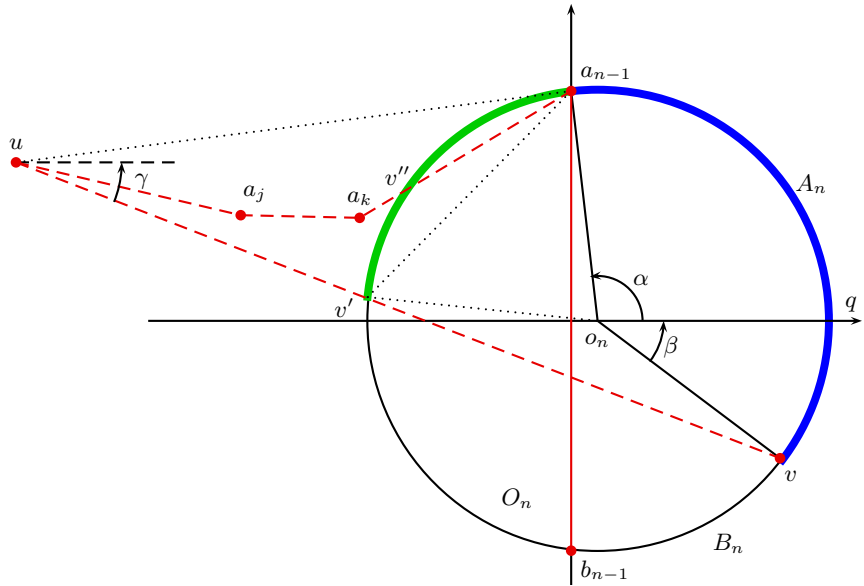
\begin{figure}[tbp]
\psset{unit=1pt}
\begin{center} \small
\begin{pspicture}(-160,-120)(160,120)
    \psset{labelsep=5pt}


    \pscircle(60,0){87.5}

    \psarc[linewidth=3,linecolor=green](60,0){87.5}{96}{175}
    \psarc[linewidth=3,linecolor=blue](60,0){87.5}{-38}{96}


    \psline(60,0)(129,-52)

  \psline{->}(-110,0)(160,0)
  \psline(60,0)(50,87)
    \psline[linecolor=red](50,87)(50,-87)
  \psline[linestyle=dashed,linecolor=red](-160,60)(129,-52)
    \psline[linestyle=dashed](-160,60)(-100,60)
\psline{->}(50,87)(50,120)
\psline(50,-87)(50,-100)

\psline[linestyle=dotted,dotsep=2pt](-160,60)(50,87)

\psline[linestyle=dashed,linecolor=red](-160,60)(-75,40)(-30,39)(50,87)

\cnode*[linecolor=red](-75,40){2pt}{j}
\cnode*[linecolor=red](-30,39){2pt}{k}



\psline[linestyle=dotted,dotsep=2pt](50,87)(-27,9)

\psline[linestyle=dotted,dotsep=2pt](60,0)(-27,9)


\uput[90](-30,39){$a_k$}
\uput[60](-75,40){$a_j$}

  \uput[90](-160,60){$u$}

    \uput[-90](60,0){$o_n$}

  \uput[45](150,0){$q$}

\uput[45](-24,48){$v''$}

\uput[-60](129,-52){$v$}

\uput[45](50,87){$a_{n-1}$}
\uput[-45](50,-87){$b_{n-1}$}

\uput[30](130,47){$A_n$}
\uput[-45](100,-77){$B_n$}

\uput[205](-27,9){$v'$}

\uput[45](70,10){$\alpha$}

\uput[0](80,-10){$\beta$}

\uput[20](-120,50){$\gamma$}


\uput[-45](20,-60){$O_n$}



\psarc{->}(60,0){17}{0}{100}
\psarc{->}(60,0){25}{-38}{0}
\psarc{->}(-160,60){40}{-22}{0}



\cnode*[linecolor=red](-160,60){2pt}{z}
\cnode*[linecolor=red](50,87){2pt}{a}
\cnode*[linecolor=red](50,-87){2pt}{b}

\cnode*[linecolor=red](129,-52){2pt}{v}

\end{pspicture}

\caption{An illustration for Proposition~\ref{upper}.}\label{fig:finalsetup2}
\end{center}
\end{figure}

%% file: fig-v2-finalsetup-2-2.tex
\begin{figure}[tbph]
\psset{unit=0.9pt}
\begin{center} \small
\begin{pspicture}(-160,-120)(160,120)
    \psset{labelsep=5pt}

    \pscircle(0,0){100}

    \pscircle[linecolor=blue](60,0){87.5}

    \pscircle[linecolor=blue,linestyle=dotted](15,0){94}
    \pscircle[linecolor=blue,linestyle=dotted](30,0){90}
    \pscircle[linecolor=blue,linestyle=dotted](45,0){87.5}


  \psline{->}(-110,0)(160,0)
    \psline[linecolor=red](50,87)(50,-87)
  \psline[linestyle=dashed](-160,60)(129,-52)
\psline{->}(50,87)(50,120)
\psline(50,-87)(50,-100)







  \uput[90](-160,60){$u$}
  \uput[90](0,0){$o_{n-1}$}

    \uput[90](60,0){$o_n$}

  \uput[45](150,0){$q$}


\uput[-60](129,-52){$v$}

\uput[30](50,87){$a_{n-1}$}
\uput[-45](50,-87){$b_{n-1}$}

\uput[30](130,47){$A_n$}
\uput[-45](100,-77){$B_n$}












\cnode*[linecolor=red](-160,60){2pt}{z}
\cnode*[linecolor=red](50,87){2pt}{a}
\cnode*[linecolor=red](50,-87){2pt}{b}

\cnode*[linecolor=red](129,-52){2pt}{v}

  \psline[linestyle=dotted](-160,60)(114,-54)
\cnode*[linecolor=red,linestyle=dotted](114,-54){2pt}{v''}

\psline[linestyle=dotted](-160,60)(102,-55)
\cnode*[linecolor=red,linestyle=dotted](102,-55){2pt}{v''}

\psline[linestyle=dotted](-160,60)(91,-56)
\cnode*[linecolor=red,linestyle=dotted](91,-56){2pt}{v''}

\psline[linestyle=dotted](-160,60)(82,-56.5)
\cnode*[linecolor=red,linestyle=dotted](82,-56.5){2pt}{v''}

\uput[-45](-120,-60){$O_{n-1}$}
\uput[-45](20,-60){$O_n$}

\end{pspicture}

\caption{An illustration for the transformation in Proposition~\ref{last}. Transform $O_n$ by moving $o_n$ (the center of $O_n$) toward $o_{n-1}$ (the center of $O_{n-1}$) along the $x$-axis while fixing $a_{n-1}$ and $b_{n-1}$ as the intersections between the boundaries of $O_n$ and $O_{n-1}$. The dotted disks show the process of the transformation. In the transformation, the radius of $O_n$ changes, but $v$ stays pivotal. The large dots show the locations of $v$ during the transformation. }\label{fig:finalsetup22}
\end{center}
\end{figure}
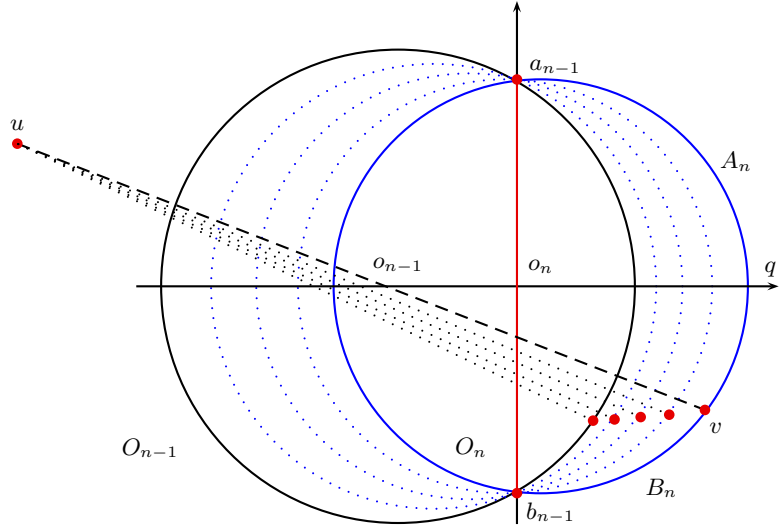

%% file: fig-v2-shrink.tex
\begin{figure}[tbph]
\psset{unit=0.8pt}
\begin{center} \small

\begin{pspicture}(-140,-130)(200,90)
    \psset{labelsep=6pt}

    \psline[linecolor=red](-142,0)(183,0)

    \psarc(-107,4){35}{10}{302}

    \psarc(-27,-31){61}{101}{138}
    \psarc(-27,-31){61}{175}{372}

    \psarc(21,41){61}{7}{192}


    \psarc(90,-10){58}{188}{308}

    \psarc(123,4){60}{-87}{133}

    \psline(-72,10)(-88,-26)
    \psline(-39,28)(32,-19)

    \psline(82,48)(32,-19)
    \psline(82,48)(125,-56)



    \psline[linecolor=red]{->}(2,-84)(-19,-43)


    \psline[linestyle=dashed](-72,10)(-78,-16)
    \psline[linestyle=dashed](-39,28)(-2,-16)

    \psarc[linestyle=dashed](-40,-9){38}{89}{150}
    \psarc[linestyle=dashed](-40,-9){38}{189}{352}

    \psarc[linestyle=dashed](-107,4){35}{302}{328}

    \psarc[linestyle=dashed](21,41){61}{248}{-80}


    \uput[180](-42,37){$A_j$}

    \uput[180](-70,10){$a_1$}
    \uput[-110](-88,-26){$b_1$}
        \uput[20](-80,-20){$b_1'$}

    \uput[60](-39,25){$a_2$}
    \uput[220](80,-10){$b_2$($b_3$)}
        \uput[20](-33,-20){$b_2'$}

    \uput[0](85,45){$a_3$($a_4$)}
    \uput[-90](125,-56){$b_4$}

    \uput[140](-140,0){$a_0$}
    \uput[180](-140,0){$u$}
    \uput[-140](-140,0){$b_0$}

    \uput[40](180,0){$a_5$}
    \uput[0](180,0){$v$}
    \uput[-40](180,0){$b_5$}


\uput[-90](0,-100){(a) Case 1: $A_j$ is not degenerated.}

\end{pspicture}

\begin{pspicture}(-140,-130)(200,140)
    \psset{labelsep=6pt}

    \psline[linecolor=red](-142,0)(183,0)

    \psarc(-107,4){35}{10}{302}

    \psarc(-27,-31){61}{101}{138}
    \psarc(-27,-31){61}{175}{372}

    \psarc(21,41){61}{7}{192}


    \psarc(90,-10){58}{188}{308}

    \psarc(123,4){60}{-87}{133}

    \psline(-72,10)(-88,-26)
    \psline(-39,28)(32,-19)

    \psline(82,48)(32,-19)
    \psline(82,48)(125,-56)

    \psline[linestyle=dashed](82,48)(49,-14)
    \psline[linestyle=dashed](82,48)(82,-36)



    \psline[linecolor=red]{->}(52,-50)(63,-30)




    \psarc[linestyle=dashed](21,41){61}{-80}{-63}

    \psarc[linestyle=dashed](123,4){60}{-136}{-87}
    \psarc[linestyle=dashed](88,5){43}{-154}{-102}


    \uput[180](-70,10){$a_1$}
    \uput[-110](-88,-26){$b_1$}

    \uput[60](-39,25){$a_2$}
    \uput[220](38,-19){$b_2$($b_3$)}
        \uput[20](48,-20){$b_3'$}

    \uput[102](93,50){$A_j$}

    \uput[0](85,45){$a_3$($a_4$)}
    \uput[-90](125,-56){$b_4$}
        \uput[20](77,-38){$b_4'$}

    \uput[140](-140,0){$a_0$}
    \uput[180](-140,0){$u$}
    \uput[-140](-140,0){$b_0$}

    \uput[40](180,0){$a_5$}
    \uput[0](180,0){$v$}
    \uput[-40](180,0){$b_5$}


    \uput[-90](0,-100){(b) Case 2: $A_j$ is degenerated}

\end{pspicture}
\caption{Illustrations of Propositions~\ref{equal}: shrinking a disk. There are two cases depending on whether $A_j$ is degenerated. }\label{fig:shrink}
\end{center}
\end{figure}

%% file: fig-v2-equals.tex
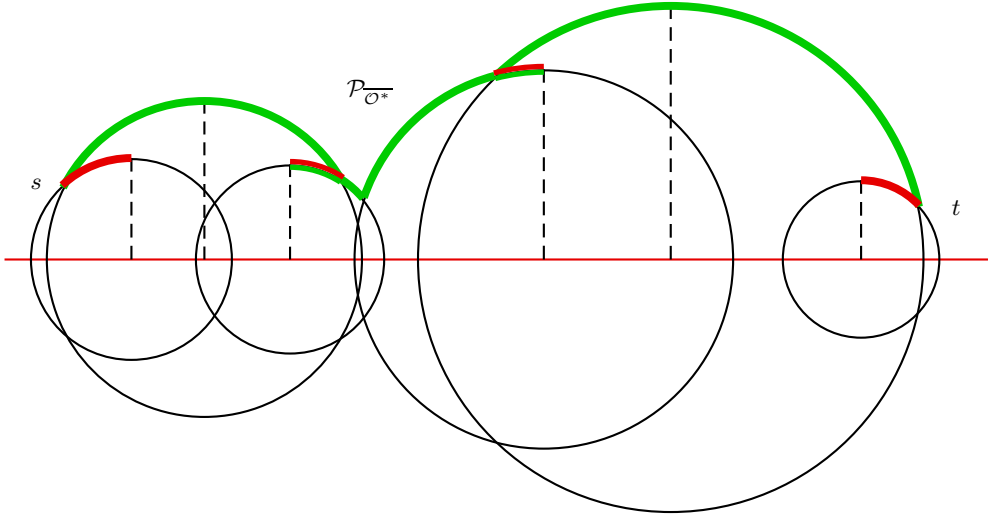
\begin{figure}[tbp]
\psset{unit=1.2pt}
\begin{center} \small

\begin{pspicture}(-180,-80)(160,80)
    \psset{labelsep=5pt}

    \psline[linecolor=red](-170,0)(140,0)

    \pscircle(-130,0){32}
    \pscircle(-107,0){50}
    \pscircle(-80,0){30}
    \pscircle(0,0){60}
    \pscircle(40,0){80}
    \pscircle(100,0){25}

    \psline[linestyle=dashed](-130,0)(-130,32)
    \psline[linestyle=dashed](-107,0)(-107,50)
    \psline[linestyle=dashed](-80,0)(-80,30)
    \psline[linestyle=dashed](0,0)(0,60)
    \psline[linestyle=dashed](40,0)(40,80)
    \psline[linestyle=dashed](100,0)(100,25)

    \psarc[linecolor=green,linewidth=3pt](-107,0){50}{30}{153}
    \psarc[linecolor=red,linewidth=3pt](-130,0){32}{90}{133}

    \psarc[linecolor=green,linewidth=3pt](-80,0){30}{40}{57}

    \psarc[linecolor=green,linewidth=3pt](0,0){60}{105}{161}

    \psarc[linecolor=green,linewidth=3pt](40,0){80}{12}{133}

    \psarc[linecolor=red,linewidth=3pt](100,0){25}{43}{90}

    \psarc[linecolor=green,linewidth=2pt](-80,0){29.1}{57}{90}
    \psarc[linecolor=red,linewidth=2pt](-80,0){30.9}{57}{90}
    \psarc[linecolor=green,linewidth=2pt](0,0){59.1}{90}{105}
    \psarc[linecolor=red,linewidth=2pt](0,0){60.9}{90}{105}

   \uput[90](-160,18){$s$}
   \uput[90](130,10){$t$}

   \uput[90](-55,45){$\mathcal{P}_{\overline{\mathcal{O}^*}}$}

\end{pspicture}

\caption{An illustration for Proposition~\ref{prop:lasatone}. The centers of the disks in $\overline{\mathcal{O}^*}$ are aligned on a straight line. The path $\mathcal{P}_{\overline{\mathcal{O}^*}}$ consists of ``heavy'' (red or dark-gray) arcs and ``light'' (green or light-gray) arcs. If we let overlapping heavy arcs and light arcs in $\mathcal{P}_{\overline{\mathcal{O}^*}}$ cancel each other out and then remove the heavy arcs at the beginning or the end of $\mathcal{P}_{\overline{\mathcal{O}^*}}$, the resulting path is $\mathcal{P}_{\overline{\mathcal{O}^*}}'$, which is a subpath of $A_1\ldots A_n$ between $s$ and $t$.}\label{fig:equals}
\end{center}
\end{figure}

%% file: fig-v2-finalsetup-3.tex
\begin{figure}[tbph]
\psset{unit=1pt}
\begin{center} \small
\begin{pspicture}(-160,-120)(160,120)
    \psset{labelsep=5pt}

\psline[linestyle=dotted](129,0)(129,-52)

    \pscircle(0,0){100}

    \pscircle[linecolor=blue](60,0){87.5}


    \psline(60,0)(129,-52)

  \psline{->}(-110,0)(160,0)
  \psline(60,0)(50,87)
    \psline[linecolor=red](50,87)(50,-87)
  \psline[linestyle=dashed](-160,60)(129,-52)
    \psline[linestyle=dashed](-160,60)(-100,60)
\psline{->}(50,87)(50,120)
\psline(50,-87)(50,-100)







  \uput[90](-160,60){$u$}
  \uput[-90](0,0){$o_{n-1}$}

    \uput[-90](60,0){$o_n$}

  \uput[45](150,0){$q$}


\uput[-60](129,-52){$v$}

\uput[30](50,87){$a_{n-1}$}
\uput[-45](50,-87){$b_{n-1}$}

\uput[30](130,47){$A_n$}
\uput[-45](100,-77){$B_n$}


\uput[45](70,10){$\alpha$}

\uput[0](80,-10){$\beta$}

\uput[0](-120,50){$\gamma$}

\uput[0](145,-20){$\eta$}



\psarc{->}(60,0){17}{0}{100}
\psarc{->}(60,0){25}{-38}{0}
\psarc[linewidth=3pt,linecolor=magenta](60,0){87}{-36}{0}

\psarc{->}(-160,60){40}{-22}{0}



\cnode*[linecolor=red](-160,60){2pt}{z}
\cnode*[linecolor=red](50,87){2pt}{a}
\cnode*[linecolor=red](50,-87){2pt}{b}

\cnode*[linecolor=red](129,-52){2pt}{v}

\uput[-45](-70,-60){$O_{n-1}$}
\uput[-45](20,-60){$O_n$}

\uput[135](50,0){$o$}
\uput[90](129,0){$v_x$}

\end{pspicture}

\caption{An illustration for the parameters defined in Section~\ref{sec:appendix-calc}.}\label{fig:finalsetup3}
\end{center}
\end{figure}
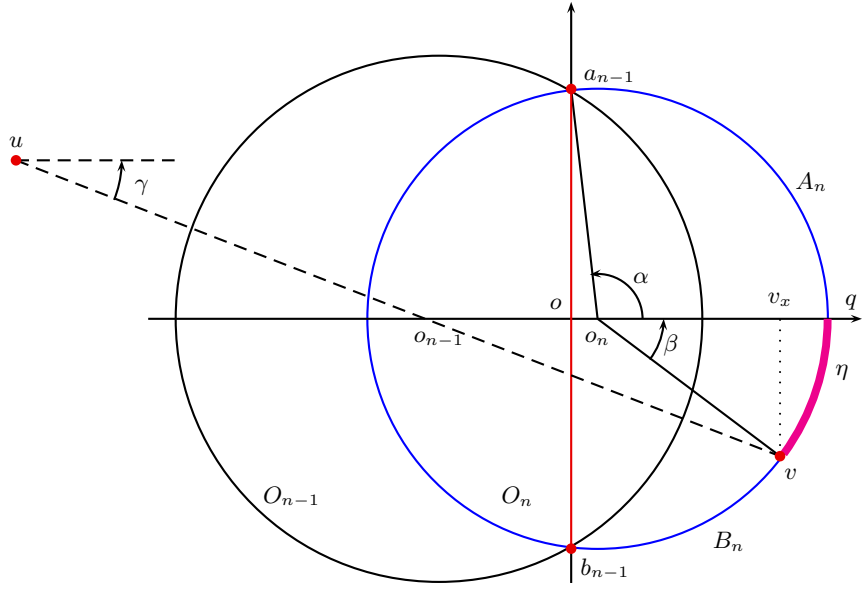

%% file: fig-v2-finalsetup-4.tex
\begin{figure}[tbp]
\psset{unit=1pt}
\begin{center} \small
\begin{pspicture}(-160,-120)(160,120)
    \psset{labelsep=5pt}

    \pscircle(0,0){100}


    \pscircle[linecolor=blue](40,0){87.5}


    \psline(40,0)(109,-52)

  \psline{->}(-110,0)(160,0)
  \psline(40,0)(50,87)
    \psline[linecolor=red](50,87)(50,-87)
  \psline[linestyle=dashed](-160,60)(109,-52)
    \psline[linestyle=dashed](-160,60)(-100,60)
\psline{->}(50,87)(50,120)
\psline(50,-87)(50,-100)





    \psline[linecolor=blue]{->}(109,-52)(160,-40)


  \uput[90](-160,60){$u$}
  \uput[90](0,0){$o_{n-1}$}

    \uput[-90](40,0){$o_n$}

  \uput[45](130,0){$q$}

\uput[45](147,-40){$v_2$}

\uput[45](125,-45){$\partial\ell$}

\uput[-60](109,-52){$v$}

\uput[30](50,87){$a_{n-1}$}
\uput[-45](50,-87){$b_{n-1}$}

\uput[30](110,47){$A_n$}
\uput[-45](80,-77){$B_n$}


\uput[45](50,10){$\alpha$}

\uput[0](60,-10){$\beta$}

\uput[0](-120,50){$\gamma$}




\psarc{->}(40,0){17}{0}{82}
\psarc{->}(40,0){25}{-38}{0}

\psarc{->}(-160,60){40}{-22}{0}

\uput[0](159,-55){$\theta$}
\psline[linestyle=dashed](109,-52)(165,-74)
\psarc{->}(109,-52){52}{-22}{13}

\psline[linestyle=dashed](109,-52)(153,-52)
\psarc{->}(109,-52){40}{0}{13}
\psarc{->}(109,-52){35}{-22}{0}
\uput[0](140,-60){$\gamma$}
\uput[0](145,-46){$\omega$}

\cnode*[linecolor=red](-160,60){2pt}{z}
\cnode*[linecolor=red](50,87){2pt}{a}
\cnode*[linecolor=red](50,-87){2pt}{b}

\cnode*[linecolor=red](109,-52){2pt}{v}

\uput[-45](-70,-60){$O_{n-1}$}
\uput[-45](0,-60){$O_n$}

\end{pspicture}

\caption{An illustration for the definition of $\theta$ and $\omega$.}\label{fig:finalsetup4}
\end{center}
\end{figure}
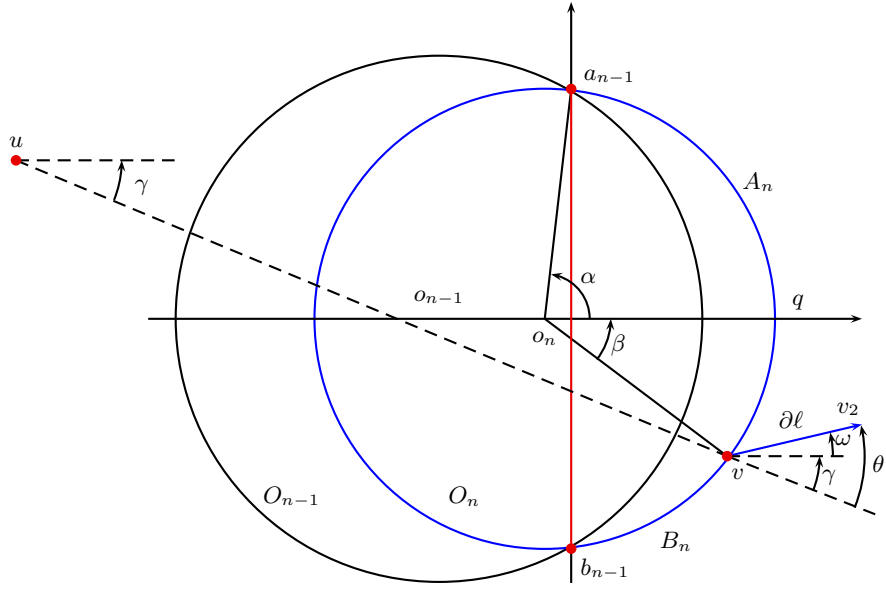